\documentclass[12pt]{amsart}
\usepackage{amsfonts,amssymb,amsmath,euscript}
    \newtheorem{thm}{Theorem}                     [section]
    \newtheorem{thm*}{Theorem}
    \newtheorem{prop}[thm]{Proposition}
    \newtheorem{lemma}[thm]{Lemma}
    \newtheorem{cor}[thm]{Corollary}

    \newtheorem{lemma*}{Lemma}    

    \newtheorem{defn}[thm]{Definition}                 
    \newtheorem{rems}[thm]{Remark}                     
    \newtheorem{rems*}{Remark}   

\newcommand{\ndef}{\newcommand*}
\def\rndef{\renewcommand}

\ndef{\myaddress}[1]{\begin{center} \it\small #1 \end{center}}

\ndef{\clA}{{\mathcal A}} \ndef{\rmA}{{\mathrm A}} \ndef{\mbA}{{\mathbb A}} \ndef{\bfA}{{\mathbf A}} \ndef{\euA}{{\EuScript A}} \ndef{\frA}{{\mathfrak A}}
\ndef{\clB}{{\mathcal B}} \ndef{\rmB}{{\mathrm B}} \ndef{\mbB}{{\mathbb B}} \ndef{\bfB}{{\mathbf B}} \ndef{\euB}{{\EuScript B}} \ndef{\frB}{{\mathfrak B}}
\ndef{\clC}{{\mathcal C}} \ndef{\rmC}{{\mathrm C}} \ndef{\mbC}{{\mathbb C}} \ndef{\bfC}{{\mathbf C}} \ndef{\euC}{{\EuScript C}} \ndef{\frC}{{\mathfrak C}}
\ndef{\clD}{{\mathcal D}} \ndef{\rmD}{{\mathrm D}} \ndef{\mbD}{{\mathbb D}} \ndef{\bfD}{{\mathbf D}} \ndef{\euD}{{\EuScript D}} \ndef{\frD}{{\mathfrak D}}
\ndef{\clE}{{\mathcal E}} \ndef{\rmE}{{\mathrm E}} \ndef{\mbE}{{\mathbb E}} \ndef{\bfE}{{\mathbf E}} \ndef{\euE}{{\EuScript E}} \ndef{\frE}{{\mathfrak E}}
\ndef{\clF}{{\mathcal F}} \ndef{\rmF}{{\mathrm F}} \ndef{\mbF}{{\mathbb F}} \ndef{\bfF}{{\mathbf F}} \ndef{\euF}{{\EuScript F}} \ndef{\frF}{{\mathfrak F}}
\ndef{\clG}{{\mathcal G}} \ndef{\rmG}{{\mathrm G}} \ndef{\mbG}{{\mathbb G}} \ndef{\bfG}{{\mathbf G}} \ndef{\euG}{{\EuScript G}} \ndef{\frG}{{\mathfrak G}}
\ndef{\clH}{{\mathcal H}} \ndef{\rmH}{{\mathrm H}} \ndef{\mbH}{{\mathbb H}} \ndef{\bfH}{{\mathbf H}} \ndef{\euH}{{\EuScript H}} \ndef{\frH}{{\mathfrak H}}
\ndef{\clI}{{\mathcal I}} \ndef{\rmI}{{\mathrm I}} \ndef{\mbI}{{\mathbb I}} \ndef{\bfI}{{\mathbf I}} \ndef{\euI}{{\EuScript I}} \ndef{\frI}{{\mathfrak I}}
\ndef{\clJ}{{\mathcal J}} \ndef{\rmJ}{{\mathrm J}} \ndef{\mbJ}{{\mathbb J}} \ndef{\bfJ}{{\mathbf J}} \ndef{\euJ}{{\EuScript J}} \ndef{\frJ}{{\mathfrak J}}
\ndef{\clK}{{\mathcal K}} \ndef{\rmK}{{\mathrm K}} \ndef{\mbK}{{\mathbb K}} \ndef{\bfK}{{\mathbf K}} \ndef{\euK}{{\EuScript K}} \ndef{\frK}{{\mathfrak K}}
\ndef{\clL}{{\mathcal L}} \ndef{\rmL}{{\mathrm L}} \ndef{\mbL}{{\mathbb L}} \ndef{\bfL}{{\mathbf L}} \ndef{\euL}{{\EuScript L}} \ndef{\frL}{{\mathfrak L}}
\ndef{\clM}{{\mathcal M}} \ndef{\rmM}{{\mathrm M}} \ndef{\mbM}{{\mathbb M}} \ndef{\bfM}{{\mathbf M}} \ndef{\euM}{{\EuScript M}} \ndef{\frM}{{\mathfrak M}}
\ndef{\clN}{{\mathcal N}} \ndef{\rmN}{{\mathrm N}} \ndef{\mbN}{{\mathbb N}} \ndef{\bfN}{{\mathbf N}} \ndef{\euN}{{\EuScript N}} \ndef{\frN}{{\mathfrak N}}
\ndef{\clO}{{\mathcal O}} \ndef{\rmO}{{\mathrm O}} \ndef{\mbO}{{\mathbb O}} \ndef{\bfO}{{\mathbf O}} \ndef{\euO}{{\EuScript O}} \ndef{\frO}{{\mathfrak O}}
\ndef{\clP}{{\mathcal P}} \ndef{\rmP}{{\mathrm P}} \ndef{\mbP}{{\mathbb P}} \ndef{\bfP}{{\mathbf P}} \ndef{\euP}{{\EuScript P}} \ndef{\frP}{{\mathfrak P}}
\ndef{\clQ}{{\mathcal Q}} \ndef{\rmQ}{{\mathrm Q}} \ndef{\mbQ}{{\mathbb Q}} \ndef{\bfQ}{{\mathbf Q}} \ndef{\euQ}{{\EuScript Q}} \ndef{\frQ}{{\mathfrak Q}}
\ndef{\clR}{{\mathcal R}} \ndef{\rmR}{{\mathrm R}} \ndef{\mbR}{{\mathbb R}} \ndef{\bfR}{{\mathbf R}} \ndef{\euR}{{\EuScript R}} \ndef{\frR}{{\mathfrak R}}
\ndef{\clS}{{\mathcal S}} \ndef{\rmS}{{\mathrm S}} \ndef{\mbS}{{\mathbb S}} \ndef{\bfS}{{\mathbf S}} \ndef{\euS}{{\EuScript S}} \ndef{\frS}{{\mathfrak S}}
\ndef{\clT}{{\mathcal T}} \ndef{\rmT}{{\mathrm T}} \ndef{\mbT}{{\mathbb T}} \ndef{\bfT}{{\mathbf T}} \ndef{\euT}{{\EuScript T}} \ndef{\frT}{{\mathfrak T}}
\ndef{\clU}{{\mathcal U}} \ndef{\rmU}{{\mathrm U}} \ndef{\mbU}{{\mathbb U}} \ndef{\bfU}{{\mathbf U}} \ndef{\euU}{{\EuScript U}} \ndef{\frU}{{\mathfrak U}}
\ndef{\clV}{{\mathcal V}} \ndef{\rmV}{{\mathrm V}} \ndef{\mbV}{{\mathbb V}} \ndef{\bfV}{{\mathbf V}} \ndef{\euV}{{\EuScript V}} \ndef{\frV}{{\mathfrak V}}
\ndef{\clW}{{\mathcal W}} \ndef{\rmW}{{\mathrm W}} \ndef{\mbW}{{\mathbb W}} \ndef{\bfW}{{\mathbf W}} \ndef{\euW}{{\EuScript W}} \ndef{\frW}{{\mathfrak W}}
\ndef{\clX}{{\mathcal X}} \ndef{\rmX}{{\mathrm X}} \ndef{\mbX}{{\mathbb X}} \ndef{\bfX}{{\mathbf X}} \ndef{\euX}{{\EuScript X}} \ndef{\frX}{{\mathfrak X}}
\ndef{\clY}{{\mathcal Y}} \ndef{\rmY}{{\mathrm Y}} \ndef{\mbY}{{\mathbb Y}} \ndef{\bfY}{{\mathbf Y}} \ndef{\euY}{{\EuScript Y}} \ndef{\frY}{{\mathfrak Y}}
\ndef{\clZ}{{\mathcal Z}} \ndef{\rmZ}{{\mathrm Z}} \ndef{\mbZ}{{\mathbb Z}} \ndef{\bfZ}{{\mathbf Z}} \ndef{\euZ}{{\EuScript Z}} \ndef{\frZ}{{\mathfrak Z}}

\ndef{\tA}{{\widetilde A}} \ndef{\tcA}{{\widetilde\clA}} \ndef{\ttcA}{\widetilde{\tcA}} \ndef{\sfA}{{\textsf A}} \ndef{\ttA}{\widetilde{\tA}} \ndef{\dzA}{{A^\sharp}}
\ndef{\tB}{{\widetilde B}} \ndef{\tcB}{{\widetilde\clB}} \ndef{\ttcB}{\widetilde{\tcB}} \ndef{\sfB}{{\textsf B}} \ndef{\ttB}{\widetilde{\tB}} \ndef{\dzB}{{B^\sharp}}
\ndef{\tC}{{\widetilde C}} \ndef{\tcC}{{\widetilde\clC}} \ndef{\ttcC}{\widetilde{\tcC}} \ndef{\sfC}{{\textsf C}} \ndef{\ttC}{\widetilde{\tC}} \ndef{\dzC}{{C^\sharp}}
\ndef{\tD}{{\widetilde D}} \ndef{\tcD}{{\widetilde\clD}} \ndef{\ttcD}{\widetilde{\tcD}} \ndef{\sfD}{{\textsf D}} \ndef{\ttD}{\widetilde{\tD}} \ndef{\dzD}{{D^\sharp}}
\ndef{\tE}{{\widetilde E}} \ndef{\tcE}{{\widetilde\clE}} \ndef{\ttcE}{\widetilde{\tcE}} \ndef{\sfE}{{\textsf E}} \ndef{\ttE}{\widetilde{\tE}} \ndef{\dzE}{{E^\sharp}}
\ndef{\tF}{{\widetilde F}} \ndef{\tcF}{{\widetilde\clF}} \ndef{\ttcF}{\widetilde{\tcF}} \ndef{\sfF}{{\textsf F}} \ndef{\ttF}{\widetilde{\tF}} \ndef{\dzF}{{F^\sharp}}
\ndef{\tG}{{\widetilde G}} \ndef{\tcG}{{\widetilde\clG}} \ndef{\ttcG}{\widetilde{\tcG}} \ndef{\sfG}{{\textsf G}} \ndef{\ttG}{\widetilde{\tG}} \ndef{\dzG}{{G^\sharp}}
\ndef{\tH}{{\widetilde H}} \ndef{\tcH}{{\widetilde\clH}} \ndef{\ttcH}{\widetilde{\tcH}} \ndef{\sfH}{{\textsf H}} \ndef{\ttH}{\widetilde{\tH}} \ndef{\dzH}{{H^\sharp}}
\ndef{\tI}{{\widetilde I}} \ndef{\tcI}{{\widetilde\clI}} \ndef{\ttcI}{\widetilde{\tcI}} \ndef{\sfI}{{\textsf I}} \ndef{\ttI}{\widetilde{\tI}} \ndef{\dzI}{{I^\sharp}}
\ndef{\tJ}{{\widetilde J}} \ndef{\tcJ}{{\widetilde\clJ}} \ndef{\ttcJ}{\widetilde{\tcJ}} \ndef{\sfJ}{{\textsf J}} \ndef{\ttJ}{\widetilde{\tJ}} \ndef{\dzJ}{{J^\sharp}}
\ndef{\tK}{{\widetilde K}} \ndef{\tcK}{{\widetilde\clK}} \ndef{\ttcK}{\widetilde{\tcK}} \ndef{\sfK}{{\textsf K}} \ndef{\ttK}{\widetilde{\tK}} \ndef{\dzK}{{K^\sharp}}
\ndef{\tL}{{\widetilde L}} \ndef{\tcL}{{\widetilde\clL}} \ndef{\ttcL}{\widetilde{\tcL}} \ndef{\sfL}{{\textsf L}} \ndef{\ttL}{\widetilde{\tL}} \ndef{\dzL}{{L^\sharp}}
\ndef{\tM}{{\widetilde M}} \ndef{\tcM}{{\widetilde\clM}} \ndef{\ttcM}{\widetilde{\tcM}} \ndef{\sfM}{{\textsf M}} \ndef{\ttM}{\widetilde{\tM}} \ndef{\dzM}{{M^\sharp}}
\ndef{\tN}{{\widetilde N}} \ndef{\tcN}{{\widetilde\clN}} \ndef{\ttcN}{\widetilde{\tcN}} \ndef{\sfN}{{\textsf N}} \ndef{\ttN}{\widetilde{\tN}} \ndef{\dzN}{{N^\sharp}}
\ndef{\tO}{{\widetilde O}} \ndef{\tcO}{{\widetilde\clO}} \ndef{\ttcO}{\widetilde{\tcO}} \ndef{\sfO}{{\textsf O}} \ndef{\ttO}{\widetilde{\tO}} \ndef{\dzO}{{O^\sharp}}
\ndef{\tP}{{\widetilde P}} \ndef{\tcP}{{\widetilde\clP}} \ndef{\ttcP}{\widetilde{\tcP}} \ndef{\sfP}{{\textsf P}} \ndef{\ttP}{\widetilde{\tP}} \ndef{\dzP}{{P^\sharp}}
\ndef{\tQ}{{\widetilde Q}} \ndef{\tcQ}{{\widetilde\clQ}} \ndef{\ttcQ}{\widetilde{\tcQ}} \ndef{\sfQ}{{\textsf Q}} \ndef{\ttQ}{\widetilde{\tQ}} \ndef{\dzQ}{{Q^\sharp}}
\ndef{\tR}{{\widetilde R}} \ndef{\tcR}{{\widetilde\clR}} \ndef{\ttcR}{\widetilde{\tcR}} \ndef{\sfR}{{\textsf R}} \ndef{\ttR}{\widetilde{\tR}} \ndef{\dzR}{{R^\sharp}}
\ndef{\tS}{{\widetilde S}} \ndef{\tcS}{{\widetilde\clS}} \ndef{\ttcS}{\widetilde{\tcS}} \ndef{\sfS}{{\textsf S}} \ndef{\ttS}{\widetilde{\tS}} \ndef{\dzS}{{S^\sharp}}
\ndef{\tT}{{\widetilde T}} \ndef{\tcT}{{\widetilde\clT}} \ndef{\ttcT}{\widetilde{\tcT}} \ndef{\sfT}{{\textsf T}} \ndef{\ttT}{\widetilde{\tT}} \ndef{\dzT}{{T^\sharp}}
\ndef{\tU}{{\widetilde U}} \ndef{\tcU}{{\widetilde\clU}} \ndef{\ttcU}{\widetilde{\tcU}} \ndef{\sfU}{{\textsf U}} \ndef{\ttU}{\widetilde{\tU}} \ndef{\dzU}{{U^\sharp}}
\ndef{\tV}{{\widetilde V}} \ndef{\tcV}{{\widetilde\clV}} \ndef{\ttcV}{\widetilde{\tcV}} \ndef{\sfV}{{\textsf V}} \ndef{\ttV}{\widetilde{\tV}} \ndef{\dzV}{{V^\sharp}}
\ndef{\tW}{{\widetilde W}} \ndef{\tcW}{{\widetilde\clW}} \ndef{\ttcW}{\widetilde{\tcW}} \ndef{\sfW}{{\textsf W}} \ndef{\ttW}{\widetilde{\tW}} \ndef{\dzW}{{W^\sharp}}
\ndef{\tX}{{\widetilde X}} \ndef{\tcX}{{\widetilde\clX}} \ndef{\ttcX}{\widetilde{\tcX}} \ndef{\sfX}{{\textsf X}} \ndef{\ttX}{\widetilde{\tX}} \ndef{\dzX}{{X^\sharp}}
\ndef{\tY}{{\widetilde Y}} \ndef{\tcY}{{\widetilde\clY}} \ndef{\ttcY}{\widetilde{\tcY}} \ndef{\sfY}{{\textsf Y}} \ndef{\ttY}{\widetilde{\tY}} \ndef{\dzY}{{Y^\sharp}}
\ndef{\tZ}{{\widetilde Z}} \ndef{\tcZ}{{\widetilde\clZ}} \ndef{\ttcZ}{\widetilde{\tcZ}} \ndef{\sfZ}{{\textsf Z}} \ndef{\ttZ}{\widetilde{\tZ}} \ndef{\dzZ}{{Z^\sharp}}

\ndef{\bfa}{{\mathbf a}}
\ndef{\bfb}{{\mathbf b}}
\ndef{\bfc}{{\mathbf c}}
\ndef{\bfd}{{\mathbf d}}

\ndef{\euu}{{\EuScript u}}

  \ndef{\eps}{\varepsilon}

\let\geq\geqslant
\let\leq\leqslant

\ndef{\lims}[1]{\lim\limits_{#1}}
\ndef{\sums}[1]{\sum\limits_{#1}}
\ndef{\ints}[1]{\int_{#1}}
\ndef{\sups}[1]{\sup\limits_{#1}}
\ndef{\liminfty}[1]{\lims{#1\to\infty}}
\ndef{\suminf}[1]{\sums{#1=1}^\infty}

\ndef{\limo}[1]{\omega\mbox{-}\!\!\!\lims{#1\to\infty}}          
\ndef{\limL}[1]{\rmL\mbox{-}\!\!\!\lims{#1\to\infty}}            
\ndef{\limLOne}[1]{\clL_1\mbox{-}\!\lims{#1}}
\ndef{\tildelimo}[1]{\tilde\omega\mbox{-}\!\!\!\lims{#1\to\infty}}
\ndef{\slim}{\mathrm{s}\mbox{-}\!\!\lim}          
\ndef{\wlim}{\mathrm{w}\mbox{-}\!\lim}          

\ndef{\Aut}{\operatorname{Aut}}      
\ndef{\Ch}{\operatorname{ch}}        
\ndef{\End}{\operatorname{End}}      
\ndef{\Hom}{\operatorname{Hom}}      
\rndef{\ker}{\operatorname{ker}}      
\ndef{\coker}{\operatorname{coker}}      
\ndef{\im}{\operatorname{im}}        
\ndef{\Log}{\operatorname{Log}}      
\ndef{\OP}{\operatorname{OP}}        
\ndef{\Op}{\operatorname{Op}}        
\ndef{\Symb}{\operatorname{Symb}}    
\ndef{\Tr}{\operatorname{Tr}}        
\ndef{\Wres}{\operatorname{Wres}}    
\ndef{\cl}{\operatorname{cl}}        
\ndef{\com}{\operatorname{com}}
\ndef{\const}{\operatorname{const}}  
\ndef{\conv}{\operatorname{conv}}    
\rndef{\det}{\operatorname{det}}     

\ndef{\detFK}[1]{\Delta\brs{#1}} 
\ndef{\detFKrel}[2]{\Delta_{#2}\brs{#1}} 

\ndef{\adj}{\operatorname{adj}}    
\ndef{\diag}{\operatorname{diag}}    
\ndef{\dist}{\operatorname{dist}}    
\ndef{\dom}{\operatorname{dom}}      
\ndef{\ec}{\operatorname{ec}}        
\ndef{\id}{\mathrm{Id}}                        
\ndef{\ind}{\operatorname{ind}}      
\ndef{\mydeg}{\operatorname{deg}}    
\ndef{\op}{\operatorname{op}}
\ndef{\rank}{\operatorname{rank}}
\ndef{\res}{\operatorname{res}}      
\ndef{\rng}{\operatorname{ran}}      
\ndef{\sflow}{\operatorname{sf}}     
\ndef{\isf}{\operatorname{isf}}      
\ndef{\sign}{\operatorname{sign}}    
\ndef{\sgn}{\operatorname{sgn}}      
\ndef{\sing}{\operatorname{sing}}    
\ndef{\supp}{\operatorname{supp}}    
\ndef{\tr}{\operatorname{tr}}        
\ndef{\var}{\operatorname{var}}      
\ndef{\vol}{\operatorname{vol}}      
\ndef{\wn}{\operatorname{wn}}        
\ndef{\wres}{\operatorname{wres}}    
\rndef{\Im}{\operatorname{Im}}       
\rndef{\Re}{\operatorname{Re}}       

\ndef{\prng}[1]{\mathrm R_{#1}} 
\ndef{\pker}[1]{\mathrm N_{#1}} 
\ndef{\rprng}[2]{\mathrm R_{#1}^{#2}}           
\ndef{\rpker}[2]{\mathrm N_{#1}^{#2}}           
\ndef{\rsupp}[1]{\supp_r(#1)}
\ndef{\lsupp}[1]{\supp_l(#1)}
\ndef{\rslv}[1]{R_z(#1)}      
\ndef{\HH}{H}                 
\ndef{\tHH}{\tilde \HH}       
\ndef{\VV}{V}                 
\ndef{\Rz}{R_z}               
\ndef{\tRz}{\tR_z}            
\ndef{\psif}[1]{#1^{[1]}} 
\ndef{\WPlus}[1]{W_{#1}(\mbR)} 

\ndef{\bndl}{\xi}                         
\ndef{\bndlA}{\eta}                       
\ndef{\GlueMap}{\varphi}                  
\ndef{\ChartMap}{h}                       
\ndef{\chern}{\ensuremath{\mathrm{ch}}}
\ndef{\hilb}{\clH}                     
\ndef{\hilba}{\clH^{(a)}}                    
\ndef{\hilbs}{\clH^{(s)}}                    
   \ndef{\hilbasargument}{(\hilb)} 
\ndef{\LpH}[1]{\clL_{#1}\hilbasargument}       
\ndef{\saLpH}[1]{\clL_{sa}^{#1}\hilbasargument}       
\ndef{\clBH}{\clB\hilbasargument}              
\ndef{\ubBH}{\clB_1\hilbasargument}            
\ndef{\clCH}{\clC\hilbasargument}              
\ndef{\clKH}{\clK\hilbasargument}              
\ndef{\clFH}{\clF\hilbasargument}              
\ndef{\clUH}{\clU\hilbasargument}              
\ndef{\clCFH}{{\clC\clF}\hilbasargument}       
\ndef{\saBH}{\clB_{sa}\hilbasargument}         
\ndef{\saCH}{\clC_{sa}\hilbasargument}         
\ndef{\saFH}{\clF_{sa}\hilbasargument}         
\ndef{\saKH}{\clK_{sa}\hilbasargument}         
\ndef{\saCFH}{\clC\clF_{sa}\hilbasargument}    
\ndef{\clUFH}{\clU\clF\hilbasargument}         
\ndef{\Uinj}{\clU_{inj}\hilbasargument}        
\ndef{\UFinj}{\clU\clF_{inj}\hilbasargument}   

\ndef{\spproj}[2]{E^{#1}_{#2}}                      
\ndef{\spprojb}[2]{E^{#2}_{#1}}                     

\ndef{\LpN}[1]{\clL^{#1}(\clN,\tau)}     
\ndef{\saLpN}[1]{\clL^{#1}_{sa}(\clN,\tau)} 
\ndef{\rLpN}[1]{L^{#1}(\clN,\tau)}       
\ndef{\clAND}{(\clA,\clN,D)}             
\ndef{\clBA}{{\clB(\clA)}}
\ndef{\saKN}{{\clK_{sa}(\clN,\tau)}}          
\ndef{\clKN}{{\clK(\clN,\tau)}}          
\ndef{\clKtN}{{\clK(\tilde\clN,\tau)}}   
\ndef{\clFN}{{\clF(\clN,\tau)}}          
\ndef{\saFN}{{\clF_{sa}(\clN,\tau)}}     
\ndef{\clPN}{\clP(\clN)}                 
\ndef{\clQN}{\clQ(\clN,\tau)}            
\ndef{\infPN}{{\clP_\tau^\infty(\clN)}}  
\ndef{\clOF}[2]{\clF_{#1\mbox{-}#2}(\clN,\tau)}         
\ndef{\oind}[2]{{\rm \tau\mbox{-}ind}_{#1\mbox{-}#2}}   
\ndef{\tind}{\tau\mbox{-}\ind}                  
\ndef{\DInd}{\ind_{\clD,\tau}}           
\ndef{\BF}{Breuer-Fredholm}              
\ndef{\skewfred}[2]{$(#1\cdot #2)$ $\tau$\tire Fredholm}   
\ndef{\affl}{\eta}                       
\ndef{\vNa}{von Neumann algebra}         
\ndef{\nsf}{faithful normal semifinite } 
\ndef{\taubrs}[1]{\tau\brackets{#1}}     
\ndef{\sqbrs}[1]{[#1]}        
\ndef{\Sqbrs}[1]{\big[#1\big]}        
\ndef{\SqBrs}[1]{\Big[#1\Big]}        

\ndef{\domd}{\bigcap\limits_{n\ge 0} \dom\;\delta^n}         
\ndef{\DiffOP}{{\rm \clD}}
\ndef{\ADA}{\clA \cup [\clD,\clA]}
\ndef{\DixIdeal}[1]{\LpH{#1,\infty}}               
\ndef{\dixideal}{\ell^{1,\infty}}                  
\ndef{\WDixIdeal}{\LpH{1,\mathrm w}}               
\ndef{\DixIdealPos}[1]{\DixIdeal{#1}_+}            
\ndef{\DixIdealN}[1]{\LpN{#1,\infty}}              
\ndef{\DixIdealNPar}[2]{\clL^{#1,\infty}_{#2}(\clN,\tau)}    
\ndef{\DixIdealNPos}[1]{\LpN{#1,\infty}_+}                   
\ndef{\TrD}{\Tr_\omega}                                      
\ndef{\tauD}{{\tau_\omega}}                                  
\ndef{\ILogN}{\frac 1{\log(1+N)}}
\ndef{\DixNorm}[1]{\norm{#1}_{(1,\infty)}}                   
\ndef{\DixInt}[1]{\ints 0^t \mu_s(#1)\,ds}
\ndef{\DixIntL}[1]{\ints 0^{\lambda_{1/t}(#1)}\mu_s(#1)\,ds}
    \ndef{\SmallIdeal}{{\clL^{1, \mathrm w}}}
    \ndef{\SmallIdealMeas}{{\clL^{1, \mathrm w}_m}}
    \ndef{\DixIntII}[1]{\int_0^t \mu_s(#1)\,ds}
    \ndef{\DixIntf}[1]{\Phi_t(#1)}
    \ndef{\DixIntg}[1]{\Psi_t(#1)}

\ndef{\lpi}{\clL^{1,\pi}(\clN,\tau)}

\ndef{\strl}[1]{\stackrel \longrightarrow {#1}}
\ndef{\IIinfty}{$\mathrm{II}_\infty$\ }

\ndef{\fourier}[1]{\clF(#1)}          
\ndef{\HaarMeasBohrs}{\nu}            
\ndef{\BrownsMeas}{\mu}               
\ndef{\BohrCont}[1]{\tilde{#1}}       
\ndef{\APMean}{{M}}                   
\ndef{\CDSS}{{\clA_B}}                
\ndef{\matr}{{\rm Mat}}               
\ndef{\seque}[1]{\ensuremath{\{#1_n\}_{n=1}^\infty}}    
\ndef{\sequen}[2]{\ensuremath{\{#1_#2\}_{#2=1}^\infty}}    
\ndef{\Seque}[1]{\ensuremath{\left(#1_0,#1_1,#1_2,\dots\right)}}    
\ndef{\Cesaro}{H}                           
\ndef{\CesaroRPlus}{M}                      
\ndef{\Dilation}{D}                         
\ndef{\Shift}{T}                            

\ndef{\norm}[1]{\left\Vert#1\right\Vert}    
\ndef{\TrNorm}[1]{\norm{#1}_1}              
\ndef{\HSNorm}[1]{\norm{#1}_2}              
\ndef{\InftyNorm}[1]{\norm{#1}_\infty}      
\ndef{\normQN}[1]{\norm{#1}_{\clQN}}        
\ndef{\clLpnorm}[2]{\norm{#2}_{\clL^{#1}}}    
\ndef{\clLnorm}[1]{\clLpnorm{1}{#1}}    

\ndef{\ccurve}{\gamma}                      

\ndef{\abs}[1]{\left\lvert#1\right\rvert}   
\ndef{\set}[1]{\left\{#1\right\}}           
\ndef{\brackets}[1]{\left(#1\right)}        
\ndef{\brs}[1]{\brackets{#1}}               
\ndef{\Brs}[1]{\big(#1\big)}                
\ndef{\BRS}[1]{\Big(#1\Big)}                
\ndef{\scal}[2]{\left\langle #1,#2\right\rangle}               
\ndef{\precprec}{\prec\!\!\!\prec}
\ndef{\qeq}{\stackrel?=}
\ndef{\spectrum}[1]{\sigma_{#1}} 
\ndef{\spectruma}[1]{\sigma^{(a)}_{#1}} 
\ndef{\numrange}[1]{\mathrm{W}(#1)}                         
\rndef{\emptyset}{\varnothing}                              
\ndef{\csupp}{c}                           
\ndef{\closure}[1]{\overline{#1}}
\ndef{\linspan}[1]{\mathrm{span}\ {#1}}
\ndef{\bddborel}[1]{B(#1)}                 
\ndef{\charfunc}{\chi}
\ndef{\FrDer}{\euD}                        
\ndef{\LieDer}[1]{\pounds_{#1}\,}          
\ndef{\dds}{\left.\frac d{ds} \right|_{s = 0}}
\ndef{\ortcmp}[1]{#1^{\scriptscriptstyle \perp}}            
\ndef{\Laplace}{\Delta}                    

\ndef{\matrPQ}[3]
{
    \left(
      \begin{array}{cc}
        #1_{11} & #1_{12} \\
        #1_{21} & #1_{22}
      \end{array}
    \right)_{[#2,#3]}
}

\ndef{\margOK}{\marginpar{\bf \small OK}}

\newcounter{margcomcount}
\ndef{\margcom}[1]{\marginpar{\bf \small #1} \addtocounter{margcomcount}{1}
   \index{\indexcom{{\bf COMMENT: #1}}}}

\newcounter{margproof}
\ndef{\margproof}{\marginpar{\bf \small PROOF} \addtocounter{margproof}{1}
  \index{**** \indexcom{{\bf PROOF}}}}

\newcounter{margdetails}
\ndef{\margdetails}{\marginpar{\bf Details} \addtocounter{margdetails}{1}
  \index{**** \indexcom{{\bf DETAILS}}}}

\newcounter{margproofb}
\ndef{\margproofb}[1]{\marginpar{\bf \small Proof(B) #1} \addtocounter{margproofb}{1}
  \index{**** \indexcom{{\bf PROOF(B): #1}}}}

\newcounter{margdetailsb}
\ndef{\margdetailsb}[1]{\marginpar{\bf \small Details(B)} \addtocounter{margdetailsb}{1}
  \index{**** \indexcom{{\bf DETAILS(B): \\ #1}}}}

\newcounter{margdetailsc}
\ndef{\margdetailsc}[1]{\marginpar{\bf \small Details(C)} \addtocounter{margdetailsc}{1}
  \index{**** \indexcom{{\bf DETAILS(C): \\ #1}}}}

\newcounter{margcomcountb}
\ndef{\margcomb}[1]{\marginpar{\bf \small #1} \addtocounter{margcomcountb}{1}
   \index{\indexcom{{\bf COMMENT(B): \\ #1}}}}

\ndef{\mytimes}{\!\times\!}
\ndef{\sss}[1]{\subsubsection{}\label{#1}}

\rndef{\phi}{\varphi} \ndef{\OpenUnitDisk}{D}
\ndef{\RHS}{RHS}                            
\ndef{\LHS}{LHS} 
\ndef{\ttt}{\Leftrightarrow}
\ndef{\then}{\Rightarrow}
\ndef{\tto}{\longrightarrow}
\ndef{\nno}{\nonumber\\}
\ndef{\newn}[1]{\index{#1} {\bfseries #1}}       
\ndef{\la}{\langle}
\ndef{\ra}{\rangle}
\ndef{\dbar}{{\;\bar{\phantom{o}} \!\!\!\! d}}
\ndef{\stl}[1]{\stackrel{\vbox to 0pt{\vss\hbox{$\scriptstyle #1$}}}}
\ndef{\mathcomment}[1]{{\hfill \qquad\qquad\qquad\text{by (#1)}}}        
\ndef{\mathcomm}[1]{{\hfill \qquad\qquad\qquad\qquad\qquad\text{#1}}}        
\ndef{\details}[1]{\smallskip\begin{center} {\bf Here:}
#1\end{center}\medskip} \ndef{\indexcom}[1]{ --- #1}
\ndef{\longsim}{\ \sim \ }              
\ndef{\tire}{-}              
\ndef{\intinfinf}{\int_{-\infty}^\infty}
\ndef{\refnsftrace}{\cite[V.\,2.\,1]{TakI}} 
\ndef{\refaffloper}{\cite[IV.\,5, Exercise 3]{TakI}} 
\ndef{\refsemifinvNa}{\cite[V.\,1.\,21]{TakI}} 
\ndef{\reftaumeasurable}{\cite[Definition 1.2]{FK86PJM}} 
\ndef{\reftautraceclassaffl}{\cite[V.2, p.\,320]{TakI}} 
\ndef{\refinvoperideal}{\cite[Appendix A.2]{CP2}} 
\ndef{\reftautracenorm}{\cite[V.2, p.\,320]{TakI}} 
\ndef{\reftaucompact}{\cite{}} 
\ndef{\reftauFredholm}{\cite[Appendix B]{PR94JFA}} 

     \ndef{\npartial}{\slash\!\!\!\partial}
     \ndef{\Heis}{\operatorname{Heis}}
     \ndef{\Solv}{\operatorname{Solv}}
     \ndef{\Spin}{\operatorname{Spin}}
     \ndef{\SO}{\operatorname{SO}}
     \ndef{\Index}{\operatorname{index}}

             \ndef{\p}{\partial}
             \ndef{\dd}{|\clD|}
             \ndef{\n}{\parallel}

     \ndef{\gf}[2]{\genfrac{}{}{0pt}{}{#1}{#2}}
     \ndef{\ta}{\widetilde{\alpha}}
     \ndef{\tb}{\widetilde{\beta}}
     \ndef{\txi}{\widetilde{\xi}}
     \ndef{\tk}{\widetilde{K}}
     \ndef{\CGh}{\widetilde{\CG}}
     \ndef{\boe}{{\bf e}}\ndef{\bt}{{\bf t}}
     \ndef{\vth}{\vartheta}
     \ndef{\db}{\overline{\partial}}
     \ndef{\hV}{\hat{V}}
     \ndef{\cag}{{\clA^\Gamma}}
     \ndef{\sind}{\sigma{\rm -ind}}

\newcounter{slidecount}
\catcode`@=11
\newcommand{\slide}[2]{
  \newpage
  \addtocounter{slidecount}{1}
  \renewcommand{\@oddhead}{{\small Advanced Mathematics 1A -- 2009 \hfil Slide #1-\arabic{slidecount}}}
  \begin{center} \bf #2 \end{center}
}
\catcode`@=12
\let\LatexCite=\cite  

\let\ifnumref\iffalse 

\ndef{\ifuncited}[4]{\expandafter\ifx\csname used#4\endcsname\relax}

\ndef{\ifcited}[4]{\expandafter\ifx\csname used#4\endcsname\relax\else}

  \ndef{\papertitle}[1]{ \emph{#1}, }
  \ndef{\paperauthor}[2]{#2}  
  \ndef{\pbbi}[9]{%
      \ifcited{#1}{#2}{#3}{#5}%
        \ifnumref%
          \bibitem{#5}\paperauthor{#1}{#6},\papertitle{#7}#8.%
        \else%
          \advance #9 by 1%
          \ifnum#9<1%
            \bibitem[#4]{#5}\paperauthor{#1}{#6}, \papertitle{#7}#8.%
          \else%
            \bibitem[#4$_{\the#9}$]{#5}\paperauthor{#1}{#6},\papertitle{#7}#8.%
          \fi%
        \fi%
      \fi%
  }
  \ndef{\mbbi}[8]{%
     \ifcited{#1}{#2}{#3}{#5}%
        \ifnumref%
          \bibitem{#5}\paperauthor{#1}{#6},\papertitle{#7}#8.%
        \else%
          \bibitem[#4]{#5}\paperauthor{#1}{#6},\papertitle{#7}#8.%
        \fi%
     \fi%
  }

\ndef{\AddCite}[1]{%
   \ifuncited{0}{0}{0}{#1}%
     \expandafter\gdef\csname used#1\endcsname {}%
   \fi%
}

\def\ProcessCite#1,{%
     \ifx\relax#1%
         \let\next=\relax%
     \else%
         \AddCite{#1}%
         \let\next=\ProcessCite%
     \fi%
     \next%
}

\ndef{\AddCites}[1]{\ProcessCite#1,\relax,}

\ndef{\CiteWithoutExtension}[1]{%
   \AddCites{#1}%
   \LatexCite{#1}%
}

\def\CiteWithExtension[#1]#2{%
   \AddCites{#2}%
   \LatexCite[#1]{#2}%
}

\ndef{\CleverCite}{%
    \ifx\NChar[ %
       \let\MyCite=\CiteWithExtension %
    \else %
       \let\MyCite=\CiteWithoutExtension %
    \fi %
    \MyCite%
}

\renewcommand{\cite}{\futurelet\NChar\CleverCite}

      \ndef{\volume}[1]{{\bf #1}}
      \ndef{\VolYearPP}[3]{\ifnum#2=0 (to appear)\else\volume{#1} (#2), #3\fi}
      \ndef{\VolNoYearPP}[4]{\ifnum#3=0 (to appear)\else\volume{#1} #2 (#3), #4\fi}
      \ndef{\libcode}[1]{}

\ndef{\jnActaMath}[3]{Acta Math. \VolYearPP{#1}{#2}{#3}}                       
\ndef{\jnAdvMath}[3]{Adv. in~Math. \VolYearPP{#1}{#2}{#3}}                     
\ndef{\jnAlgAnal}[3]{Algebra i~Analiz \VolYearPP{#1}{#2}{#3}}
\ndef{\jnAmerJMath}[3]{Amer. J. Math. \VolYearPP{#1}{#2}{#3}}                  
\ndef{\jnAmerMathMonth}[3]{Amer. Math. Monthly \VolYearPP{#1}{#2}{#3}}         
\ndef{\jnAnnMath}[4]{Ann. of~Math. \VolNoYearPP{#1}{#2}{#3}{#4}}               
\ndef{\jnAnalMath}[3]{J. Anal. Math. \VolYearPP{#1}{#2}{#3}}                   
\ndef{\jnArchRatMechAnal}[3]{Arch. Rational Mech. Anal. \VolYearPP{#1}{#2}{#3}}                   
\ndef{\jnBullLondMathSoc}[3]{Bull. London Math. Soc. \VolYearPP{#1}{#2}{#3}}   
\ndef{\jnBullAMS}[3]{Bull. Amer. Math. Soc. \VolYearPP{#1}{#2}{#3}}   
\ndef{\jnCanMathBull}[3]{Canad. Math. Bull. \VolYearPP{#1}{#2}{#3}}            
\ndef{\jnCanMath}[3]{Canad. J.~Math. \VolYearPP{#1}{#2}{#3}}             
\ndef{\jnCommMathPhys}[3]{Comm. Math. Phys. \VolYearPP{#1}{#2}{#3}}             
\ndef{\jnCommPDE}[3]{Comm. Partial Differential Equations \VolYearPP{#1}{#2}{#3}}             
\ndef{\jnComptRendue}[3]{C.\,R.~Acad. Sci. Paris S\'er. A-B \VolYearPP{#1}{#2}{#3}}      
\ndef{\jnContMath}[3]{Contemporary Math. \VolYearPP{#1}{#2}{#3}}               %
\ndef{\jnDukeMJ}[3]{Duke Math. J. \VolYearPP{#1}{#2}{#3}}
\ndef{\jnDiffGeom}[3]{J.~Diff. Geom. \VolYearPP{#1}{#2}{#3}}                   
\ndef{\jnErgodicTheory}[3]{Ergodic Theory and Dynamical Systems \VolYearPP{#1}{#2}{#3}} 
\ndef{\jnFuncAnal}[3]{J.~Functional Analysis \VolYearPP{#1}{#2}{#3}}           
\ndef{\jnFunkAnalPril}[4]{Funct. Anal. Appl. \VolNoYearPP{#1}{#2}{#3}{#4}}  
\ndef{\jnGAFA}[3]{GAFA \VolYearPP{#1}{#2}{#3}}                                 
\ndef{\jnIHES}[3]{IHES Publ. Math. (Paris) \VolYearPP{#1}{#2}{#3}}             
\ndef{\jnIEOT}[3]{Integral Equations Operator Theory   \VolYearPP{#1}{#2}{#3}} 
\ndef{\jnIsrMath}[3]{Israel J.~Math. \VolYearPP{#1}{#2}{#3}}                   
\ndef{\jnKTheory}[3]{K-Theory \VolYearPP{#1}{#2}{#3}}                          
\ndef{\jnLetMathPhys}[3]{Lett. Math. Phys. \VolYearPP{#1}{#2}{#3}}             
\ndef{\jnMathAnn}[3]{Math. Ann. \VolYearPP{#1}{#2}{#3}}                        
\ndef{\jnMathAnalAppl}[3]{J.~Math. Anal. and Appl. \VolYearPP{#1}{#2}{#3}}     
\ndef{\jnMathNachr}[3]{Math. Nachr. \VolYearPP{#1}{#2}{#3}}
\ndef{\jnMathPhys}[3]{J. Math. Phys. \VolYearPP{#1}{#2}{#3}}
\ndef{\jnMathSocJap}[3]{J. Math. Soc. Japan \VolYearPP{#1}{#2}{#3}}
\ndef{\jnOperTheory}[3]{J.~Operator Theory \VolYearPP{#1}{#2}{#3}}             
\ndef{\jnPacJMath}[3]{Pacific J.~Math. \VolYearPP{#1}{#2}{#3}}                  
\ndef{\jnPositivity}[3]{Positivity \VolYearPP{#1}{#2}{#3}}
\ndef{\jnProcAmerMS}[3]{Proc. Amer. Math. Soc. \VolYearPP{#1}{#2}{#3}}         
\ndef{\jnProcCambPhilSoc}[3]{Math. Proc. Camb. Phil. Soc. \VolYearPP{#1}{#2}{#3}}
\ndef{\jnReineAngew}[3]{J.~Reine Angew. Math. \VolYearPP{#1}{#2}{#3}}          
\ndef{\jnTokyoMath}[3]{Tokyo J.~Math. \VolYearPP{#1}{#2}{#3}}
\ndef{\jnTopology}[3]{Topology \VolYearPP{#1}{#2}{#3}}
\ndef{\jnTransAmerMathSoc}[3]{Trans. Amer. Math. Soc. \VolYearPP{#1}{#2}{#3}}
\ndef{\jnIzvANSSSR}[3]{Izv. Akad. Nauk SSSR, Ser. Mat. \VolYearPP{#1}{#2}{#3}}
\ndef{\jnIzvVyshUchZav}[3]{Izv. Vyssh. Uch. Zav., Mat. \VolYearPP{#1}{#2}{#3} (Russian)}
\ndef{\jnIzdatLenUniv}[2]{Izdat. Leningrad. Univ., Leningrad, (#1), #2 (Russian)}
\ndef{\jnFieldsInsComm}[3]{Fields Inst. Comm. \VolYearPP{#1}{#2}{#3}}
\ndef{\jnDoklANSSSR}[3]{Dokl. Akad. Nauk SSSR \VolYearPP{#1}{#2}{#3}}
\ndef{\jnMatZametki}[3]{Matem. zametki \VolYearPP{#1}{#2}{#3}}
\ndef{\jnRussMathSurvey}[3]{Russian Math. Surveys \VolYearPP{#1}{#2}{#3}}
\ndef{\jnSibMathJ}[3]{Sib. Math.~J. \VolYearPP{#1}{#2}{#3}}
\ndef{\jnSovMath}[3]{J.~Soviet math. \VolYearPP{#1}{#2}{#3}}
\ndef{\jnTransMoscMathSoc}[3]{Trans. Moscow Math. Soc. \VolYearPP{#1}{#2}{#3}}
\ndef{\jnUMN}[3]{Uspekhi Mat. Nauk \VolYearPP{#1}{#2}{#3}}

\ndef{\bkTransMathMon}[2]{Trans. Math. Monographs, AMS, \volume{#1}, #2}

\ndef{\pbBirkhauser}[1]{Birkh\"auser, Boston, #1}
\ndef{\pbFactorial}[1]{Moscow, Factorial, #1}
\ndef{\pbGauthier}[1]{Gauthier-Villars, Paris, #1}
\ndef{\pbNauka}[1]{Moscow, Nauka, #1 (Russian)}
\ndef{\pbNaukaR}[1]{Москва, Наука, #1}
\ndef{\pbPrinceton}[1]{Princeton University Press, Princeton, New Jersey, #1}
\ndef{\pbPublPerish}[1]{Publish or Perish Inc., Berkeley, #1}
\ndef{\pbSpringer}[1]{Springer-Verlag, #1}

\ndef{\myauthor}[1]{\mbox{#1}}

\ndef{\Agmon}{\myauthor{Sh.\,Agmon}}
\ndef{\Ahiezer}{\myauthor{N.\,I.\,Ahiezer}}
\ndef{\Arazy}{\myauthor{J.\,Arazy}}
\ndef{\Aronszajn}{\myauthor{N.\,Aronszajn}}
\ndef{\Astashkin}{\myauthor{S.\,V.\,Astashkin}}
\ndef{\Atiyah}{\myauthor{M.\,Atiyah}}
\ndef{\Avron}{\myauthor{J.\,E.\,Avron}}
\ndef{\Azamov}{\myauthor{N.\,A.\,Azamov}}
\ndef{\Banach}{\myauthor{S.\,Banach}}
\ndef{\Benameur}{\myauthor{M-T.\,Benameur}}
\ndef{\Bennett}{\myauthor{C.\,Bennett}}
\ndef{\Berezin}{\myauthor{F.\,A.\,Berezin}}
\ndef{\Berline}{\myauthor{N.\,Berline}}
\ndef{\Birman}{\myauthor{M.\,Sh.\,Birman}}
\ndef{\Blackadar}{\myauthor{B.\,Blackadar}}
\ndef{\Bogolyubov}{\myauthor{N.\,N.\,Bogolyubov}}
\ndef{\Bonsall}{\myauthor{F.\,F.\,Bonsall}}
\ndef{\Bony}{\myauthor{J.\,F.\,Bony}}
\ndef{\BoosBavnbek}{\myauthor{B.\,Boo$\beta$-Bavnbek}}
\ndef{\Bott}{\myauthor{R.\,Bott}}
\ndef{\Branges}{\myauthor{L.\,de Branges}}
\ndef{\Bratteli}{\myauthor{O.\,Bratteli}}
\ndef{\Bredon}{\myauthor{G.\,E.\,Bredon}}
\ndef{\Breuer}{\myauthor{M.\,Breuer}}
\ndef{\Brown}{\myauthor{L.\,G.\,Brown}}
\ndef{\Bruneau}{\myauthor{V.\,Bruneau}}
\ndef{\Buslaev}{\myauthor{V.\,S.\,Buslaev}}
\ndef{\Carey}{\myauthor{A.\,L.\,Carey}}
\ndef{\CareyRW}{\myauthor{R.\,W.\,Carey}} 
\ndef{\Cartan}{\myauthor{H.\,Cartan}}
\ndef{\Chilin}{\myauthor{V.\,I.\,Chilin}}
\ndef{\Coburn}{\myauthor{L.\,A.\,Coburn}}
\ndef{\Connes}{\myauthor{A.\,Connes}}
\ndef{\Cornfeld}{\myauthor{I.\,P.\,Cornfeld}}
\ndef{\Daletskii}{\myauthor{Yu.\,L.\,Daletski\u\i}}   
\ndef{\Dixmier}{\myauthor{J.\,Dixmier}}
\ndef{\DoddsPG}{\myauthor{P.\,G.\,Dodds}}
\ndef{\DoddsTK}{\myauthor{T.\,K.\,Dodds}}
\ndef{\Douglas}{\myauthor{R.\,G.\,Douglas}}
\ndef{\Dubrovin}{\myauthor{B.\,A.\,Dubrovin}}
\ndef{\Dugundji}{\myauthor{J.\,Dugundji}}
\ndef{\Duncan}{\myauthor{J.\,Duncan}}
\ndef{\Dunford}{\myauthor{N.\,Dunford}}
\ndef{\Dykema}{\myauthor{K.\,J.\,Dykema}}
\ndef{\Edwards}{\myauthor{R.\,E.\,Edwards}}
\ndef{\Eilenberg}{\myauthor{S.\,Eilenberg}}
\ndef{\Entina}{\myauthor{S.\,B.\,\`Entina}}
\ndef{\Fack}{\myauthor{T.\,Fack}} 
\ndef{\Faddeev}{\myauthor{L.\,D.\,Faddeev}}
\ndef{\Farber}{\myauthor{M.\,Farber}}
\ndef{\Farforovskaya}{\myauthor{Yu.\,B.\,Farforovskaya}}
\ndef{\Federer}{\myauthor{H.\,Federer}}
\ndef{\Fedosov}{\myauthor{B.\,V.\,Fedosov}}
\ndef{\Figiel}{\myauthor{T.\,Figiel}} 
\ndef{\Figueroa}{\myauthor{H.\,Figueroa}}
\ndef{\Fillmore}{\myauthor{P.\,A.\,Fillmore}}
\ndef{\Fomenko}{\myauthor{A.\,T.\,Fomenko}} 
\ndef{\Fomin}{\myauthor{S.\,V.\,Fomin}}
\ndef{\Frohlich}{\myauthor{J.\,Fr\"ohlich}}
\ndef{\Fuglede}{\myauthor{B.\,Fuglede}}
\ndef{\Furutani}{\myauthor{K.\,Furutani}}
\ndef{\Gelfand}{\myauthor{I.\,M.\,Gelfand}}
\ndef{\Gesztesy}{\myauthor{F.\,Gesztesy}}     
\ndef{\Getzler}{\myauthor{E.\,Getzler}} 
\ndef{\Gilkey}{\myauthor{P.\,B.\,Gilkey}}
\ndef{\Gitler}{\myauthor{S.\,Gitler}}
\ndef{\Glazman}{\myauthor{I.\,M.\,Glazman}}
\ndef{\Glimm}{\myauthor{J.\,Glimm}}
\ndef{\Gohberg}{\myauthor{I.\,C.\,Gohberg}}
\ndef{\Goldshtein}{\myauthor{Ya.\,Goldshtein}}
\ndef{\Golze}{\myauthor{F.\,Golze}}
\ndef{\GraciaBondia}{\myauthor{J.\,M.\,Gracia-Bond\'{i}a}}
\ndef{\Greenleaf}{\myauthor{F.\,P.\,Greenleaf}}
\ndef{\Gromov}{\myauthor{M.\,Gromov}}
\ndef{\Gunning}{\myauthor{R.\,C.\,Gunning}}
\ndef{\Haagerup}{\myauthor{U.\,Haagerup}}
\ndef{\Haag}{\myauthor{R.\,Haag}}
\ndef{\Halmos}{\myauthor{P.\,R.\,Halmos}}
\ndef{\Hardy}{\myauthor{G.\,H.\,Hardy}}
\ndef{\Herbst}{\myauthor{I.\,W.\,Herbst}}
\ndef{\Higson}{\myauthor{N.\,Higson}}  
\ndef{\Hoermander}{\myauthor{L.\,Hoermander}} 
\ndef{\Hoffman}{\myauthor{K.\,Hoffman}} 
\ndef{\Ito}{\myauthor{K.\,Ito}}
\ndef{\Ikebe}{\myauthor{T.\,Ikebe}}
\ndef{\Jaffe}{\myauthor{A.\,Jaffe}}
\ndef{\James}{\myauthor{I.\,M.\,James}}
\ndef{\Javrjan}{\myauthor{V.\,A.\,Javrjan}}
\ndef{\Jitomirskaya}{\myauthor{S.\,Jitomirskaya}}
\ndef{\Kadison}{\myauthor{R.\,V.\,Kadison}}
\ndef{\Kalton}{\myauthor{N.\,J.\,Kalton}} 
\ndef{\Kato}{\myauthor{T.\,Kato}} 
\ndef{\Kobayashi}{\myauthor{S.\,Kobayashi}}
\ndef{\Koplienko}{\myauthor{L.\,S.\,Koplienko}}
\ndef{\Korotyaev}{\myauthor{E.\,Korotyaev}}
\ndef{\Kosaki}{\myauthor{H.\,Kosaki}}
\ndef{\Kostrykin}{\myauthor{V.\,Kostrykin}}
\ndef{\Kotani}{\myauthor{S.\,Kotani}}
\ndef{\Krein}{\myauthor{Kre\u\i n}}
\ndef{\KreinMG}{\myauthor{M.\,G.\,Kre\u\i n}}
\ndef{\KreinSG}{\myauthor{S.\,G.\,Kre\u\i n}}
\ndef{\Kuroda}{\myauthor{S.\,T.\,Kuroda}}
\ndef{\Leichtnam}{\myauthor{E.\,Leichtnam}}
\ndef{\Lesch}{\myauthor{M.\,Lesch}}
\ndef{\Lesniewski}{\myauthor{A.\,Lesniewski}}
\ndef{\Levitan}{\myauthor{B.\,M.\,Levitan}}
\ndef{\Lidskii}{\myauthor{V.\,B.\,Lidskii}}
\ndef{\Lifshits}{\myauthor{I.\,M.\,Lifshits}}
\ndef{\Lindenstrauss}{\myauthor{J.\,Lindenstrauss}}
\ndef{\Loday}{\myauthor{J.-L.\,Loday}}
\ndef{\Lord}{\myauthor{S.\,Lord}}      
\ndef{\Lorentz}{\myauthor{G.\,Lorentz}}
\ndef{\Magnus}{\myauthor{W.\,Magnus}}
\ndef{\Makarov}{\myauthor{K.\,A.\,Makarov}}
\ndef{\MakarovN}{\myauthor{N.\,Makarov}}
\ndef{\Mathai}{\myauthor{V.\,Mathai}}         
\ndef{\McKean}{\myauthor{H.\,P.\,McKean}}
\ndef{\Mishchenko}{\myauthor{A.\,S.\,Mishchenko}}
\ndef{\Molchanov}{\myauthor{S.\,A.\,Molchanov}}
\ndef{\Moore}{\myauthor{C.\,C.\,Moore}}
\ndef{\Moscovici}{\myauthor{H.\,Moscovici}}  
\ndef{\Motovilov}{\myauthor{A.\,K.\,Motovilov}}
\ndef{\Moyer}{\myauthor{R.\,D.\,Moyer}}
\ndef{\Naboko}{\myauthor{S.\,N.\,Naboko}}
\ndef{\Narasimhan}{\myauthor{R.\,Narasimhan}}
\ndef{\Nomizu}{\myauthor{K.\,Nomizu}}
\ndef{\Novikov}{\myauthor{S.\,P.\,Novikov}}
\ndef{\Osterwalder}{\myauthor{K.\,Osterwalder}}
\ndef{\Patodi}{\myauthor{V.\,Patodi}}
\ndef{\Pagter}{\myauthor{B.\,de~Pagter}}  
\ndef{\Pastur}{\myauthor{L.\,A.\,Pastur}}  
\ndef{\Pavlov}{\myauthor{B.\,S.\,Pavlov}}
\ndef{\Pedersen}{\myauthor{G.\,K.\,Pedersen}}
\ndef{\Peller}{\myauthor{V.\,V.\,Peller}}
\ndef{\Perera}{\myauthor{V.\,S.\,Perera}}
\ndef{\Petunin}{\myauthor{Ju.\,I.\,Petunin}}
\ndef{\Phillips}{\myauthor{J.\,Phillips}}  
\ndef{\Piazza}{\myauthor{P.\,Piazza}}   
\ndef{\Pincus}{\myauthor{J.\,D.\,Pincus}}   
\ndef{\Poincare}{Poincar\'e}
\ndef{\Postnikov}{\myauthor{M.\,M.\,Postnikov}} 
\ndef{\Povzner}{\myauthor{A.\,Ya.\,Povzner}}
\ndef{\Prinzis}{\myauthor{R.\,Prinzis}}
\ndef{\Privalov}{\myauthor{I.\,I.\,Privalov}}
\ndef{\Pushnitski}{\myauthor{A.\,B.\,Pushnitski}} 
\ndef{\Raeburn}{\myauthor{I.\,Raeburn}}
\ndef{\Raikov}{\myauthor{G.\,Raikov}}
\ndef{\Reed}{\myauthor{M.\,Reed}}
\ndef{\Rennie}{\myauthor{A.\,Rennie}}
\ndef{\Rickart}{\myauthor{C.\,E.\,Rickart}}
\ndef{\Riesz}{\myauthor{F.\,Riesz}}
\ndef{\Ringrose}{\myauthor{J.\,Ringrose}}
\ndef{\Rio}{\myauthor{R.\,del Rio}}
\ndef{\Robinson}{\myauthor{D.\,Robinson}}
\ndef{\Rossi}{\myauthor{H.\,Rossi}}
\ndef{\Rudin}{\myauthor{W.\,Rudin}}
\ndef{\Ruelle}{\myauthor{D.\,Ruelle}}
\ndef{\Ruzhansky}{\myauthor{M.\,Ruzhansky}}
\ndef{\Sakai}{\myauthor{Sh.\,Sakai}}
\ndef{\Sargsjan}{\myauthor{I.\,S.\,Sargsjan}}
\ndef{\Sato}{\myauthor{H.\,Sato}}
\ndef{\Schaeffer}{\myauthor{D.\,G.\,Schaeffer}}
\ndef{\Schluchtermann}{\myauthor{G.\,Schluchtermann}}
\ndef{\Schochet}{\myauthor{C.\,Schochet}}
\ndef{\SchroedingerE}{\myauthor{E.\,Schr\"odinger}}
\ndef{\Schroedinger}{\myauthor{Schr\"odinger}}
\ndef{\Schrohe}{\myauthor{E.\,Schrohe}}
\ndef{\Schwartz}{\myauthor{J.\,T.\,Schwartz}}
\ndef{\Sedaev}{\myauthor{A.\,A.\,Sedaev}}
\ndef{\Seiler}{\myauthor{R.\,Seiler}}
\ndef{\Semenov}{\myauthor{E.\,M.\,Semenov}}
\ndef{\Shabat}{\myauthor{B.\,V.\,Shabat}}
\ndef{\Shafarevich}{\myauthor{I.\,R.\,Shafarevich}}
\ndef{\Sharpley}{\myauthor{R.\,Sharpley}}
\ndef{\Shilov}{\myauthor{G.\,E.\,Shilov}}
\ndef{\Shirkov}{\myauthor{D.\,V.\,Shirkov}}
\ndef{\Shubin}{\myauthor{M.\,A.\,Shubin}}
\ndef{\Silverman}{\myauthor{H.\,Silverman}}
\ndef{\Simon}{\myauthor{B.\,Simon}}
\ndef{\Sinai}{\myauthor{Ya.\,G.\,Sinai}}
\ndef{\Singer}{\myauthor{I.\,M.\,Singer}}
\ndef{\Solomyak}{\myauthor{M.\,Z.\,Solomyak}}
\ndef{\Soloviev}{\myauthor{Yu.\,P.\,Soloviev}}
\ndef{\Spivak}{\myauthor{M.\,Spivak}}
\ndef{\Stein}{\myauthor{E.\,M.\,Stein}}
\ndef{\Stenkin}{\myauthor{V.\,V.\,Sten'kin}}
\ndef{\Stratila}{\myauthor{S.\,Stratila}}
\ndef{\Sucheston}{\myauthor{L.\,Sucheston}}
\ndef{\Sukochev}{\myauthor{F.\,A.\,Sukochev}}
\ndef{\Switzer}{\myauthor{R.\,M.\,Switzer}}
\ndef{\SzNagy}{\myauthor{B.\,Sz.-Nagy}}
\ndef{\Takesaki}{\myauthor{M.\,Takesaki}}
\ndef{\Taylor}{\myauthor{M.\,E.\,Taylor}}
\ndef{\Treves}{\myauthor{F.\,Treves}}
\ndef{\Troitsky}{\myauthor{E.\,V.\,Troitsky}}
\ndef{\Tzafriri}{\myauthor{L.\,Tzafriri}}
\ndef{\Varilly}{\myauthor{J.\,C.\,V\'{a}rilly}}
\ndef{\Vergne}{\myauthor{M.\,Vergne}}
\ndef{\Vladimirov}{\myauthor{V.\,S.\,Vladimirov}}
\ndef{\Voiculescu}{\myauthor{D.\,Voiculescu}}
\ndef{\Weiss}{\myauthor{G.\,Weiss}}
\ndef{\Wells}{\myauthor{R.\,O.\,Wells}}
\ndef{\Williams}{\myauthor{J.\,P.\,Williams}}
\ndef{\Winkler}{\myauthor{S.\,Winkler}}
\ndef{\Witten}{\myauthor{E.\,Witten}}
\ndef{\Wodzicki}{\myauthor{M.\,Wodzicki}}
\ndef{\Wojciechowski}{\myauthor{K.\,P.\,Wojciechowski}}
\ndef{\Yafaev}{\myauthor{D.\,R.\,Yafaev}}
\ndef{\Yosida}{\myauthor{K.\,Yosida}}
\ndef{\Zsido}{\myauthor{L.\,Zsido}}

\ndef{\hlambda}{{\mathfrak h_\lambda}}

\begin{document}
\title[A constructive approach]{A constructive approach to stationary \\ scattering theory}
\author{\Azamov}
\address{School of Computer Science, Engineering and Mathematics
   \\ Flinders University
   \\ Bedford Park, 5042, SA Australia.}
\email{nurulla.azamov@flinders.edu.au}
\keywords{Wave matrix, scattering matrix, wave operator, scattering operator, sheaf of Hilbert spaces, rigged Hilbert space, Schr\"odinger operator}

\subjclass[2000]{ 
    47A55. 
}
\begin{abstract} In this paper we give a new and constructive approach to stationary
scattering theory for pairs of self-adjoint operators~$H_0$ and
$H_1$ on a Hilbert space~$\hilb$ which satisfy the following
conditions: (i) for any open bounded subset
$\Delta$ of $\mbR,$ the operators $F E_\Delta^{H_0}$ and $F
E_\Delta^{H_1}$ are Hilbert-Schmidt and
(ii) $V = H_1- H_0$ is bounded and admits decomposition
$V = F^*JF,$ where~$F$ is a bounded operator with trivial kernel
from~$\hilb$ to another Hilbert space~$\clK$ and $J$ is a bounded
self-adjoint operator on $\clK.$ An example of a pair of
operators which satisfy these conditions is the Schr\"odinger
operator $H_0 = -\Delta + V_0$ acting on $L^2(\mbR^\nu),$ where
$V_0$ is a potential of class $K_\nu$ (see B.\,Simon, {\it Schr\"odinger semigroups,} Bull. AMS \VolYearPP{7}{1982}{447--526}) and $H_1 = H_0
+ V_1,$ where $V_1 \in L^\infty(\mbR^\nu) \cap L^1(\mbR^\nu).$
Among results of this paper is a new proof of existence and completeness of wave operators $W_\pm(H_1,H_0)$
and a new constructive proof of stationary formula for the scattering matrix.
This approach to scattering theory is based on explicit diagonalization of a self-adjoint operator~$H$
on a sheaf of Hilbert spaces $\euS(H,F)$ associated with the pair $(H,F)$ and with subsequent construction and study of properties of wave matrices $w_\pm(\lambda; H_1,H_0)$
acting between fibers $\hlambda(H_0,F)$ and $\hlambda(H_1,F)$ of sheaves $\euS(H_0,F)$ and $\euS(H_1,F)$ respectively.
The wave operators  $W_\pm(H_1,H_0)$ are then defined as direct integrals of wave matrices and are proved to coincide with classical
time-dependent definition of wave operators.
\end{abstract}
\maketitle

\tableofcontents
\section{Introduction and preliminaries}
In this paper we develop a new approach to stationary scattering theory of abstract Schr\"odinger equation $\psi_t = -i H \psi.$
It is known that in some cases wave operators (see e.g. \cite{BW,RS3,TayST,Ya})
\begin{equation} \label{F: def of W(+-)}
  W_\pm(H_1,H_0) = \mbox{s}-\lim_{t \to \pm \infty} e^{itH_1}e^{-itH_0}P^{(a)}(H_0)
\end{equation}
of two self-adjoint operators~$H_0$ and~$H_1$ can be defined in terms of eigenfunction expansions of absolutely continuous parts of both the initial~$H_0$
and perturbed~$H_1$ operators, see e.g. formula (83) in \cite{RS3}. This relation between wave operators and eigenfunction expansions was used for instance by \Povzner\ \cite{Povz53,Povz55}
and \Ikebe\ \cite{Ikebe60} (see also e.g. \cite{Thoe67,ASch71,Agm,KuJMSJ73I,KuJMSJ73II,Kur} and \cite[\S XI.6]{RS3}) in the case of Schr\"odinger operators $-\Delta + V$ on $L^2(\mbR^\nu).$
The approach of this paper to scattering theory as well as that of \cite{Az3v6} combines the auxiliary space method (see for example \cite{Agm}, \cite[\S XI.6 Appendix]{RS3})
and definition of the wave matrix via eigenfunction expansions. In this regard, this method has something in common with potential scattering theory.
On the other hand, the premise of the approach presented in this paper is of abstract trace-class type and as such it belongs to the trace-class method of Birman-Kato theory,
as opposed to the smooth method of Kato \cite{Kato66,Kato68} which originated as an abstraction of potential scattering theory methods (see also \cite[Chapter 4]{Ya}).

Apart from wave operators~(\ref{F: def of W(+-)}), in scattering theory there are three other objects of prime importance,
the scattering operator
$$
  \bfS(H_1,H_0) = W^*_+(H_1,H_0)W_-(H_1,H_0),
$$
the wave matrices $w_\pm(\lambda; H_1,H_0)$ and the scattering matrix $S(\lambda; H_1,H_0),$ which are related by formulas
\begin{equation} \label{F: bfS and W(pm) = int}
  \bfS(H_1,H_0) = \int^\oplus_{\hat \sigma(H_0)} S(\lambda; H_1,H_0)\,d\rho(\lambda), \quad W_\pm(H_1,H_0) = \int^\oplus_{\hat \sigma(H_0)} w_\pm(\lambda; H_1,H_0)\,d\rho(\lambda),
\end{equation}
where the wave matrices $w_\pm(\lambda; H_1,H_0)$ are unitary operators $\hlambda(H_0) \to \hlambda(H_1)$ and the scattering matrix is a unitary operator
$\hlambda(H_0) \to \hlambda(H_0)$ for a.e. $\lambda \in \hat \sigma(H_0),$ $\hat \sigma(H_j)$ is a core of spectrum of $H_j$ and $\hlambda(H_j), j=0,1,$ are fiber Hilbert spaces from the direct integrals
\begin{equation} \label{F: euF(j)}
  \euF_j \colon \hilb^{(a)}(H_j) \cong  \int^\oplus _{\hat \sigma(H_0)} \hlambda(H_j)\,d\rho(\lambda),
\end{equation}
which diagonalize absolutely continuous parts of the operators $H_j, j=0,1,$ see for instance~\cite{Ya}.
For the scattering matrix $S(\lambda; H_1,H_0)$ there exist explicit stationary formulas, which are important in physics.
An advantage of the new approach to scattering theory of relatively trace-class perturbations considered in this paper is that
in this theory we first construct the operators $w_\pm(\lambda; H_1,H_0)$ on a pre-defined and explicitly described core of the absolutely continuous spectrum.
This circumstance has its advantages, since, for instance, in some questions it is necessary to consider a family of scattering matrices $\set{w_\pm(\lambda; H_r,H_0) \colon r \in [a,b]}$
for a.e. $\lambda \in \mbR.$ This is impossible unless at the very least for each $r \in [a,b]$ we know an explicit description (that is, not up to a null set of uncertain nature) of the set of values of $\lambda$
for which the wave matrices $w_\pm(\lambda; H_r,H_0)$ exist. In potential scattering theory for Schr\"odinger operators such an explicit description of the core exists, namely it is
the set $(0,\infty) \setminus e_+(H),$ where $e_+(H)$ is the discrete set of eigenvalues of the Schr\"odinger operator $H=-\Delta+V,$ see e.g. \cite{Agm}.
In conventional trace-class method such a description does not exist, since this method relies on abstract spectral theorem for arbitrary self-adjoint operators.
None of the numerous versions of the spectral theorem (see for instance \cite{RS1}) gives an explicit description of a core of spectrum
which can be used for this purpose; moreover it is known to be generally impossible.
Unlike the usual trace-class method in scattering theory, our approach addresses this issue by introducing an additional structure into the Hilbert space~$\hilb,$
on which operators~$H_0$ and~$H_1$ act, in the form of a fixed bounded operator~$F$ with trivial kernel from~$\hilb$ to possibly another Hilbert space $\clK.$
This additional structure which we call rigging generates
the auxiliary Hilbert spaces~$\hilb_{\pm 1}$ which serve as analogous of the weighted $L^2$-spaces $L^{2,s}(\mbR^\nu)$ used in potential scattering (see e.g. \cite{Agm},\cite[\S XI.6, Appendix]{RS3})
and, more importantly, it allows to construct an explicit diagonalization of absolutely continuous parts of both the initial~$H_0$ and perturbed~$H_1$ self-adjoint operators.
It may be worth mentioning that the diagonalization process works equally well for any self-adjoint operator which is compatible in a certain sense with the rigging~$F,$
unlike the situation in the potential scattering theory where the initial operator $H_0 = -\Delta$ is trivially diagonalized by the Fourier transform, while perturbed operator $-\Delta+V$ is very hard
to diagonalize. Once both operators~$H_0$ and~$H_1$ are diagonalized, the wave matrices can be defined and their properties studied. The core of spectrum for \emph{all} elements $\lambda$ of which the wave matrices are constructed
allows an explicit description, --- namely, it is the set of full Lebesgue measure for which abstract limiting absorption principle holds for both~$H_0$ and $H_1.$
Finally, we note that the new approach allows to avoid many technical difficulties of the trace-class method approach to stationary scattering theory, as it is given for example in \cite{Ya},
and therefore it is essentially simpler; moreover, some tools used in this approach such as an explicit diagonalization of a self-adjoint operator are of interest on their own.
Other objects of scattering theory such as wave operators, the scattering matrix and the scattering operator are defined via the wave matrices.
Trace-class version of this approach has found applications to the theory of spectral shift function, cf. \cite{Az3v6}.

\subsection{Hilbert-Schmidt rigging}
\label{SS: HS rigging}
In case of Hilbert-Schmidt rigging~$F$ this program was carried out in \cite{Az3v6}. Since the perturbation operator $V=H_1-H_0$ is assumed to admit decomposition $V=F^*JF$
with bounded operator $J$ on $\clK,$ the setting of \cite{Az3v6} covers only trace-class perturbations $V.$ In this paper we show that the condition ``$F$ is Hilbert-Schmidt''
can be replaced by the condition ``for any bounded Borel sets~$\Delta$ of real numbers the operators $FE^{H_0}_\Delta, FE^{H_1}_\Delta$ are Hilbert-Schmidt''.
Below we give a brief description of results of \cite{Az3v6}, both for the purpose of introduction and as a preliminary material used in this paper.

Assume that we are given a Hilbert space~$\hilb$ rigged with a Hilbert-Schmidt operator~$F$ with trivial kernel and dense range. Let
$$
  F = \sum_{j=1}^\infty \kappa_j \scal{\phi_j}{\cdot}\psi_j, \ \ F \colon \hilb \to \clK,
$$
be a fixed Schmidt representation of~$F,$ where $(\kappa_j)$ are $s$-numbers of~$F,$ $(\phi_j)$ is an orthonormal basis of~$\hilb,$
$(\psi_j)$ is an orthonormal basis of~$\clK$ (see e.g. \cite{RS1}). For any self-adjoint operator~$H$ which acts on such a rigged Hilbert space $(\hilb,F)$
we introduce the following objects which depend only on the pair $(H,F).$

I. A \underline{set $\Lambda(H,F)$} of real numbers, defined as follows: $\lambda \in \Lambda(H,F)$ if and only if (i) the trace-class operator
$$
  T_{\lambda + iy}(H) := F R_{\lambda + iy}(H) F^* := F \brs{H-\lambda - iy}^{-1} F^*
$$
has a limit, denoted by $T_{\lambda + i0}(H),$ in uniform operator norm as $y \to 0^+$
and (ii) the imaginary part
$$
  \Im T_{\lambda + iy}(H) = \frac {1}{2i} (T_{\lambda + iy}(H) - T_{\lambda - iy}(H))
$$
of $T_{\lambda + iy}(H)$ has a limit, denoted by $\Im T_{\lambda + i0}(H),$ in trace-class norm.

\begin{thm} \label{T: Abstract L.A.P.} (The abstract limiting absorption principle \cite{BE}, see also \cite[\S 6.1]{Ya})
\ The set $\Lambda(H,F)$ has full Lebesgue measure, that is, the Lebesgue measure of the set $\mbR \setminus \Lambda(H,F)$ is zero.
\end{thm}
\noindent In fact, the abstract limiting absorption principle asserts that the limit $T_{\lambda + i0}(H)$ exists for a.e. $\lambda$ in stronger Hilbert-Schmidt norm,
but we shall not need this. We prefer to use the qualifier ``abstract'' in the name of this limiting absorption principle to distinguish it from
the limiting absorption principle for differential operators, see e.g. \cite{Agm,Kur}. The set of full Lebesgue measure $\Lambda(H,F)$
will play the role of the core of spectrum $\hat \sigma(H)$ mentioned previously. Of course, the set $\Lambda(H,F)$ can be much larger than the spectrum of $H,$
but what is important for us is the following property of  $\Lambda(H,F).$

\begin{thm} \label{T: H E(H) is a.c.} The operator $H E^H_{\Lambda(H,F)}$ is absolutely continuous.
\end{thm}
\noindent
This theorem asserts that the set $\Lambda(H,F)$ cuts away from the set of all real numbers the singular part of the spectrum of $H.$
In particular, the null set $\mbR \setminus \Lambda(H,F)$ contains all eigenvalues of $H.$
Proof can be found in e.g. \cite{Ya}, see also \cite[Corollary 2.5.3]{Az3v6}.

II. Now for each $\lambda \in \Lambda(H,F)$ we introduce the \underline{fiber Hilbert space $\hlambda(H,F)$}
as a certain (closed linear) subspace of the Hilbert space $\ell_2.$
For any $\lambda \in \Lambda(H,F)$ and $y > 0$ let (see \cite[(2.10)]{Az3v6})
\begin{equation} \label{F: phi(l+iy)}
  \phi(\lambda+iy) = \frac 1\pi \brs{\scal{\psi_i}{\Im T_{\lambda + iy}(H) \psi_j}}_{i,j=1}^\infty;
\end{equation}
this is the matrix of the operator $\frac 1\pi\Im T_{\lambda + iy}(H)$ in the basis $(\psi_j)$ of $\clK,$ and therefore
$\phi(\lambda+iy)$ is a positive trace-class operator on $\ell_2.$ Let also
\begin{equation} \label{F: eta(l+iy)}
  \eta(\lambda+iy) = \sqrt{\phi(\lambda+iy)};
\end{equation}
this is a positive Hilbert-Schmidt operator on $\ell_2.$ By Theorem~\ref{T: Abstract L.A.P.}, there exist the trace-class limit $\phi(\lambda+i0)$
and the compact norm limit $\eta(\lambda+i0).$
We define the fiber Hilbert space $\hlambda = \hlambda(H,F)$ as a subspace of $\ell_2$ by
\begin{equation} \label{F: hlambda}
  \hlambda = \closure{\rng{\eta(\lambda+i0)}}.
\end{equation}
Dimension of $\hlambda$ can be zero too. If necessary we indicate depends of $\hlambda$ on $H,$ but we shall always omit~$F$
and write $\hlambda(H).$

III. With a Hilbert-Schmidt rigging~$F$ one can associate in a standard way two Hilbert spaces $\hilb_{\pm 1} = \hilb_{\pm 1}(F)$
with natural Hilbert-Schmidt inclusions (see e.g. \cite[XI.6 Appendix]{RS3}, \cite[\S 2.6]{Az3v6})
\begin{equation} \label{F: hilb(1) hilb hilb(-1)}
  \hilb_1(F) \subset \hilb \subset \hilb_{-1}(F).
\end{equation}
The Hilbert space $\hilb_{1}$ is a vector space $\rng(F^*)$ endowed with scalar product $\scal{\cdot}{\cdot}_{\hilb_1}$
defined by formula
$$
  \scal{F^*\psi'}{F^*\psi''}_{\hilb_1} = \scal{\psi'}{\psi''}_\clK.
$$
The Hilbert space $\hilb_{-1}$ is the closure of~$\hilb$ endowed with scalar product $\scal{\cdot}{\cdot}_{\hilb_{-1}}$
defined by formula
$$
  \scal{\phi'}{\phi''}_{\hilb_{-1}} = \scal{F\phi'}{F\phi''}_\clK.
$$
There exists a natural pairing $\scal{\cdot}{\cdot}_{1,-1} \colon \hilb_{1} \times \hilb_{-1} \to \mbC$
defined by formula: for all $\phi',\phi'' \in \rng(F^*)$
$$
  \scal{\phi'}{\phi''}_{1,-1} = \scal{\phi'}{\phi''}.
$$
IV. For all $\lambda \in \Lambda(H,F)$ we define the \underline{evaluation operator}
\begin{equation} \label{F: euE def of}
  \euE_\lambda = \euE_\lambda(H,F) \colon \hilb_1(F) \to \hlambda
\end{equation}
as follows. Any vector $f \in \hilb_1(F)$ considered as an element of~$\hilb$ can uniquely be written as
$$
  f = \sum_{j=1}^\infty \beta_j \kappa_j \phi_j,
$$
where $(\beta_j) \in \ell_2.$ The value of the operator $\euE_\lambda$ at $f \in \hilb_1(F)$ is defined by formula
\begin{equation} \label{F: euE(f)=sum beta(j)eta(j)}
  \euE_\lambda(f) = \sum_{j=1}^\infty \beta_j \eta_j(\lambda),
\end{equation}
where $\eta_j(\lambda)$ is the $j$-th column of the matrix $\eta(\lambda+i0)$ (see (\ref{F: eta(l+iy)})).
One can show that this correctly defines a Hilbert-Schmidt operator (\ref{F: euE def of}),
see \cite[\S 3.1]{Az3v6} for details.

V. We define a \underline{direct integral Hilbert space $\euH = \euH(H,F)$} by formula
$$
  \euH = \int^\oplus_{\Lambda(H,F)} \hlambda(H)\,d\lambda.
$$
As a measurability base of this direct integral one can take functions $\Lambda(H,F) \ni \lambda \mapsto \euE_\lambda(\phi_j) \in \hlambda,$ $j=1,2,\ldots.$
Recall that it is possible that $\dim \hlambda = 0;$ thus, the set $\Lambda(H,F)$ can be replaced by the set
$\set{\lambda \in \Lambda(H,F) \colon \dim \hlambda = 0},$ which is a core of absolutely continuous spectrum of $H;$ but we prefer to work with the set
$\Lambda(H,F).$ Further, the choice of Lebesgue measure $d\lambda$ in definition of $\euH$ is not necessary, but quite natural as we shall see.

VI. One can now consider the \underline{operator $\euE = \euE(H,F) \colon \hilb \to \euH$} defined for $f \in \hilb_1$ by formula
$$
  [\euE(f)](\lambda) = \euE_\lambda(f).
$$
\begin{thm} \label{T: euE} \cite[Proposition 3.2.1, Proposition 3.3.5, Theorem 3.4.2]{Az3v6}
The operator $\euE\colon \hilb \to \euH$ is bounded, it vanishes on the singular subspace $\hilb^{(s)}(H)$ of $H,$
it is isometric on the absolutely continuous subspace $\hilb^{(a)}(H)$ of $H,$ and it is onto, that is, $\rng(\euE) = \euH.$
Further, the operator $\euE$ diagonalizes the absolutely continuous part of $H,$ that is, for all $f \in \dom(H)$ and for a.e. $\lambda \in \Lambda(H,F),$
we have
\begin{equation} \label{F: [euE(Hf)(l)=...}
  [\euE(Hf)](\lambda) = \lambda \euE_\lambda(f).
\end{equation}
\end{thm}
\noindent
This theorem shows that the operator $\euE$ is a version of the operator $\euF_j$ from~(\ref{F: euF(j)}).
But unlike~(\ref{F: euF(j)}), the set $\Lambda(H,F),$ the family of Hilbert spaces $\set{\hlambda \colon \lambda \in \Lambda(H,F)}$ and the operator $\euE$
are explicitly constructed without a.e. ambiguity. Further, if $f \in \hilb_1$ then the right hand side of~(\ref{F: [euE(Hf)(l)=...})
is defined for all $\lambda \in \Lambda(H,F).$

VII. Now assume that we are given two self-adjoint operators~$H_0$ and~$H_1$ on a rigged Hilbert space $(\hilb,F),$ such that the perturbation $V = H_1-H_0$
admits decomposition $F^*JF,$ where $J$ is a bounded self-adjoint operator on $\clK.$ Since on one hand~$F$ can be treated as an isomorphism of Hilbert spaces $\hilb_{-1}$ and~$\clK$
and on the other hand $F^*$ can be treated as isomorphism of~$\clK$ and $\hilb_1,$ it follows that the operator $V = F^*JF$ can be treated as a bounded operator $V \colon \hilb_{-1} \to \hilb_1.$
The operator $V$ considered as acting from~$\hilb$ to~$\hilb$ is a composition (in appropriate order) of the bounded operator $V \colon \hilb_{-1} \to \hilb_1$
with two Hilbert-Schmidt inclusions~(\ref{F: hilb(1) hilb hilb(-1)}).
Further, Theorem~\ref{T: Abstract L.A.P.} shows that for every $\lambda \in \Lambda(H_0,F) \cap \Lambda(H_1,F)$ the limits $R_{\lambda+i0}(H_j),$ $j=0,1,$
exist in uniform operator norm, if the operators $R_{\lambda+iy}(H_j)$ are considered as follows
$$
  R_{\lambda+iy}(H_j) \colon \hilb_1(F) \to \hilb_{-1}(F).
$$
Similarly, the operators $\Im R_{\lambda+i0}(H_j)$ can be treated as a trace-class operator $\hilb_1(F) \to \hilb_{-1}(F).$
Therefore, we can define a trace class operator
$$
  \mathfrak a_\pm(\lambda; H_1,H_0) \colon \hilb_1(F) \to \hilb_{-1}(F)
$$
by formula (compare with \cite[(2.7.4)]{Ya})
$$
  \mathfrak a_\pm(\lambda; H_1,H_0) = [1- R_{\lambda+i0}(H_1)V] \cdot \frac 1\pi \Im R_{\lambda+i0}(H_0).
$$
\begin{thm} \cite[\S 5.3]{Az3v6} For all $\lambda \in \Lambda(H_0,F) \cap \Lambda(H_1,F)$ there exists a unique (for each sign) bounded operator
\begin{equation} \label{F: w(+-) intro}
  w_\pm(\lambda; H_1,H_0) \colon \hlambda(H_0) \to \hlambda(H_1),
\end{equation}
such that for all $f,g \in \hilb_1$ there holds the equality
\begin{equation} \label{F: def of w(+-) from intro}
  \scal{\euE_\lambda(H_1) f}{w_\pm(\lambda; H_1,H_0) \euE_\lambda(H_0) g} = \scal{f}{\mathfrak a_\pm(\lambda;H_1,H_0)g}_{1,-1}.
\end{equation}
The operators~(\ref{F: w(+-) intro}), thus defined, have the following properties:
\begin{enumerate}
  \item[(a)] the operators~(\ref{F: w(+-) intro}) are unitary.
  \item[(b)] $w_\pm(\lambda; H_0,H_0) = 1_\hlambda$ and $w_\pm^*(\lambda; H_1,H_0) = w_\pm(\lambda; H_0,H_1).$
  \item[(c)] for any three self-adjoint operators $H_0,$~$H_1$ and $H_2$ such that $H_2-H_1$ and $H_1-H_0$
  admit decompositions $F^*J_1F$ and $F^*J_0F$ with bounded $J_0,J_1 \colon \clK \to \clK$ there hold the equalities
  $$
    w_\pm(\lambda; H_2,H_0) = w_\pm(\lambda; H_2,H_1)w_\pm(\lambda; H_1,H_0).
  $$
\end{enumerate}
\end{thm}
\noindent
Once the operators $w_\pm(\lambda; H_1,H_0)$ have been constructed, one defines the wave operators $W_\pm(H_1,H_0)$
and the scattering operator $\bfS(H_1,H_0)$ by formulas~(\ref{F: bfS and W(pm) = int}), where the scattering matrix $S(\lambda; H_1,H_0)$ is defined as an operator
$w_+^*(\lambda; H_1,H_0)w_-(\lambda; H_1,H_0).$ It is shown in \cite{Az3v6} that thus defined objects of scattering theory
possess well-known properties such as multiplicative property, the stationary formula and agreement with classical time-dependent definitions.

\subsection{Non-compact rigging}
It turns out that the method discussed above can be adjusted for pairs $H,F$ of operators such that $F E_\Delta$ is Hilbert-Schmidt,
where $E_\Delta$ is a spectral projection of $H$ and~$\Delta$ is a bounded measurable set.
Here we describe briefly the main idea.
Let $F \colon \hilb \to \clK$ be a bounded operator with trivial kernel and co-kernel, considered as a rigging in Hilbert space~$\hilb,$
and let~$H$ be a self-adjoint operator on~$\hilb,$ such that
\begin{equation} \label{I: FE(D) is H.S.}
  F E^H_\Delta \ \text{is Hilbert-Schmidt}
\end{equation}
for all bounded open intervals~$\Delta.$
Since $F E^H_\Delta$ is Hilbert-Schmidt, we define the set $\Lambda(H,F)$ as
$$
  \Lambda(H,F) = \bigcup_\Delta (\Lambda(HE_\Delta, FE_\Delta) \cap \Delta),
$$
where the union is taken over all bounded open sets $\Delta \subset \mbR.$
If $\lambda \in \Lambda(H,F)\cap \Delta$ then for the operator $HE_\Delta$ on the Hilbert
space $E_\Delta \hilb$ one can construct the fiber Hilbert spaces $\hlambda(\Delta),$ the evaluation operator $\euE_\lambda^\Delta,$ etc,
using the operator $F E_\Delta \colon E_\Delta \hilb \to \clK$ as a Hilbert-Schmidt rigging.
A difficulty here is that these objects depend on a choice of an interval~$\Delta.$ It is shown that the Hilbert
spaces $\hlambda(\Delta)$ for different bounded open sets~$\Delta,$ containing $\lambda,$ are naturally isomorphic and
the evaluation operators $\euE_\lambda^\Delta(H)$ can be naturally identified via the unitary operators used to identify
the Hilbert spaces $\hlambda(\Delta).$ The collection of Hilbert spaces
$$
  \hlambda(H,F) = \set{\hlambda(\Delta) \colon \Delta \ \text{is bounded, open and } \lambda \in \Delta}
$$
form a sheaf of fiber Hilbert spaces and the collection of operators
$$
  \euE_\lambda(H) = \set{\euE_\lambda^\Delta(H) \colon \Delta \ \text{is bounded, open and } \lambda \in \Delta}
$$
can be considered as an operator
$$
  \euE_\lambda(H) \colon \hilb_1(F) \to \hlambda(H,F),
$$
which is a generalization of the evaluation operator (\ref{F: euE def of}),
where $\hilb_{\pm 1}(F)$ are Hilbert spaces generated by the rigging $F.$
It is shown (Theorem~\ref{T: euE}) that extension of the operator $\euE_\lambda(H)$ to the Hilbert space associated with the direct integral
$$
  \euS(H,F) := \int_{\Lambda(H,F)}^\oplus \hlambda(H,F)\,d\lambda
$$
diagonalizes the operator $H.$ Once this is done, the rest of the theory is constructed similarly to the case of Hilbert-Schmidt rigging~$F.$
For example, the stationary formula (Theorem \ref{T: stationary formula}) for the scattering matrix takes the form
$$
  S(\lambda; H_1, H_0) = 1_\hlambda - 2\pi i \euE_\lambda(H_0) V (1 + R_{\lambda+i0}(H_0)V)^{-1}\euE_\lambda^\diamondsuit(H_0),
$$
where the operators in the right hand side are understood as follows
$$
   \hlambda(H,F) \stackrel {\euE_\lambda(H_0)} {\longleftarrow\!\!\!-\!\!-} \hilb_{1}(F)
   \stackrel{V}{\longleftarrow} \hilb_{-1}(F)
   \stackrel{R_{\lambda+i0}(H_0)}{\longleftarrow\!\!\!-\!\!-\!\!-\!\!-} \hilb_{1}(F)
   \stackrel{V}{\longleftarrow}  \hilb_{-1}(F)
   \stackrel {\euE_\lambda^\diamondsuit(H_0)} {\longleftarrow\!\!\!-\!\!-} \hlambda(H,F).
$$
Here $\euE_\lambda^\diamondsuit(H_0)$ is a modified conjugate of $\euE_\lambda(H_0) \colon \hilb_1(F) \to \hlambda(H,F)$ defined by equality
$$
  \scal{\euE_\lambda(H_0) f}{g}_{\hlambda(H,F)} = \scal{f}{\euE_\lambda^\diamondsuit(H_0) g}_{1,-1} \ \ \forall \ f \in \hilb_1(F), \ g \in \hlambda(H,F).
$$

\subsection{Description of sections}
In section~\ref{S: Sheaves of Hilbert spaces} we give an exposition of sheaves of Hilbert spaces.
In section~\ref{S: s.a. op-rs on rigged Hilbert spaces} we study self-adjoint operators~$H$ on rigged Hilbert spaces $(\hilb,F)$
which are compatible with the rigging~$F$ in the sense that the condition (\ref{I: FE(D) is H.S.}) holds;
in particular we construct a sheaf $\euS(H,F)$ of Hilbert spaces over an explicitly defined set $\Lambda(H,F)$ of full Lebesgue measure,
associated with a compatible pair $(H,F)$ and we show that the sheaf $\euS(H,F)$ gives a natural diagonalization of the operator~$H$
(Theorem~\ref{T: op-r euE}). In section~\ref{S: wave matrix} we give new definitions of wave matrices (\ref{F: def of w +-}) and wave operators
(\ref{F: def of W(pm)}), prove unitarity
(Corollary~\ref{C: w(pm) is unitary}) and the multiplicative property (Theorem~\ref{T: mult property of w(pm)})
of the wave matrices and show that these definitions coincide with classical time-dependent definitions (Theorem~\ref{T: class-l def of W(pm)}).
In section~\ref{S: scat matrix} we give new definitions of the scattering matrix (\ref{F: def of S(lambda)}) and the scattering operator
(\ref{F: def of bfS}) and give a new proof of the stationary formula for the scattering matrix (Theorem~\ref{T: stationary formula}).
We also show that thus introduced notions of scattering theory possess many other well-known properties, e.g. Theorems~\ref{T: properties of W(pm)} and~\ref{T: properties of S(lambda)}.
Finally, in Section~\ref{S: example} we give an example of a class of Schr\"odinger operators, to which the results of Sections~\ref{S: wave matrix}
and~\ref{S: scat matrix} are applied.

\section{Sheaves of Hilbert spaces}
\label{S: Sheaves of Hilbert spaces}
The notion of a sheaf was introduced by A.\,Grothendieck with the aim to give a general coordinate-independent definition of algebraic variety.
Sheaves of rings and vectors spaces are used extensively in topology and geometry, see for instance \cite{Sh} and \cite{Bred67}.
Our approach to scattering theory uses sheaves of Hilbert spaces. Since I was not able to find an appropriate reference on sheaves of Hilbert spaces,
which would satisfy needs of this paper, this section is devoted to an exposition of this notion.

Before proceeding to this exposition we note that reasons for using sheaves in geometry and in this paper
are different. In topology sheaves are used because of and for the study of non-trivial homotopical and homological structure of underlying topological space.
In this paper we use sheaves of Hilbert spaces over a certain subset $\Lambda$ of~$\mbR$ which has full Lebesgue measure and topological structure
of this set is not of interest.
The need in such sheaves arises here since we are able to construct certain direct integrals of fiber Hilbert spaces
over all bounded open subsets of~$\Lambda,$ but not over all~$\Lambda$ and therefore we need to ``glue'' together the direct integrals over intersecting bounded subsets.

\subsection{A sheaf of fiber Hilbert spaces}
\label{SS: sheaf of f.H.s-es at lambda}
Let~$\Lambda$ be a topological space with a fixed base~$\euB$ of
topology. In addition, later we assume that~$\euB$ contains intersection of any two sets from~$\euB$ as long as this intersection
is not empty and that the space $\Lambda$ is a union of an increasing family of sets $\Delta_1 \subset \Delta_2 \subset \ldots$
from~$\euB.$

Let $\lambda \in \Lambda.$ By $\euB_\lambda$ we denote
the subset $\set{\Delta \in \euB \colon \lambda \in \Delta}$
of~$\euB.$ A \emph{sheaf of fiber Hilbert spaces} (or fiber of a sheaf of Hilbert spaces) $\hlambda$
at~$\lambda$ is a collection of Hilbert spaces
$$
  \hlambda := \set{\hlambda(\Delta) \colon \Delta \in \euB_\lambda}
$$
and a collection of unitary isomorphisms
$$
  \set{U_{\Delta_2,\Delta_1}(\lambda) \colon \hlambda(\Delta_1) \stackrel \sim\longrightarrow \hlambda(\Delta_2) \mid (\Delta_1,\Delta_2) \in \euB_\lambda^2}
$$
such that for any three, not necessarily distinct, open subsets $\Delta_1, \Delta_2$ and $\Delta_3$ from $\euB_\lambda$
we have the equality
\begin{equation} \label{F: gluing property for U}
  U_{\Delta_3,\Delta_1}(\lambda) = U_{\Delta_3,\Delta_2}(\lambda)U_{\Delta_2,\Delta_1}(\lambda).
\end{equation}
It follows from this that for any $\Delta_1, \Delta_2,
\Delta \in \euB_\lambda$ the equalities
$$
  U_{\Delta_2,\Delta_1}(\lambda)^* = U_{\Delta_1,\Delta_2}(\lambda) \ \text{and} \ U_{\Delta,\Delta}(\lambda) = 1_{\hlambda(\Delta)}
$$
hold, where $1_{\hlambda(\Delta)}$ is the identity operator on $\hlambda(\Delta).$ An element of
the sheaf $\hlambda$ is a collection of vectors
$$
  f(\lambda) = \set{f_\Delta(\lambda) \in \hlambda(\Delta) \colon \Delta \in \euB_\lambda},
$$
such that for any $\Delta_1, \Delta_2 \in \euB_\lambda$ there holds the
equality
\begin{equation} \label{F: glue f(Delta)}
  U_{\Delta_2,\Delta_1}(\lambda) f_{\Delta_1}(\lambda) = f_{\Delta_2}(\lambda).
\end{equation}
A sheaf $\hlambda$ of fiber Hilbert spaces at $\lambda \in
\Lambda$ is a vector space with scalar product
\begin{equation} \label{F: scal prod in sheaf}
  \scal{f(\lambda)}{g(\lambda)}_\hlambda = \scal{f_\Delta(\lambda)}{g_\Delta(\lambda)}_{\hlambda(\Delta)},
\end{equation}
where~$\Delta$ is any element of $\euB_\lambda.$ Plainly, this
scalar product does not depend on the choice of~$\Delta.$ It is
equally obvious that $\hlambda$ is a Hilbert space. In what
follows, the Hilbert spaces $\hlambda(\Delta)$ will usually be
subspaces of a single Hilbert space~$\mathfrak h,$ namely
$\mathfrak h = \ell_2.$ In this case we say that $\hlambda$ is a
sheaf of fiber Hilbert spaces in $\mathfrak h.$

\subsubsection{Operators acting on a sheaf of fiber Hilbert spaces}
Let~$\clK$ be a Hilbert space. An operator $T$ from a Hilbert
space~$\clK$ to $\hlambda$ is a family of operators $
  \set{T_\Delta\colon \clK \to \hlambda(\Delta), \Delta \in \euB_\lambda}
$ such that for any $\Delta_1, \Delta_2 \in \euB_\lambda$ the
equality
\begin{equation} \label{F: T(D2)=U(D2,D1)T(D1)}
  T_{\Delta_2} = U_{\Delta_2,\Delta_1}(\lambda)T_{\Delta_1}
\end{equation}
holds. An operator $T$ from $\hlambda$ to a Hilbert space~$\clK$
is a family of operators $$
  \set{T_\Delta\colon \hlambda(\Delta) \to \clK, \Delta \in \euB_\lambda}
$$ such that for any $\Delta_1, \Delta_2 \in \euB_\lambda$ and any
$f(\lambda) \in \hlambda$ there holds the equality
$$
  T_{\Delta_2} (f_{\Delta_2}(\lambda)) = T_{\Delta_1}(f_{\Delta_1}(\lambda)).
$$
This condition is equivalent to this one:
$$
  T_{\Delta_2}U_{\Delta_2,\Delta_1}(\lambda) = T_{\Delta_1}.
$$

\subsection{A sheaf of Hilbert spaces}
Let $\set{\hilb_\lambda, \lambda \in \Delta}$ be a family of
subspaces of a Hilbert space~$\hilb.$ We say that this family is
\emph{measurable}, if the family of orthogonal projection
$P_\lambda$ onto $\hilb_\lambda$ is measurable, that is, if for
any $f,g \in \hilb,$ the function $\Delta \ni \lambda \mapsto \scal{f}{P_\lambda g}$ is
measurable. A~\emph{measurable section} of this family is a
measurable function $f \colon \Delta \to \hilb$ such that
$f(\lambda) \in \hilb_\lambda$ for all $\lambda \in \Delta.$

Let $\mathfrak h$ be a Hilbert space, let~$\Lambda$ be a Hausdorff
topological space with a fixed base of topology $\euB,$ such that
if $\Delta_1, \Delta_2 \in \euB$ then $\Delta_1 \cap \Delta_2 \in
\euB,$ and let $\rho$ be a Borel measure in~$\Lambda.$ A~\emph{sheaf of Hilbert spaces} $\euS$ over~$\Lambda$ is a family
of sheaves of fiber Hilbert spaces $\set{\hlambda \colon \lambda
\in \Lambda}$ in $\mathfrak h,$ such that (1) for every $\Delta
\in \euB$ the family of Hilbert spaces $\set{\hlambda(\Delta)
\colon \lambda \in \Delta}$ is measurable and (2) for any
$\Delta_1, \Delta_2 \in \euB$ the family
$$
  U_{\Delta_2,\Delta_1} = \set{U_{\Delta_2,\Delta_1}(\lambda) , \lambda \in \Delta_1 \cap \Delta_2}
$$
is also measurable, that is, it maps measurable sections of
$\set{\hlambda(\Delta_1) \colon \lambda \in \Delta_1 \cap
\Delta_2}$ to measurable sections of $\set{\hlambda(\Delta_2)
\colon \lambda \in \Delta_1 \cap \Delta_2}.$ It follows that for
every $\Delta \in \euB$ the dimension function
$$
  \Delta \ni \lambda \mapsto \dim(\hlambda(\Delta)) \in \set{0,1,2,\ldots,\infty}
$$
is measurable. Further, to every $\Delta \in \euB$ we
can assign a Hilbert space
$$
  \euH(\Delta) := \int_\Delta^\oplus \hlambda(\Delta)\,\rho(d\lambda),
$$
--- the direct integral of fiber Hilbert spaces $\hlambda(\Delta).$
Elements of $\euH(\Delta)$ are measurable square integrable
sections of the family $\set{\hlambda(\Delta) \colon \lambda \in
\Delta},$ where two sections are identified if they coincide for $\rho$-a.e. $\lambda \in \Lambda.$ Definition of the scalar product in $\euH(\Delta)$ is
obvious. 
The operator
$$
  U_{\Delta_2,\Delta_1} = \int_{\Delta_1\cap \Delta_2}^\oplus U_{\Delta_2,\Delta_1}(\lambda)\,\rho(d\lambda)
$$
is a unitary isomorphism of Hilbert spaces $\euH(\Delta_1)
\big|_{\Delta_1 \cap \Delta_2}$ and $\euH(\Delta_2)
\big|_{\Delta_1 \cap \Delta_2},$ where
$$
  \euH(\Delta_j) \big|_{\Delta_1 \cap \Delta_2} = \int_{\Delta_1 \cap \Delta_2}^\oplus \hlambda(\Delta_j)\,\rho(d\lambda), \quad j = 1,2.
$$

A sheaf $\euS$ of Hilbert spaces can be given a structure of
Hilbert space, the construction of which follows. A
\emph{measurable section} $f$ of the sheaf $\euS$ is a family
$f(\lambda)$ of elements of fiber Hilbert spaces $\hlambda,$ given
for $\rho$-a.e.~$\lambda,$ such that for any $\Delta \in \euB$ the
section $f_\Delta(\cdot)$ of the family  $\set{\hlambda(\Delta)
\colon \lambda \in \Delta}$ is measurable. A section $f$ is
\emph{square integrable}, if the number
\begin{equation} \label{F: S(H,F) norm}
  \norm{f}_{\euS} := \lim_{n \to \infty} \norm{f_{\Delta_n}}_{\euH\brs{\Delta_n}}
  = \sup_{\Delta \in \euB} \norm{f_{\Delta}}_{\euH\brs{\Delta}}
\end{equation}
is finite, where $\Delta_1 \subset \Delta_2 \subset \ldots$ is an
increasing sequence of elements of~$\euB$ such that
$\cup_{n=1}^\infty \Delta_n = \Lambda.$ Obviously, this definition
does not depend on the choice of the sequence $(\Delta_n)_{n=1}^\infty.$
The set of all square integrable sections of $\euS$ is a vector space
in an obvious way, where as usual we identify two sections which
coincide for $\rho$-a.e. $\lambda \in \Lambda.$ We denote this vector space
by the same symbol $\euS.$
We also denote an element $\set{f(\lambda) \colon \lambda \in \Lambda}$ of $\euS$ by
\begin{equation} \label{F: int f(l)rho(dl) in euS}
  \int_{\Lambda}^\oplus f(\lambda)\,\rho(d\lambda).
\end{equation}
\noindent
The scalar product of two square
integrable sections $f,g \in \euS$ is defined by the formula
$$
  \scal{f}{g}_{\euS} := \lim_{n \to \infty}
  \scal{f_{\Delta_n}}{g_{\Delta_n}}_{\euH\brs{\Delta_n}},
$$
where $\Delta_1 \subset \Delta_2 \subset \ldots$ is a sequence as
above. This scalar product is well-defined in the sense that it
does not depend on the choice of the sequence $\Delta_1 \subset
\Delta_2 \subset \ldots$ of elements of~$\euB.$
\begin{thm} $\euS$ is a Hilbert space.
\end{thm}
\begin{proof}
Let $f_1, f_2, \ldots \in \euS$ be a Cauchy sequence. Plainly, for
any open set $\Delta \in \euB$ the sequence $f_{1,\Delta},
f_{2,\Delta}, \ldots$ is also Cauchy. Since $\euH(\Delta)$ is
complete, the last sequence converges to some $f_\Delta \in
\euH(\Delta).$
Since $U_{\Delta_2,\Delta_1}(\lambda)$ is continuous, we have for a.e. $\lambda$
\begin{equation*}
  \begin{split}
   U_{\Delta_2,\Delta_1}(\lambda) f_{\Delta_1}(\lambda)
    & = U_{\Delta_2,\Delta_1}(\lambda) \lim_{n \to \infty} f_{n,\Delta_1}(\lambda)
    \\ & = \lim_{n \to \infty} U_{\Delta_2,\Delta_1}(\lambda) f_{n,\Delta_1}(\lambda)
    = \lim_{n \to \infty} f_{n,\Delta_2}(\lambda)
    = f_{\Delta_2}(\lambda).
  \end{split}
\end{equation*}
It follows that $f = \set{f_\Delta}$ defines a measurable section of $\euS.$
We have
\begin{equation*}
  \begin{split}
  \norm{f}^2_{\euS} & = \sup_{k\in\mbN} \norm{f_{\Delta_k}}^2_{\euH\brs{\Delta_k}}
    = \sup_{k\in\mbN} \lim_{n \to \infty} \norm{f_{n,\Delta_k}}^2_{\euH\brs{\Delta_k}}
  \\ & \leq \lim_{n \to \infty} \sup_{k\in\mbN} \norm{f_{n,\Delta_k}}^2_{\euH\brs{\Delta_k}} = \lim_{n \to \infty} \norm{f_n}_\euS^2 < \infty,
  \end{split}
\end{equation*}
where the last inequality follows from $(f_n)$ being Cauchy.
It follows that $f \in \euS.$ Now, we show that $f_n$ converges to
$f.$ Since $(f_n)$ is Cauchy, it follows from $\abs{\norm{f_n(\lambda)} - \norm{f(\lambda)}}
\leq \norm{f_n(\lambda) - f(\lambda)}$ that
the sequence of functions $\lambda \mapsto (\norm{f_n(\lambda)}_\hlambda)$ is
Cauchy in $L^2(\Lambda,d\rho).$ By construction,
$\norm{f_n(\lambda)-f(\lambda)}_{\hlambda} \to 0$ for $\rho$-a.e.
$\lambda \in \Lambda.$ It follows that
$\norm{f_n(\lambda)}_{\hlambda} \to \norm{f(\lambda)}_{\hlambda}$
as functions of~$\lambda$ in $L^2(\Lambda,d\rho).$ This implies
that for any $\eps>0$ there exists an open set $\Delta \subset \Lambda$
and a number $N_1$ such that for all $n\geq N_1$ we have
$$
  \int_{\Lambda\setminus \Delta} \norm{f_n(\lambda)}^2\,d\rho(\lambda) < \eps/4 \quad \text{and} \quad
  \int_{\Lambda\setminus \Delta} \norm{f(\lambda)}^2\,d\rho(\lambda) < \eps/4.
$$
Further, for some $N_2$ and all $n\geq N_2$ we have $\norm{f_{n,\Delta}-f_\Delta} < \eps/2.$
It follows that $f_n$ converges to $f$ in~$\euS.$
\end{proof}
The \emph{support} of an element $f$ of the Hilbert space $\euS$
is defined by formula $\supp f = \bigcup_{\Delta} \supp f_\Delta,$
where the union is over all open sets $\Delta \in \euB.$ Elements
of the sheaf Hilbert space $\euS$ can be represented by either a
measurable square integrable family (\ref{F: int f(l)rho(dl) in euS})
or by a family of vectors $f_\Delta \in \euH(\Delta)$ which satisfy the gluing property~(\ref{F: glue
f(Delta)}) for all $\Delta_1$ and $\Delta_2$ and for a.e. $\lambda \in \Delta_1 \cap \Delta_2$ and such that the supremum
in (\ref{F: S(H,F) norm}) is finite.

\subsubsection{Operators acting on a sheaf of Hilbert spaces}
\label{SSS: op-rs on sheaf euS}
Let~$\clK$ be a Hilbert space.
As long as definition of sheaf Hilbert space $\euS$ is given, a standard definition of an operator acting from $\clK$ to $\euS$
applies. But in practice there are several equivalent ways to define such an operator.
In order to define an operator $T$
from~$\clK$ to the sheaf Hilbert space $\euS$ with domain $\euD \subset \clK$ one can present for $\rho$-a.e. $\lambda \in \Lambda$
and for all $\Delta \in \euB_\lambda$ an operator
$$
  T_\Delta(\lambda) \colon \euD \to \hlambda(\Delta),
$$
such that for $\rho$-a.e. $\lambda \in \Lambda,$ for any $\Delta_1, \Delta_2 \in \euB_\lambda$ and for any $f \in \euD$ there holds the equality
$$
  U_{\Delta_2,\Delta_1}(\lambda)T_{\Delta_1}(\lambda)f = T_{\Delta_2}(\lambda)f,
$$
and such that the mapping $\Delta \ni \lambda \mapsto T_\Delta(\lambda)f \in \hlambda(\Delta)$ is measurable.
Another way to define an operator $\clK \to \euS$ is to assign
to every $f \in \clK$ an element $Tf = \set{(Tf)_\Delta \in \euH(\Delta) \colon \Delta \in \euB}$ of~$\euS.$
A family of operators $T(\lambda)\colon \clK \to \hlambda$
is measurable, if for any $f \in \clK$ the section $T(\lambda)f$
of the family $\set{\hlambda}$ is measurable. Given a measurable
family of operators $T(\lambda)\colon \clK \to \hlambda,$ one can
define an operator
$$
  T = \int_\Lambda^\oplus T(\lambda)\,\rho(d\lambda) \colon \clK \to \euS.
$$

\section{Self-adjoint operators on rigged Hilbert spaces}
\label{S: s.a. op-rs on rigged Hilbert spaces}
Given a self-adjoint operator~$H_0$ on a Hilbert space~$\hilb$ and
a self-adjoint perturbation~$V,$ our aim is to construct
explicitly the wave matrix $w_\pm(\lambda; H_0+V,H_0)$ for all
real numbers~$\lambda$ from some explicitly given set of full
Lebesgue measure $\Lambda \subset \mbR.$ In order to do this, we
need to impose some additional structure. In case of trace-class
perturbations~$V$ this additional structure is a Hilbert-Schmidt
rigging operator~$F$ (see \cite{Az3v6}). Given a Hilbert-Schmidt
rigging~$F,$ for any self-adjoint operator~$H_0$ one can define the
set of full Lebesgue measure $\Lambda(H_0,F),$ such that for all
$\lambda \in \Lambda(H_0,F) \cap \Lambda(H_0+V,F)$ it is possible
to define the fiber Hilbert space $\hlambda,$ the wave matrices
$w_\pm(\lambda; H_0+V,H_0)$ etc (see \cite{Az3v6}).

In case of non-compact perturbations $V,$ we need to generalize
the notion of the Hilbert-Schmidt rigging operator~$F.$ In the
trace-class case, the perturbation~$V$ admits the factorization $V
= F^*JF,$ where $J$ is any bounded operator. We keep this
factorization in the generalization of the trace-class theory
which covers the case of non-compact perturbations $V.$ We also
assume that in the factorization $V = F^*JF$ the operator $J$ is
still allowed to be any bounded operator on the auxiliary Hilbert
space~$\clK.$ This implies that the rigging operator~$F$
can no longer be assumed to be compact. Further, since the wave matrix
cannot exist for general pairs $(H_0,H_0+V),$ this also means
that restrictions, which ensure existence of the wave matrices, shift from the perturbation~$V$ to the
rigging operator~$F.$ Thus, we need to impose some conditions on the pairs $(H_j,F),$ $j=0,1.$ It
turns out that we need only one condition: the operators
$FE_\Delta^{H_j},$ $j=0,1,$ are Hilbert-Schmidt for any bounded open set
$\Delta.$ In this section we study the pairs $(H,F)$ of operators which satisfy this condition.

\subsection{Generalized rigging}
\label{S: generalized rigging} Let~$\hilb$ be a Hilbert space (all
Hilbert spaces in this paper are complex and separable, but not
necessarily infinite dimensional). The Hilbert space~$\hilb$ is
the main Hilbert space on which operators act. In addition
to~$\hilb,$ we use an auxiliary Hilbert space~$\clK.$ By $\dom(T)$
we denote domain of an operator~$T.$

\begin{defn} A (generalized) rigging operator~$F$ on~$\hilb$ is a
bounded operator $F \colon \hilb \to \clK$ with trivial kernel and dense range.
\end{defn}
\begin{rems} \rm That the range of~$F$ is dense in~$\clK$ is not essential: one can always replace~$\clK$
by the closure of the image of~$F.$ But this is convenient; otherwise, we would need for instance to write $g \in \overline{\rng{F}}$
instead of $g \in \clK.$ Sometimes riggings with not dense ranges appear naturally; in such cases we assume that the auxiliary Hilbert space~$\clK$
changes appropriately.
\end{rems}
Our aim is to study operators~$H$ which act on~$\hilb$ with a
predefined rigging~$F.$ Firstly, we shall consider objects which already can be associated
with the pair $(\hilb,F)$ (all of them are well-known).
A rigging operator~$F$ admits polar decomposition $F = U\abs{F},$
where the self-adjoint operator $\abs{F}$ has trivial kernel and (therefore)
dense range, so that the inverse operator $\abs{F}^{-1}$ exists as an unbounded operator.
Therefore, we have a scale of Hilbert spaces $\hilb_\alpha(F),$ $\alpha \in \mbR,$ introduced as follows: let
$$
  \hilb_\infty(F) = \set{f \in \hilb \colon f \in \dom(\abs{F}^k) \ \text{for all} \   k \in \mbZ},
$$
then $\hilb_\alpha$ is the completion of $\hilb_\infty$ endowed
with scalar product
$$
  \scal{f}{g}_{\alpha} =
  \scal{\abs{F}^{-\alpha}f}{\abs{F}^{-\alpha}g}.
$$
Since we have assumed~$F$ to be a bounded operator, for $\alpha > \beta$
there is a natural inclusion $\hilb_\alpha(F) \hookrightarrow
\hilb_\beta(F).$ Hilbert spaces $\hilb_\alpha$ and $\hilb_{\beta}$
are isomorphic, and the operator $\abs{F}^{\alpha-\beta} \colon
\hilb_\beta(F) \to \hilb_\alpha(F)$ is an isomorphism. Hence, the
expression $\abs{F}$ can be understood in two different ways: as
an operator on~$\hilb,$ or as an isomorphism of Hilbert spaces
$\hilb_\alpha$ and $\hilb_{\alpha+1}.$ Since the operator $U$ from
the polar decomposition of~$F$ is unitary (assuming that $\clK = \overline{\rng{F}}$), the operator~$F$ can
also be treated as a natural isomorphism of $\hilb_{-1}$ and
$\clK,$ and the operator $F^*$ can be treated as a natural
isomorphism of~$\clK$ and $\hilb_1.$ The same remark relates to
any power of $\abs{F}.$

For any real number $\alpha$ there is a natural pairing
$\scal{\cdot}{\cdot}_{\alpha,-\alpha} \colon \hilb_{\alpha} \times
\hilb_{-\alpha} \to \mbC,$ such that for any $f,g \in
\hilb_{\alpha} \cap \hilb_{-\alpha}$ we have
$$
  \scal{f}{g}_{\alpha,-\alpha} = \scal{f}{g}.
$$
In the sequel we shall need only Hilbert spaces $\hilb_1(F)$ and
$\hilb_{-1}(F).$ For this reason, further definitions are given
only for this case of $\alpha = \pm 1.$ Elements of the Hilbert
space $\hilb_1(F)$ are to be considered as smooth or regular
vectors of the Hilbert space~$\hilb,$ while elements of
$\hilb_{-1}(F)$ are considered as singular (improper) vectors, which may not
lie in the main Hilbert space~$\hilb.$
If $\tilde \hilb$ is another Hilbert space and $A \colon \hilb_1
\to \tilde \hilb$ is a bounded operator, then there is a unique bounded operator
$A^\diamondsuit \colon \tilde \hilb \to \hilb_{-1},$ such that for any
$f \in \hilb_{1}$ and $g \in \tilde \hilb$ there holds the equality
\begin{equation} \label{F: def of diamond}
  \scal{A f}{g}_{\tilde \hilb} = \scal{f}{A^\diamondsuit g}_{-1,1}.
\end{equation}
The operator $A^\diamondsuit$ can also be defined by formula $A =
\abs{F}^{-2}A^*,$ where $\abs{F}^{-2}$ is understood as an
isomorphism of $\hilb_1(F)$ and $\hilb_{-1}(F).$
We shall have an opportunity to use the following formula which holds for any $\alpha \in \mbR:$
\begin{equation} \label{F: norm(f)(-alpha)}
  \norm{f}_{\hilb_{-\alpha}(F)} = \sup_{\norm{g}_{\hilb_\alpha(F)}=1} \abs{\scal{f}{g}_{-\alpha,\alpha}}.
\end{equation}

\subsection{Compatible pairs $(H,F)$}
Throughout this paper we shall assume that $F$ is a bounded operator from $\hilb$ to $\clK$
with trivial kernel. We shall assume that $F$ also has dense range though it is not necessary.
We shall call $F$ a rigging operator.
\begin{defn} \label{D: compatible (H,F)} Let~$H$ be a self-adjoint operator on the Hilbert
space~$\hilb$ with rigging~$F.$ The operator~$H$ is
\emph{compatible} with rigging~$F,$ if
for all bounded open subsets~$\Delta$ of~$\mbR$ the operator
\begin{equation} \label{F: F(Delta)}
  F_\Delta := F E_\Delta \ \ \text{is Hilbert-Schmidt.}
\end{equation}
\end{defn}
\noindent
Proof of the following lemma is standard and well-known, but we give it for completeness.
\begin{lemma} \label{L: FRF* is compact} If $H$ is a self-adjoint operator compatible with a rigging~$F,$
then for all non-real $z$ the operator $FR_z(H)F^*$ is compact.
\end{lemma}
\begin{proof} The operator $FR_z(H)F^*$ can be approximated in uniform norm arbitrarily well by a compact operator $FE^H_{[-a,a]} R_z(H)F^*$
where $a$ is large enough number.
\end{proof}

\noindent Given a compatible pair $(H,F)$ of operators one can introduce a sheaf of
Hilbert spaces $\euS(H,F)$ over some measurable set $\Lambda(H,F) \subset \mbR$
of full Lebesgue measure, naturally associated with the pair
$(H,F);$ the main property of this sheaf is that it diagonalizes the
operator $H.$ The rest of this section is devoted to this construction.

\subsection{Hilbert spaces $\hilb_{\pm 1}(\Delta)$}

Given a self-adjoint operator~$H$ on a rigged Hilbert space
$(\hilb,F),$ such that~$H$ is compatible with the rigging~$F,$ we
can further study some properties of Hilbert spaces $\hilb_1(F)$
and $\hilb_{-1}(F),$ introduced in Section~\ref{S: generalized
rigging}. Given a bounded open set~$\Delta,$ we can consider the
Hilbert-Schmidt operator $F_\Delta = FE_\Delta$ as a
Hilbert-Schmidt rigging in the Hilbert space $E_\Delta \hilb.$ Corresponding Hilbert spaces~$\hilb_{\pm 1}$ of
regular and singular vectors will be denoted by~$\hilb_{\pm 1}(\Delta).$ Let
\begin{equation} \label{F: F(Delta) = Schimidt ...}
  F_\Delta = \sum_{j=1}^\infty \kappa_j^\Delta \scal{\phi_j^\Delta}{\cdot}\psi_j^\Delta
\end{equation}
be a fixed Schmidt representation of~$F_\Delta.$ Elements of
$\hilb_1(\Delta)$ have the form
\begin{equation} \label{F: elements of hilb1(Delta)}
  f = \sum_{j=1}^\infty \beta_j \kappa_j^\Delta \phi_j^\Delta
\end{equation}
where $\beta=(\beta_j) \in \ell_2$ and
\begin{equation} \label{F: norm(f)hilb1=norm(beta)ell2}
  \norm{f}_{\hilb_1(\Delta)} = \norm{\beta}_{\ell_2}.
\end{equation}
The one-to-one correspondence between elements $f$ of $\hilb_1(\Delta)$ and $\ell_2$-vectors
$\beta=(\beta_j)$ given by~(\ref{F: elements of hilb1(Delta)}) is a unitary operator.
For this reason, for an element $f$ of $\hilb_1(\Delta)$ we use notation $\beta^\Delta_j(f)$ so that, by definition,
\begin{equation} \label{F: elements of hilb1(Delta)TWO}
  f = \sum_{j=1}^\infty \beta^\Delta_j(f) \kappa_j^\Delta \phi_j^\Delta, \ \ \beta^\Delta(f) \in \ell_2.
\end{equation}
We have from~(\ref{F: F(Delta) = Schimidt ...})
\begin{equation} \label{F: F(Delta)*=Schmidt series}
  F_\Delta^* = \sum_{j=1}^\infty \kappa_j^\Delta \scal{\psi_j^\Delta}{\cdot}\phi_j^\Delta.
\end{equation}
It follows that the range of $F_\Delta^*$ consists of all vectors of the
form~(\ref{F: elements of hilb1(Delta)}), that is,
\begin{equation} \label{F: hilb1(Delta)=ran(F*)}
  \hilb_1(\Delta) = \rng(F_\Delta^*).
\end{equation}
The inclusion $E_\Delta f \in \hilb_1(\Delta)$ means that restriction $E_\Delta f$ of
a vector $f \in \hilb$ to a bounded Borel set~$\Delta$ is regular in a
certain sense. It is natural to expect that if a vector $f$ is
regular on some Borel set~$\Delta,$ then $f$ is to be regular also
on any Borel subset $\Delta_0$ of the set~$\Delta.$
One could compare this situation to continuity property of a mapping: if a mapping is continuous on a set, then it is also continuous on a subset of the set.
The following proposition asserts that regularity of a vector $f$ on~$\Delta$ is inherited by subsets of~$\Delta,$
but this is where the analogy with continuity of functions ends, since the restriction mapping turns out to be surjective.

\begin{prop} \label{P: E(Delta) is contraction} Let $\Delta_1 \subset \Delta_2$ be two open bounded subsets of $\mbR.$
If $f \in \hilb_1(\Delta_2),$ then $E_{\Delta_1} f \in
\hilb_1(\Delta_1).$ Moreover, the mapping $E_{\Delta_1} \colon
\hilb_1(\Delta_2) \to \hilb_1(\Delta_1)$ is a contraction and its
image coincides with $\hilb_1(\Delta_1).$
\end{prop}
\begin{proof} If $f \in \hilb_1(\Delta_2),$ then, by~(\ref{F: hilb1(Delta)=ran(F*)}), there exists $g \in \clK$ such that $f = E_{\Delta_2}F^*g.$
It follows that
$$
  E_{\Delta_1} f = E_{\Delta_1}F^*g = F_{\Delta_1}^* g \in \hilb_1(\Delta_1),
$$
where the last inclusion follows from~(\ref{F: hilb1(Delta)=ran(F*)}). Since, again by~(\ref{F: hilb1(Delta)=ran(F*)}), the set $\hilb_1(\Delta)$ is the image of
the operator $F_\Delta^* = E_\Delta F^*,$ it follows that
$E_{\Delta_1} \hilb_1(\Delta_2)$ coincides with
$\hilb_1(\Delta_1).$

Direct calculation shows that for any $\Delta_1 \subset \Delta_2$ 
the following formula holds
\begin{equation} \label{F: beta(D)(Ef)=}
  \beta_j^{\Delta_1}\brs{E_{\Delta_1} f} = \sum_{k=1}^{d(\Delta_2)}  \scal{\psi_j^{\Delta_1}}{\psi_k^{\Delta_2}} \beta^{\Delta_2}_k(f),
\end{equation}
where $d(\Delta_2)$ is the dimension of the Hilbert space $E_{\Delta_2}\hilb.$ Letting
$$
  \Psi^{\Delta_1,\Delta_2} = \brs{\scal{\psi_j^{\Delta_1}}{\psi_k^{\Delta_2}}} _{j=1,k=1}^{d(\Delta_1), d(\Delta_2)},
$$
we can rewrite (\ref{F: beta(D)(Ef)=}) as
\begin{equation} \label{F: boxed formula}
  \beta^{\Delta_1}\brs{E_{\Delta_1} f} = \Psi^{\Delta_1,\Delta_2} \beta^{\Delta_2}\brs{f}.
\end{equation}
The matrix $\Psi^{\Delta_1,\Delta_2}$ is a contraction. Indeed,
$\Psi^{\Delta_1,\Delta_2}$ is the matrix of the orthogonal
projection of the linear span of $(\psi_j^{\Delta_2})$ onto the
linear span of $(\psi_j^{\Delta_1}),$ with the orthonormal bases
$(\psi_j^{\Delta_2})_{j=1}^{d(\Delta_2)}$ and
$(\psi_j^{\Delta_1})_{j=1}^{d(\Delta_1)}$ respectively. That the
operator $E_{\Delta_1} \colon \hilb_1(\Delta_2) \to
\hilb_1(\Delta_1)$ is a contraction now follows from~(\ref{F:
boxed formula}), (\ref{F: norm(f)hilb1=norm(beta)ell2}) and from the fact that the matrix
$\Psi^{\Delta_1,\Delta_2}$ is a contraction.
\end{proof}
Clearly, for any $\Delta_1 \subset \Delta_2 \subset \Delta_3,$ we
have the equality
$$
  \Psi^{\Delta_1,\Delta_3} = \Psi^{\Delta_1,\Delta_2} \Psi^{\Delta_2,\Delta_3}.
$$
The formula~(\ref{F: boxed formula}) can be rewritten as
\begin{equation} \label{F: E(Delta0)=Psi(Delta1,Delta0)E(Delta0)}
  E_{\Delta_1} \brs{\sum_{k=1}^{d(\Delta_2)}  \beta^{\Delta_2}_k\kappa^{\Delta_2}_k \phi^{\Delta_2}_k}
  = \brs{\sum_{j=1}^{d(\Delta_1)}  \brs{\Psi^{\Delta_1,\Delta_2}\beta^{\Delta_2} }_j \kappa^{\Delta_1}_j \phi^{\Delta_1}_j}.
\end{equation}

\begin{lemma} \label{L: norm(f)=lim norm(E(Delta)f}
For any $f \in \hilb_1(F)$ we have \ $E_\Delta f \in \hilb_1(\Delta)$ and $$\norm{f}_{\hilb_1(F)} =
\lim_{\Delta \to \mbR} \norm{E_\Delta f}_{\hilb_1(\Delta)}.$$
\end{lemma}
\begin{proof} Since $\hilb_1(F)$ and $\rng{F^*}$ coincide as sets,
by (\ref{F: hilb1(Delta)=ran(F*)}) the inclusion $f \in \hilb_1(F)$ implies
that $E_\Delta f \in \hilb_1(\Delta).$ In particular, the equality to be proved makes
sense. Further, there exists a unique vector $g \in \clK$ such that $f =
F^*g,$ and, by definition,
\begin{equation} \label{F: norm(f)hilb1=norm(g)clK}
  \norm{f}_{\hilb_1(F)} = \norm{g}_{\clK}.
\end{equation}
Hence, $E_\Delta f = F_\Delta^* g.$ Since the operator
$F_\Delta^*$ has the form~(\ref{F: F(Delta)*=Schmidt series}), it
follows that
$$
  E_\Delta f = F_\Delta^* g = \sum_{j=1}^\infty \kappa_j^\Delta \scal{\psi_j^\Delta}{g}\phi_j^\Delta.
$$
This and (\ref{F: norm(f)hilb1=norm(beta)ell2}) imply that
$$
  \norm{E_\Delta f}_{\hilb_1(\Delta)} = \norm{\brs{\scal{\psi_j^\Delta}{g}}}_{\ell_2}.
$$
If we denote by $E_\Delta^\psi$ the projection onto the closed
linear span of $(\psi_j^\Delta),$ we can rewrite the previous
equality as
\begin{equation} \label{F: norm 1(E(Delta)f)=...}
  \norm{E_\Delta f}_{\hilb_1(\Delta)} = \norm{E_\Delta^\psi g}.
\end{equation}
Since $E_\Delta^\psi \nearrow 1_\clK$ as $\Delta \to \mbR$
(indeed, $E_\Delta^\psi$ projects onto the image of
$E_\Delta\hilb$ under the mapping~$F$), it follows that
$E_\Delta^\psi  \to 1_\clK$ in the strong operator topology, so that $
  \norm{E_\Delta^\psi g} \to \norm{g}
$ as $\Delta \to \mbR.$ Combining this with~(\ref{F:
norm(f)hilb1=norm(g)clK}) and~(\ref{F: norm 1(E(Delta)f)=...})
completes the proof.
\end{proof}

\begin{lemma} \label{L: E(Delta) is a contraction}
  The operator $E_\Delta,$ considered as an operator from $\hilb_1(F)$ to $\hilb_1(\Delta),$ is a contraction.
  The image of this contraction coincides with $\hilb_1(\Delta).$
\end{lemma}
\begin{proof} Let $f$ be a unit vector from $\hilb_1(F)$
and let $\eps>0.$ By Lemma~\ref{L: norm(f)=lim norm(E(Delta)f} there
exists large enough $\Delta_2 \supset \Delta$ such that $
  \norm{E_{\Delta_2} f}_{\hilb_1(\Delta_2)} < 1 + \eps.
$ Since $E_{\Delta} f = E_{\Delta} E_{\Delta_2} f,$ and since, by
Proposition~\ref{P: E(Delta) is contraction}, $E_{\Delta}$ is a
contraction from $\hilb_1\brs{\Delta_2}$ to $\hilb_1\brs{\Delta},$ it
follows that
$$
  \norm{E_{\Delta} f}_{\hilb_1(\Delta)} < 1 + \eps.
$$
Since $\eps>0$ is arbitrary, it follows that $E_\Delta$ is a
contraction. Further, by (\ref{F: hilb1(Delta)=ran(F*)}) the set $\hilb_1(\Delta)$ is the image of
the operator $F_\Delta^* = E_\Delta F^*,$ while $\hilb_1(F)$ is
the image of $F^*.$ It follows that $E_\Delta \hilb_1(F)
= \hilb_1(\Delta).$
\end{proof}

\begin{prop} \label{P: H(-1)(Delta) incl H(-1)} Let $\Delta_1 \subset \Delta_2$ be two open bounded subsets of $\mbR.$
A natural inclusion
$$
  \hilb_{-1}(\Delta_1) \hookrightarrow \hilb_{-1}(\Delta_2), \ \ \ f_{\Delta_1} \mapsto f_{\Delta_2},
$$
given by formula
$$
  \scal{f_{\Delta_2}}{k_{\Delta_2}}_{\hilb_{-1}(\Delta_2),\hilb_{1}(\Delta_2)} = \scal{f_{\Delta_1}}{E_{\Delta_1} k_{\Delta_2}}_{\hilb_{-1}(\Delta_1),\hilb_{1}(\Delta_1)},
$$
where $k_{\Delta_2} \in \hilb_{1}(\Delta_2),$ is a contraction.
Further, for any bounded open $\Delta \subset \mbR$ the natural
inclusion $
  \hilb_{-1}(\Delta) \hookrightarrow \hilb_{-1}(F), \ f_{\Delta} \mapsto f,
$ given by formula
$$
  \scal{f}{k}_{-1,1} = \scal{f_{\Delta}}{E_{\Delta} k}_{-1,1},
$$
where $k \in \hilb_{1}(F),$ is also a contraction. Moreover, the
union $
  \bigcup_{\Delta} \hilb_{-1}(\Delta),
$ taken over all bounded open sets $\Delta \subset \mbR,$ is dense
in $\hilb_{-1}(F).$
\end{prop}
\begin{proof}
We have, by~(\ref{F: norm(f)(-alpha)}),
\begin{equation*}
  \begin{split}
  \norm{f_{{\Delta_2}}}_{\hilb_{-1}({\Delta_2})}
     & = \sup_{g_{{\Delta_2}}\colon \norm{g_{{\Delta_2}}}_{\hilb_{-1}({\Delta_2})} \leq 1} \abs{\scal{f_{{\Delta_2}}}{g_{{\Delta_2}}}_{\hilb_{-1}({\Delta_2})}}
   \\ & = \sup_{g_{{\Delta_2}}\colon \norm{g_{{\Delta_2}}}_{\hilb_{-1}({\Delta_2})} \leq 1} \abs{\scal{f_{{\Delta_2}}}{\abs{F_{{\Delta_2}}}^2g_{{\Delta_2}}}_{\hilb_{-1}({\Delta_2}),\hilb_{1}({\Delta_2})}}.
  \end{split}
\end{equation*}
Let $k_{{\Delta_2}} = \abs{F_{{\Delta_2}}}^2g_{{\Delta_2}} \in
\hilb_1({\Delta_2}).$ Since, by definition,
$\norm{g_{{\Delta_2}}}_{\hilb_{-1}({\Delta_2})} =
\norm{k_{{\Delta_2}}}_{\hilb_{1}({\Delta_2})},$ it follows that
\begin{equation*}
  \begin{split}
  \norm{f_{{\Delta_2}}}_{\hilb_{-1}({\Delta_2})}
     & = \sup_{k_{{\Delta_2}}\colon \norm{k_{{\Delta_2}}}_{\hilb_{1}({\Delta_2})} \leq 1} \abs{\scal{f_{{\Delta_2}}}{k_{{\Delta_2}}}_{\hilb_{-1}({\Delta_2}),\hilb_{1}({\Delta_2})}}
    \\ & = \sup_{k_{{\Delta_2}}\colon \norm{k_{{\Delta_2}}}_{\hilb_{1}({\Delta_2})} \leq 1} \abs{\scal{f_{{\Delta_1}}}{E_{\Delta_1} k_{{\Delta_2}}}_{\hilb_{-1}({\Delta_1}),\hilb_{1}({\Delta_1})}},
  \end{split}
\end{equation*}
where the last equality follows from the definition of
$f_{\Delta_2}.$ Since, by Proposition~\ref{P: E(Delta) is
contraction}, the operator $E_{\Delta_1}\colon \hilb_1(\Delta_2)
\to \hilb_1(\Delta_1)$ is a contraction, it follows
from the last inequality and~(\ref{F: norm(f)(-alpha)})
that $\norm{f_{{\Delta_2}}}_{\hilb_{-1}({\Delta_2})} \leq
\norm{f_{{\Delta_1}}}_{\hilb_{-1}({\Delta_1})}.$

Proof of the second part is similar, but instead of Proposition
\ref{P: E(Delta) is contraction} one has to use Lemma~\ref{L:
E(Delta) is a contraction}. The last assertion is obvious.
\end{proof}

\subsection{The set $\Lambda(H,F)$}
Given a fixed Hilbert-Schmidt rigging operator~$F,$ to every
self-adjoint operator~$H$ we can assign a set of full Lebesgue
measure $\Lambda(H,F),$ so that for every number~$\lambda$ from
$\Lambda(H,F)$ one can define the fiber Hilbert space (\ref{F: hlambda})
and the evaluation operator (\ref{F: euE def of}) by formula (\ref{F: euE(f)=sum beta(j)eta(j)}).
The evaluation operator $\euE_\lambda$ is an important
tool in the approach to abstract scattering theory, discussed here. In this section
we define and study the set $\Lambda(H,F)$ for a rigging operator~$F,$ which is not necessarily compact.

\begin{defn} \label{D: Lambda(H,F)} Let a self-adjoint operator~$H$ on a Hilbert space~$\hilb$ be compatible with a rigging operator $F \colon \hilb \to \clK.$
The set $\Lambda(H,F)$ of~$H$-regular points (or just regular, if there is
no danger of confusion) for the pair $(H,F)$ consists of all those
non-zero real numbers~$\lambda,$ such that there exists the norm limit
$$
  FR_{\lambda+i0}(H)F^* = \lim_{y \to 0^+}FR_{\lambda+iy}(H)F^*
$$
and such that for some (and thus for any) bounded open set
$\Delta,$ containing the point~$\lambda,$ the limit
$$
  F E_\Delta \Im R_{\lambda+i0}(H) F^* = \lim_{y \to 0^+}F E_\Delta \Im R_{\lambda+iy}(H)F^*
$$
exists in trace-class norm.
\end{defn}
\noindent If~$F$ is a Hilbert-Schmidt operator, then this
definition coincides with that of part I of \S~\ref{SS: HS rigging} (cf. also \cite[Definition 2.4.1]{Az3v6}).
Definition of the set $\Lambda(H,F)$ does not depend on the choice
of an open bounded set~$\Delta$ containing~$\lambda.$ Indeed, if
$\Delta_2 \ni \lambda$ is another such set, then~$\lambda$ belongs
to the resolvent set of $E_{\Delta\setminus \Delta_2}
R_{\lambda+iy}(H),$ so that the norm limit $E_{\Delta\setminus
\Delta_2} R_{\lambda+i0}(H)$ exists even without sandwiching
by~$F$ and $F^*.$ Using notation~(\ref{F: F(Delta)}), the second
condition of Definition~\ref{D: Lambda(H,F)} is equivalent to the
existence of the trace class norm limit of the operator
$$
  F_\Delta \Im R_{\lambda+iy}(H) F^*_\Delta
$$
as $y \to 0^+.$ It is also not difficult to see that the first
condition of Definition~\ref{D: Lambda(H,F)} implies that for any
bounded open set~$\Delta,$ containing the point~$\lambda,$ the
norm limit
$$
  F_\Delta R_{\lambda+i0}(H)F_\Delta ^* = \lim_{y \to 0^+} F_\Delta R_{\lambda+iy}(H)F_\Delta ^*
$$
exists. That is, if~$\lambda$ belongs to the set $\Lambda(H,F)$ in
the sense of Definition~\ref{D: Lambda(H,F)}, then for any bounded
open set~$\Delta$ containing~$\lambda,$ the point~$\lambda$
belongs to the set $\Lambda(H E_\Delta,F_\Delta)$ in the sense of
part I of \S~\ref{SS: HS rigging}, where $H E_\Delta$ and $F_\Delta$
are considered as operators on the Hilbert space $E_\Delta \hilb.$
In fact, if $\lambda \in \Lambda(H E_\Delta, F_\Delta)$ for some
$\Delta \in \euB_\lambda,$ then
$\lambda \in \Lambda(H,F)$ in the sense of Definition~\ref{D: Lambda(H,F)}, so that
\begin{equation} \label{F: Lambda(H,F)=bigcup...}
  \Lambda(H,F) = \bigcup_{\Delta \in \euB} \brs{\Lambda(H E_\Delta, F_\Delta) \cap \Delta}.
\end{equation}
Thus, Definition~\ref{D: Lambda(H,F)} is a natural extension of
the definition of the set of regular points $\Lambda(H,F)$ to the
case of generalized rigging operators~$F.$

\begin{prop} The set $\Lambda(H,F)$ has full Lebesgue measure.
\end{prop}
\begin{proof}
By Theorem~\ref{T: Abstract L.A.P.}, the set $\Lambda(HE_\Delta, F_\Delta)$ has full Lebesgue measure in~$\Delta.$ Further, as it was discussed above, the inclusion
\begin{equation} \label{F: obvious inclusion}
  \Lambda(HE_{\Delta_1}, F_{\Delta_1})  \cap {\Delta_1} \subset \Lambda(HE_{\Delta_2}, F_{\Delta_2})  \cap {\Delta_2},
\end{equation}
holds for any two bounded open subsets $\Delta_1 \subset \Delta_2$ of $\mbR.$ 
\end{proof}

We call a null Borel set $Z$ a \emph{core} of the singular spectrum of~$H,$ if the operator $E_{\mbR \setminus Z}^H H$ is absolutely continuous.
\begin{prop} \label{P: Lambda is a core} The complement of $\Lambda(H,F)$ is a core of singular spectrum of $H_0.$
\end{prop}
\begin{proof} It follows from~(\ref{F: Lambda(H,F)=bigcup...}) and~(\ref{F: obvious inclusion}) that
$$
  \mbR \setminus \Lambda(H,F) = \bigcap_{\Delta \in \euB} \brs{\mbR \setminus \Lambda(HE_\Delta,F_\Delta)}.
$$
By Theorem~\ref{T: H E(H) is a.c.}, for any bounded open set~$\Delta$ the set $\mbR \setminus
\Lambda(HE_\Delta,F_\Delta)$ is a core of singular spectrum of
$HE_\Delta.$ It follows that $\mbR \setminus \Lambda(H,F)$ is a
core of the singular spectrum of $H.$
\end{proof}

\subsection{A sheaf of fiber Hilbert spaces $\hlambda(H,F)$}
Let~$\euB$ be the family of all bounded open subsets of~$\Lambda(H,F)$
and let $\euB_\lambda$ be the family of all those bounded open
sets which contain a point~$\lambda \in \Lambda(H,F).$
Our next task is to construct the fiber Hilbert space $\hlambda.$ In
the case of non-compact rigging~$F,$ there is no canonical choice of
the fiber Hilbert space $\hlambda.$ Using the construction of the
fiber Hilbert space from part II of \S~\ref{SS: HS rigging} for the case of
Hilbert-Schmidt rigging~$F$, once some bounded open set $\Delta \in
\euB_\lambda$ is fixed, one can define a fiber Hilbert space
$\hlambda(\Delta),$ using the fact that the operator $F_\Delta =
FE_\Delta$ is Hilbert-Schmidt. Thus, instead of one fiber Hilbert
space $\hlambda,$ we get a family of Hilbert spaces
$\hlambda(\Delta).$ Fortunately, it turns out that all
these fiber Hilbert spaces are naturally isomorphic; in fact, they
form a sheaf of fiber Hilbert spaces at~$\lambda,$ in the sense of
subsection~\ref{SS: sheaf of f.H.s-es at lambda}. In order to
make the construction of the fiber Hilbert space $\hlambda$
natural, we treat the Hilbert space $\hlambda$ as a
sheaf of fiber Hilbert spaces at a point. This subsection is devoted to
the construction of this sheaf of fiber Hilbert spaces.

\medskip

Let $(H,F)$ be a compatible pair (see Definition~\ref{D: compatible (H,F)}).
For any bounded open set~$\Delta,$ the operator $F_\Delta$ is
Hilbert-Schmidt (see (\ref{F: F(Delta)}) and (\ref{F: F(Delta) = Schimidt ...})) and has zero kernel in $E_\Delta \hilb.$ Hence, the operator $F_\Delta$
can be considered as a Hilbert-Schmidt rigging in $E_\Delta \hilb.$
For any real number~$\lambda$ from the set
$\Lambda(HE_{\Delta}^{H},F_\Delta) \cap \Delta,$ the rigging
$F_\Delta$ generates the evaluation operator (see part IV of \S~\ref{SS: HS rigging})
$$
  \euE^\Delta_\lambda \colon \hilb_1(\Delta) \to \ell_2,
$$
This operator acts by formula
\begin{equation} \label{F: euD(D)(l)f=...}
  \euE^\Delta_\lambda f = \sum_{j=1}^\infty \beta^\Delta_j \eta^\Delta_j(\lambda),
\end{equation}
where $f$ is an element of the Hilbert space $\hilb_1(\Delta),$
which is given by formula~(\ref{F: elements of hilb1(Delta)TWO}), and
$\eta^\Delta_j(\lambda)$ is the $j$-th column of the
Hilbert-Schmidt matrix
\begin{equation} \label{F: eta(Delta)(l)}
  \eta^\Delta(\lambda) = \sqrt{\phi^\Delta(\lambda)},
\end{equation}
where
\begin{equation} \label{F: phi(Delta)(l)}
 \begin{split}
  \phi^\Delta(\lambda) & = \pi^{-1} \brs{\kappa_j^\Delta\kappa_k^\Delta\scal{\phi_j^\Delta}{\Im R_{\lambda+i0}(E_\Delta^{H}H)\phi_k^\Delta}_{1,-1}}_{j,k=1}^\infty
   \\ & = \pi^{-1} \brs{\scal{\psi_j^\Delta}{F_\Delta \Im R_{\lambda+i0}(E_\Delta^{H}H)F_\Delta^* \psi_k^\Delta}_{1,-1}}_{j,k=1}^\infty.
 \end{split}
\end{equation}
(Here the first pair of brackets $\scal{\cdot}{\cdot}_{1,-1}$ is the pairing of $\hilb_1(\Delta)$ and $\hilb_{-1}(\Delta)$).
In particular, the value $\euE_\lambda^\Delta \phi_j^\Delta = \phi_j^\Delta(\lambda)$ of the
vector $\phi_j^\Delta$ at~$\lambda$ is defined by
$$
  \euE_\lambda^\Delta \phi_j^\Delta = \frac 1{k_j^\Delta}\eta_j^\Delta(\lambda).
$$
The Hilbert space $\hlambda(\Delta)$ is defined as the closure of the image of the
evaluation operator
\begin{equation} \label{F: hlambda(Delta)}
  \hlambda(\Delta) = \overline{\euE^\Delta_\lambda \hilb_1(\Delta)},
\end{equation}
where $\hilb_1(\Delta)$ is the Hilbert space $\hilb_1(F_\Delta)$
of regular vectors of the Hilbert space $E_\Delta \hilb$ with
respect to the rigging~$F_\Delta$ (see (\ref{F: elements of hilb1(Delta)}) and (\ref{F: hilb1(Delta)=ran(F*)})).

\subsubsection{Natural isomorphisms $U_{\Delta_2,\Delta_1}(\lambda)$}
Let $\lambda \in \Lambda(H,F)$ and let $\Delta_1, \Delta_2 \in
\euB_\lambda.$ We define the gluing unitary operators
$U_{\Delta_2,\Delta_1}(\lambda)$ by the formula
\begin{equation} \label{F: U(D2,D1) def of}
  U_{\Delta_2,\Delta_1}(\lambda) \euE^{\Delta_1}_\lambda E_{\Delta_1} f := \euE^{\Delta_2}_\lambda E_{\Delta_2} f,
\end{equation}
where $f \in \hilb_1(F).$ Note that by~(\ref{F:
hilb1(Delta)=ran(F*)}) vectors $E_{\Delta_j} f,$ $f \in
\hilb_1(F),$ exhaust the Hilbert space $\hilb_1(\Delta_j),$ and by
(\ref{F: hlambda(Delta)}) and (\ref{F: hilb1(Delta)=ran(F*)}),
the Hilbert space $\hlambda(\Delta_j)$ is the
completion of the linear manifold $$\euE^{\Delta_j}_\lambda E_{\Delta_j} \rng(F^*) = \euE_\lambda^{\Delta_j} \hilb_1(\Delta_j).$$
It follows that it is sufficient to define $U_{\Delta_2,\Delta_1}(\lambda)$ on this image.

\begin{prop} \label{P: U(Delta1,Delta2) is unitary} The operator $U_{\Delta_2,\Delta_1}(\lambda) \colon \hlambda(\Delta_1) \to \hlambda(\Delta_2),$
defined by formula~(\ref{F: U(D2,D1) def of}), is unitary.
\end{prop}
\begin{proof} What we have to show is that for any $f \in \hilb_1(F)$
the norm of the vector $\euE^{\Delta}_\lambda E_{\Delta} f$ of the
Hilbert space $\hlambda(\Delta)$ does not depend on the choice of
a bounded open subset $\Delta \subset \Lambda(H,F)$ as long as it
contains~$\lambda.$ So, let~$\Delta$ be such a set and let
$(\beta_j^\Delta)$ be an $\ell_2$-sequence such that (see (\ref{F: elements of hilb1(Delta)TWO}))
$$
  E_{\Delta} f = \sum_{j=1}^\infty \beta^\Delta_j \kappa^\Delta_j \phi^\Delta_j.
$$
By definition of the scalar product in $\hlambda(\Delta) \subset \ell_2,$ we have
$$
  \norm{\euE^{\Delta}_\lambda E_{\Delta} f}^2_{\hlambda(\Delta)}
    = \sum_{i=1}^\infty \sum_{j=1}^\infty  \bar \beta^\Delta_i \beta^\Delta_j \kappa^\Delta_i \kappa^\Delta_j\scal{\euE_\lambda^\Delta \phi^\Delta_i}{\euE_\lambda^\Delta\phi^\Delta_j},
$$
where the double series converges absolutely (see the remark after
formula \cite[(3.1)]{Az3v6}). Hence, using \cite[(5.5)]{Az3v6} we obtain
\begin{equation} \label{F: terrible formula}
 \begin{split}
  \norm{\euE^{\Delta}_\lambda E_{\Delta} f}^2_{\hlambda(\Delta)}
   & = \frac 1 \pi \sum_{i=1}^\infty \sum_{j=1}^\infty  \bar \beta^\Delta_i \beta^\Delta_j \kappa^\Delta_i \kappa^\Delta_j \scal{\phi^\Delta_i}{\Im R_{\lambda+i0}(E^{H_0}_\Delta H) \phi^\Delta_j}_{1,-1}
  \\ & = \frac 1 \pi \sum_{i=1}^\infty \sum_{j=1}^\infty  \bar \beta^\Delta_i \beta^\Delta_j \kappa^\Delta_i \kappa^\Delta_j \lim_{y \to 0^+} \scal{\phi^\Delta_i}{\Im R_{\lambda+iy}(E^{H_0}_\Delta H) \phi^\Delta_j}
 \end{split}
\end{equation}
Since $\lambda \in \Lambda(H,F),$ the trace-class valued matrix
$$
  \phi^\Delta_{\lambda+iy} = \frac 1 \pi \brs{\kappa^\Delta_i \kappa^\Delta_j \scal{\phi^\Delta_i}{\Im R_{\lambda+iy}(E^{H_0}_\Delta H) \phi^\Delta_j}}_{i,j=1}^\infty
$$
continuously depends on $y$ up to $y=0$ (see part II of \S~\ref{SS: HS rigging}). So, the function $
  \scal{\beta}{\phi^\Delta_{\lambda+iy}\beta}
$ is also continuous up to $y=0.$ Therefore, we can interchange
the summations and the limit operation in~(\ref{F: terrible
formula}) to get
\begin{equation*}
 \begin{split}
    \norm{\euE^{\Delta}_\lambda E_{\Delta} f}^2_{\hlambda(\Delta)} & = \frac 1 \pi \lim_{y \to 0^+} \sum_{i=1}^\infty \sum_{j=1}^\infty
       \bar \beta^\Delta_i \beta^\Delta_j \kappa^\Delta_i \kappa^\Delta_j \scal{\phi^\Delta_i}{\Im R_{\lambda+iy}(E^{H_0}_\Delta H) \phi^\Delta_j}
  \\ & = \frac 1 \pi \lim_{y \to 0^+} \scal{E_\Delta f}{\Im R_{\lambda+iy}(E^{H_0}_\Delta H) E_\Delta f}.
 \end{split}
\end{equation*}
The last expression does not depend on~$\Delta$ as long as it
contains~$\lambda.$
\end{proof}

Definition (\ref{F: U(D2,D1) def of}) of operators $U_{\Delta_2,\Delta_1}(\lambda)$ implies that they satisfy the
gluing property~(\ref{F: gluing property for U}). Thus, to any
compatible pair $(H,F)$ and a number $\lambda \in \Lambda(H,F)$ we
can assign a sheaf of fiber Hilbert spaces
\begin{equation} \label{F: hlambda(H,F)}
  \hlambda(H,F) = \set{\hlambda(\Delta), \Delta \in \euB_\lambda},
\end{equation}
where the Hilbert spaces $\hlambda(\Delta)$ are defined by
(\ref{F: hlambda(Delta)}) and the gluing unitary isomorphisms
$U_{\Delta_2,\Delta_1}(\lambda)$ are defined by~(\ref{F: U(D2,D1) def of}).

\subsection{Sheaf $\euS(H,F)$}
By \cite[Corollary 3.1.8]{Az3v6}, for any bounded open set
$\Delta\subset \mbR$ the family of Hilbert spaces
$\set{\hlambda(\Delta), \lambda \in \Delta}$ is measurable. It
follows from the definition~(\ref{F: U(D2,D1) def of}) of the
gluing unitary operators $U_{\Delta_2,\Delta_1}(\lambda)$ that for
any sets $\Delta_1, \Delta_2$ from $\euB,$ the family of operators
$\set{U_{\Delta_2,\Delta_1}(\lambda), \lambda \in \Delta_1 \cap
\Delta_2}$ is also measurable. Thus, for any compatible pair
$(H,F)$ we have assigned a sheaf of Hilbert spaces
over~$\Lambda(H,F).$ We shall denote this sheaf of Hilbert spaces
by $\euS(H,F).$ We shall also use the notions and notation
associated with a sheaf of Hilbert spaces, such as $\euH(\Delta),$
etc, introduced in \S\,\ref{S: Sheaves of Hilbert spaces}.

\subsection{Evaluation operator $\euE_\lambda(H)$}
Let $\lambda \in \Lambda(H,F).$ For any $f \in \hilb_1(F)$ and for
any bounded open subset~$\Delta$ of $\Lambda(H,F),$ such that
$\lambda \in \Delta,$ we have a vector
\begin{equation} \label{F: f(D)(l) def of}
  f_\Delta(\lambda) := \euE^\Delta_\lambda E_\Delta f \in \hlambda(\Delta).
\end{equation}
It follows from~(\ref{F: U(D2,D1) def of}) that the family
$$
  f(\lambda) := \set{f_\Delta(\lambda) \colon \Delta \in \euB_\lambda}
$$
correctly defines an element of the sheaf of fiber Hilbert spaces (\ref{F: hlambda(H,F)}).
Further, by definition~(\ref{F: T(D2)=U(D2,D1)T(D1)}), the formula
\begin{equation} \label{F: euE(l)=set(f(D):D in euB(l))}
  \euE_\lambda(H) f = \set{f_\Delta(\lambda) \colon \Delta \in \euB_\lambda}
\end{equation}
combined with (\ref{F: U(D2,D1) def of}) correctly defines an operator
\begin{equation} \label{F: euE(l): hilb1(F) to hlambda}
  \euE_\lambda(H) \colon \hilb_1(F) \to \hlambda(H,F).
\end{equation}
We often omit the argument $H$ in $\euE_\lambda(H),$ if there is no danger of confusion.
According to (\ref{F: scal prod in sheaf}), for any $f,g \in \hilb_1(F)$ and any $\Delta \in \euB_\lambda$
\begin{equation} \label{F: (euEl(f),euEl(g))=(euEl(D)(f),euEl(D)(g))}
  \scal{\euE_\lambda f}{\euE_\lambda g}_{\hlambda(H,F)} = \scal{\euE_\lambda^\Delta E_\Delta f}{\euE_\lambda^\Delta E_\Delta g}_{\hlambda(\Delta)}.
\end{equation}
\begin{prop} \label{P: euE is HS} The operator~(\ref{F: euE(l): hilb1(F) to hlambda}),
defined by equality (\ref{F: euE(l)=set(f(D):D in euB(l))}), is Hilbert-Schmidt.
\end{prop}
\begin{proof}
Let $f_1,f_2,\ldots$ be an orthonormal basis of $\hilb_1(F).$ For
any $\Delta \in \euB_\lambda,$ we have
$$
  \norm{\euE_\lambda}_{\clL_2(\hilb_1(F),\hlambda)}^2 = \sum_{j=1}^\infty \norm{\euE_\lambda f_j}_{\hlambda}^2 = \sum_{j=1}^\infty \norm{\euE^\Delta_\lambda E_\Delta f_j}_{\hlambda(\Delta)}^2,
$$
where the second equality follows from~(\ref{F: (euEl(f),euEl(g))=(euEl(D)(f),euEl(D)(g))}). By Lemma~\ref{L: E(Delta) is a contraction}, the operator
$E_\Delta,$ considered as an operator $\hilb_1(F) \to
\hilb_1(\Delta),$ is bounded, and the operator
$$
  \euE^\Delta_\lambda \colon \hilb_1(\Delta) \to \hlambda(\Delta)
$$ is Hilbert-Schmidt (see part IV of \S\,\ref{SS: HS rigging} and/or \cite[\S 3.1]{Az3v6}). It follows that the composition $\euE^\Delta_\lambda E_\Delta$
is also Hilbert-Schmidt and therefore the sum above is finite. It
follows that $\euE_\lambda$ is also Hilbert-Schmidt.
\end{proof}
We say that two vectors $f$ and $g$ from $\hilb_1(F)$ are
\emph{equivalent at~$\lambda,$} if $E_\Delta f = E_\Delta g$ for
some $\Delta \in \euB_\lambda.$ Obviously, equivalence
at~$\lambda$ is an equivalence relation. A jet of vectors
at~$\lambda$ is an equivalence class with respect to this
equivalence relation and the set of all jets of vectors at~$\lambda$
is a linear space. As can be seen from the definition of the
operator $\euE_\lambda,$ it is in fact defined on the space of
jets of vectors at~$\lambda.$

\subsubsection{Operator $\euE$}
For any bounded open subset~$\Delta$ of
$\Lambda(H,F)$ we define the operator
\begin{equation} \label{F: euE(D)}
  \euE^\Delta \colon \hilb_1(\Delta) \to \euH(\Delta), \quad \euE^\Delta = \int_\Delta ^\oplus \euE^\Delta_\lambda \, d\lambda,
\end{equation}
where $\euE^\Delta_\lambda$ is defined by (\ref{F: euD(D)(l)f=...}).
\begin{lemma} \label{L: euE(d) diag-zes H E(D)} The operator $\euE^\Delta$ defined by (\ref{F: euE(D)}) and considered as an operator from $E_\Delta \hilb$ to $\euH(\Delta)$
is bounded and unitary. Moreover, for any $f_\Delta \in \hilb_1(\Delta),$ and for a.e. $\lambda \in \Lambda.$
$$
  [\euE^\Delta (Hf_\Delta)](\lambda) = \lambda\euE^\Delta_\lambda (f_\Delta).
$$
\end{lemma}
\begin{proof}
Comparing definitions (\ref{F: euD(D)(l)f=...}), (\ref{F: eta(Delta)(l)}) and (\ref{F: phi(Delta)(l)}) with definitions
(\ref{F: euE(f)=sum beta(j)eta(j)}), (\ref{F: eta(l+iy)}) and (\ref{F: phi(l+iy)}), we infer
that the operator $\euE^\Delta$ is the evaluation operator (see parts IV-VI of \S\ref{SS: HS rigging}) corresponding to self-adjoint
operator $HE_\Delta$ on the Hilbert space $E_\Delta\hilb$ with Hilbert-Schmidt rigging $F_\Delta=FE_\Delta.$
Therefore, the required assertion is a direct corollary of Theorem~\ref{T: euE}.
\end{proof}

Now we give definition of the evaluation operator
$$
  \euE \colon \hilb \to \euS(H,F).
$$
We shall do it in two different ways, mentioned in subsection~\ref{SSS: op-rs on sheaf euS}.
It is not difficult to see that the family of operators
$$
  \set{\euE_\lambda \in \clL_2(\hilb_1(F), \hlambda(H,F)) \colon \lambda \in \Lambda(H,F)},
$$
which are defined by~(\ref{F: euE(l): hilb1(F) to hlambda}), is measurable. Thus, we can define the operator
$$
  \euE = \int_\Lambda ^\oplus \euE_\lambda \, d\lambda,
$$
which acts from the dense subspace $\hilb_1(F)$ of $\hilb$ to the sheaf of Hilbert spaces $\euS(H,F).$
Further, comparing definition (\ref{F: euE(l)=set(f(D):D in euB(l))}) and (\ref{F: f(D)(l) def of}) of $\euE_\lambda$
and definition (\ref{F: euE(D)}) of $\euE^\Delta$ we infer that for any $f \in \hilb_1(F)$
\begin{equation} \label{F: (euE f)Delta=euE Delta E(Delta)f}
  (\euE f)_\Delta = \euE^\Delta E_\Delta f,
\end{equation}
where $E_\Delta$ is considered here as an operator acting from
$\hilb_1(F)$ onto $\hilb_1(\Delta)$ (see Lemma~\ref{L: E(Delta) is a contraction}). Since $\hilb_1(F)$
is naturally embedded into~$\hilb,$ we can consider $\euE$ as an
operator from~$\hilb$ to $\euS(H,F).$ The equality (\ref{F: (euE f)Delta=euE Delta E(Delta)f}) gives a second equivalent definition of $\euE.$
Finally, the collection of mappings
$$
  \set{\euE^\Delta_\lambda \colon \hilb_1(\Delta) \to \hlambda(\Delta) \ | \ \lambda \in \Lambda(H,F), \Delta \in \euB_\lambda},
$$
defined by (\ref{F: euD(D)(l)f=...}) and combined with gluing property (\ref{F: U(D2,D1) def of})  can also be considered as a definition of $\euE.$

\begin{thm} \label{T: op-r euE} The operator $\euE$ defines a bounded operator from~$\hilb$ to the sheaf of Hilbert spaces $\euS(H,F).$
The operator $\euE$ vanishes on singular (with respect to~$H$)
subspace of~$\hilb$ and it is a unitary isomorphism of the
absolutely continuous (with respect to~$H$) part of~$\hilb$ onto
the sheaf $\euS(H,F).$ This mapping diagonalizes the absolutely
continuous part of the operator~$H,$ that is, for any $f \in \dom(H)$ and for a.e. $\lambda \in \Lambda(H,F)$ the equality
$\euE_\lambda(Hf) = \lambda \euE_\lambda(f)$ holds.
\end{thm}
\begin{proof} 
For any $f \in \hilb_1(F),$ it follows from~(\ref{F:
S(H,F) norm}),~(\ref{F: (euE f)Delta=euE Delta E(Delta)f}) and Theorem~\ref{T: euE}, that
\begin{equation} \label{F: norm of euE f in S}
 \begin{split}
  \norm{\euE f}_{\euS(H,F)} & = \sup_{\Delta \in \euB} \norm{(\euE f)_\Delta}_{\euH(\Delta)} = \sup_{\Delta \in \euB} \norm{\euE^\Delta E_\Delta f}_{\euH(\Delta)}
    \\ & = \sup_{\Delta \in \euB} \norm{P^{(a)}E_\Delta f}_{E_\Delta \hilb} = \norm{P^{(a)}f}_{\hilb}.
 \end{split}
\end{equation}
Since $\hilb_1(F)$ is densely included in~$\hilb,$ it follows that $\euE$ can be continued
to a bounded operator on~$\hilb,$ which is a partial
isometry with initial space $P^{(a)}\hilb.$
By definition~(\ref{F: (euE f)Delta=euE Delta
E(Delta)f}) of $\euE,$ on the subspace $E_\Delta \hilb$ the
operator $\euE$ coincides with the unitary operator
$\euE^\Delta\colon E_\Delta \hilb \to \euH(\Delta).$ By Lemma~\ref{L: euE(d) diag-zes H E(D)}, the unitary operator $\euE^\Delta \colon E_\Delta \hilb \to \euH(\Delta)$
diagonalizes the self-adjoint operator $E_\Delta H,$ that is,
$\euE^\Delta(E_\Delta H f) = \int_\Delta^\oplus
\lambda\euE^\Delta(E_\Delta f)\,d\lambda.$ It follows from this
and~(\ref{F: (euE f)Delta=euE Delta E(Delta)f}) that for any $f
\in \hilb$
$$
  \euE(Hf) = \int_{\Lambda(H,F)}^\oplus \euE_\lambda (Hf) \,d\lambda = \int_{\Lambda(H,F)}^\oplus \lambda\euE_\lambda f \,d\lambda.
$$
Since the image of $\euE^\Delta$ coincides with $\euH(\Delta),$ it
also follows from the last equality that the operator $\euE \colon
\hilb \to \euS(H,F)$ is onto.
\end{proof}
\begin{cor} \label{C: euE is diagonalizing}
  If~$h$ is a bounded measurable function on~$\mbR$ and $f \in \hilb,$ then for a.e. $\lambda \in \Lambda(H,F)$ we have $\euE_\lambda(h(H)f) = h(\lambda)\euE_\lambda(f).$
\end{cor}

\subsection{Green operator $R_{\lambda+i0}(H)$}

If a self-adjoint operator $H$ is compatible with a rigging~$F,$ then by Lemma~\ref{L: FRF* is compact} the operator $FR_z(H)F^*$ is compact.
By definition, the norm limit
$
  FR_{\lambda+i0}(H)F^*
$
exists for all~$\lambda$ from the set of full Lebesgue
measure~$\Lambda(H,F).$ Thus, $FR_{\lambda+i0}(H)F^*$ is a compact
operator on the Hilbert space~$\clK$ for all $\lambda \in
\Lambda(H,F).$ Since the operator~$F$ (respectively, $F^*$) can be
seen as a natural isomorphism of Hilbert spaces
$\hilb_{-1}(F)$ and~$\clK$ (respectively,~$\clK$ and $\hilb_1(F)$), the operator
$FR_{\lambda+i0}(H)F^*$ can be interpreted as a compact limit
\begin{equation} \label{F: R(l+i0): hilb(1)to hilb(-1)}
  R_{\lambda+i0}(H) \colon \hilb_1 \to \hilb_{-1}
\end{equation}
of the operator $R_{\lambda+i0}(H) \colon \hilb_1 \to \hilb_{-1}$ in the norm of $\clB(\hilb_1,\hilb_{-1}).$
In this subsection we study properties of the operator (\ref{F: R(l+i0): hilb(1)to hilb(-1)}).

We denote by $R^\Delta_{\lambda+iy}(H)$ the operator $E_\Delta
R_{\lambda+iy}(H).$
The operator $R^\Delta_{\lambda+iy}(H)$ can be considered as an
operator on the Hilbert space $E_\Delta$ with Hilbert-Schmidt
rigging~$F_\Delta.$ So, there exists the limit operator
$R^\Delta_{\lambda+i0}(H)$ which acts from $\hilb_1(\Delta)$ to
$\hilb_{-1}(\Delta),$ provided that $\lambda \in \Lambda(H,F) \cap
\Delta.$
The operator $R^\Delta_{\lambda+i0}(H)$ can also be interpreted as
an operator from $\hilb_1(F)$ to $\hilb_{-1}(F)$
as a composition
$$
  \hilb_{-1}(F) \hookleftarrow \hilb_{-1}(\Delta) \stackrel{R^\Delta_{\lambda+i0}(H)}{\leftarrow\!\!\!-\!\!\!-\!\!\!-\!\!\!-\!\!\!-\!\!\!-}
  \hilb_{1}(\Delta) \stackrel{E_\Delta}{\leftarrow\!\!\!-\!\!\!-} \hilb_{1}(F).
$$
\noindent
Further, the operator $R^\Delta_{\lambda+i0}(H)$ can also be naturally defined as the norm
limit of the operator $R^\Delta_{\lambda+iy}(H) \colon \hilb_1(F)
\to \hilb_{-1}(F).$ In the following lemma we show that these two interpretations are identical.
\begin{lemma} \label{L: lim R(Delta) exists} Let $\lambda \in \Lambda(H,F).$ For any $\Delta \in \euB_\lambda$ the limit $R^\Delta_{\lambda+i0}(H)$
of the cut off resolvent $R^\Delta_{\lambda+iy}(H)$ as $y \to 0$
exists in the norm topology of $\clB(\hilb_{1},\hilb_{-1}).$
Further, 
the norm limit
$R^\Delta_{\lambda+i0}(H) \in \clB(\hilb_{1},\hilb_{-1})$ above
coincides with the composition of the contraction $E_\Delta \colon
\hilb_1(F) \to \hilb_1(\Delta),$ the norm limit
$R^\Delta_{\lambda+i0}$ of $R^\Delta_{\lambda+iy}$ as an operator
from $\hilb_1(\Delta)$ to $\hilb_{-1}(\Delta)$ and the inclusion
operator $\hilb_{-1}(\Delta) \hookrightarrow \hilb_{-1}(F).$
\end{lemma}
\begin{proof}
For any $f,g \in \hilb_1(F)$ we have
\begin{equation} \label{F: scal in H(1,-1)(F)=scal in H(1,-1)(Delta)}
  \begin{split}
    \scal{f}{R^\Delta_{\lambda+iy}(H)g}_{\hilb_1(F),\hilb_{-1}(F)} & = \scal{f}{R^\Delta_{\lambda+iy}(H)g}_{\hilb}
    \\ & = \scal{E_\Delta f}{R^\Delta_{\lambda+iy}(H)E_\Delta g}_{E_\Delta \hilb}
    \\ & = \scal{E_\Delta f}{R^\Delta_{\lambda+iy}(H) E_\Delta g}_{\hilb_1(\Delta),\hilb_{-1}(\Delta)},
  \end{split}
\end{equation}
where according to Lemma~\ref{L: norm(f)=lim norm(E(Delta)f} in the last equality we can consider vectors $E_\Delta f$ and
$E_\Delta g$ as elements of the Hilbert space $\hilb_1(\Delta).$
The right hand side of (\ref{F: scal in H(1,-1)(F)=scal in H(1,-1)(Delta)}) has a limit as $y\to 0^+,$ since according to (\ref{F: Lambda(H,F)=bigcup...})
and (\ref{F: obvious inclusion}) we have $\lambda \in \Lambda(HE_\Delta, F_\Delta).$ The left hand side of (\ref{F: scal in H(1,-1)(F)=scal in H(1,-1)(Delta)}) has a limit,
since for $\lambda \in \Lambda(H,F)$ the limit $\lim_{y \to 0} R_{\lambda+iy}(H)$ exists in the norm of $\clB(\hilb_1,\hilb_{-1})$ and the limit
$\lim_{y \to 0} R^{\mbR\setminus \Delta}_{\lambda+iy}(H)$ exists even in the norm of $\clB(\hilb).$
Taking these limits we obtain
$$
  \scal{f}{R^\Delta_{\lambda+i0}(H)g}_{\hilb_1(F),\hilb_{-1}(F)} = \scal{E_\Delta f}{R^\Delta_{\lambda+i0}(H) E_\Delta g}_{\hilb_1(\Delta),\hilb_{-1}(\Delta)}.
$$
This equality completes the proof.
\end{proof}

Next we show that the resolvent $R_{\lambda+i0}(H)$ can be treated as the limit of
the operator $R^\Delta_{\lambda+i0}(H)$ as $\Delta \to \mbR.$
\begin{prop} \label{P: Delta to R}
If $H$ is a self-adjoint operator on a Hilbert space $\hilb$ compatible with a rigging $F$
and if $\lambda \in \Lambda(H,F),$ then there holds the equality
$$R_{\lambda+i0}(H) = \lim\limits_{\Delta \nearrow \mbR} R^\Delta_{\lambda+i0}(H),$$
where the limit is taken in norm topology of $\clB(\hilb_1,\hilb_{-1}).$
\end{prop}
\noindent Proof of this proposition is the same as that of Lemma~\ref{L: FRF* is compact} and therefore is omitted.

The operator $\Im R_{\lambda+i0}(H) \colon \hilb_1(F) \to
\hilb_{-1}(F)$ is compact as a difference of two compact
operators. We now study some additional properties of this operator.

\begin{lemma} \label{L: Im R=Im R(Delta)}
If $H$ is a self-adjoint operator on a rigged Hilbert space $(\hilb,F)$ which is compatible with a rigging~$F,$
if $\lambda \in \Lambda(H,F)$ and if $\Delta \in \euB_\lambda,$ then for any $f,g \in \hilb_1(F)$
\begin{equation} \label{F: (f,Im Rg)=(f,Im R(D)g)}
  \scal{f}{\Im R_{\lambda+i0}(H) g}_{1,-1} = \scal{f}{\Im R^\Delta_{\lambda+i0}(H) g}_{1,-1}.
\end{equation}
\end{lemma}
\begin{proof} It follows from the first equality of~(\ref{F: scal in H(1,-1)(F)=scal in H(1,-1)(Delta)}) that
the right hand side of (\ref{F: (f,Im Rg)=(f,Im R(D)g)}) does not depend on
$\Delta \in \euB_\lambda.$ This observation, combined with
Proposition~\ref{P: Delta to R}, completes the proof.
\end{proof}

\begin{prop} \label{P: delta(H)=euE*euE}
If $H$ is a self-adjoint operator on a Hilbert space $\hilb$ and is compatible with a rigging $F$
and if $\lambda \in \Lambda(H,F),$ then
$$
  \frac 1\pi \Im R_{\lambda+i0}(H) = \euE^\diamondsuit_\lambda(H)\euE_\lambda(H)
$$
as the equality in $\clB(\hilb_1(F),\hilb_{-1}(F)).$ In particular, the operator $\Im R_{\lambda+i0}(H)$ is trace
class.
\end{prop}
\begin{proof} It is enough to show that for any $f,g \in \hilb_1(F)$
the equality
$$
  \frac 1\pi \scal{f}{\Im R_{\lambda+i0}(H) g}_{1,-1} = \scal{f}{\euE^\diamondsuit_\lambda\euE_\lambda g}_{1,-1}
$$
holds.
For any $\Delta \in \euB_\lambda,$ by (\ref{F: (euEl(f),euEl(g))=(euEl(D)(f),euEl(D)(g))}) and \cite[(5.5)]{Az3v6},
the right hand side of this equality is equal to
\begin{equation*}
  \begin{split}
    \scal{\euE_\lambda f}{\euE_\lambda g}_{\hlambda} & = \scal{\euE^\Delta_\lambda E_\Delta f}{\euE^\Delta_\lambda E_\Delta g}_{\hlambda(\Delta)}
   \\ & = \frac 1 \pi \scal{E_\Delta f}{\Im R^\Delta_{\lambda+i0}(H) E_\Delta g}_{\hilb_1(\Delta),\hilb_{-1}(\Delta)},
  \end{split}
\end{equation*}
where
$\Im R^\Delta_{\lambda+i0}(H)$ is considered as an operator from
$\hilb_1(\Delta)$ to $\hilb_{-1}(\Delta).$ Thus, it follows from
(\ref{F: scal in H(1,-1)(F)=scal in H(1,-1)(Delta)}) that
$$
  \scal{f}{\euE^\diamondsuit_\lambda\euE_\lambda g}_{1,-1} = \frac 1 \pi \scal{f}{\Im R^\Delta_{\lambda+i0}(H) g}_{1,-1},
$$
where now $\Im R^\Delta_{\lambda+i0}(H) \colon \hilb_1(F) \to
\hilb_{-1}(F).$ Combining this with Lemma~\ref{L: Im R=Im
R(Delta)} completes the proof.
\end{proof}

\section{Wave matrix}
\label{S: wave matrix}
So far we have considered a single self-adjoint operator~$H$ on a
Hilbert space~$\hilb$ with a fixed rigging operator~$F,$ and for
this reason we did not use sub-indices. From now on we are going
to consider not only the operator~$H$ but also its perturbations.
We denote by~$H_0$ an ``initial'' self-adjoint operator, and by
$H_1 = H_0+V$ we denote its perturbation by a self-adjoint
operator $V.$
We assume that the operator~$V$ admits a decomposition
\begin{equation} \label{F: V=F^*JF}
  V =F^*JF,
\end{equation}
where $J$ is a bounded operator on~$\clK.$ Existence of
such a factorization allows to treat~$V$ as a bounded operator (see e.g. discussion in introduction)
\begin{equation} \label{F: V:hilb(-1)to hilb(1)}
  V \colon \hilb_{-1}(F) \to \hilb_{1}(F).
\end{equation}
\subsection{Operators $\mathfrak a_\pm(\lambda; H_1,H_0)$}
We shall assume that~$H_0$ and~$H_1$ are two self-adjoint operators on a rigged Hilbert space $(\hilb,F),$
compatible with the rigging $F$ and such that the difference $V=H_1-H_0$ admits decomposition (\ref{F: V=F^*JF}).
This means that $V$ can be considered as a bounded operator (\ref{F: V:hilb(-1)to hilb(1)}) and we shall do this without further reference.
Whether $V$ is considered as an operator $\hilb \to \hilb$ or as an operator (\ref{F: V:hilb(-1)to hilb(1)}) will be clear from the context.

With the preparations given in previous sections, for any real
number~$\lambda,$ which belongs to the set $\Lambda(H_0,F) \cap
\Lambda(H_1,F),$ we can now define operators
$$
  \mathfrak a_\pm(\lambda; H_1,H_0) \colon \hilb_{1}(F) \to \hilb_{-1}(F),
$$
which are analogues of the forms \cite[(2.7.4)]{Ya}. We have (cf. e.g.~\cite[(2.7.10)]{Ya}, see also \cite[(5.3)]{Az3v6})
\begin{equation} \label{F: Ya (2.7.10)}
 \begin{split}
  \frac y \pi R_{\lambda\mp iy}(H_1) R_{\lambda\pm iy}(H_0)
        & = \frac 1 \pi \Im R_{\lambda+iy}(H_1) \SqBrs{1 + V R_{\lambda\pm iy}(H_0)}
     \\ & = \SqBrs{1 - R_{\lambda\mp iy}(H_1)V} \cdot \frac 1 \pi \Im R_{\lambda+iy}(H_0).
 \end{split}
\end{equation}
Since $\lambda \in \Lambda(H_0,F) \cap \Lambda(H_1,F),$ by Definition~\ref{D: Lambda(H,F)} the limits
\begin{equation} \label{D: of mathfrak a pm}
 \begin{split}
    \mathfrak a_\pm(\lambda; H_1,H_0)  & := \frac 1 \pi \Im R_{\lambda+i0}(H_1) \SqBrs{1 + V R_{\lambda\pm i0}(H_0)}
     \\ & = \SqBrs{1 - R_{\lambda\mp i0}(H_1)V} \cdot \frac 1 \pi \Im R_{\lambda+i0}(H_0)
 \end{split}
\end{equation}
exist in norm topology of the space $\clB(\hilb_1,\hilb_{-1}).$
Since the resolvent operator $R_{\lambda+iy}(H)$ is
compact as an operator from $\clB(\hilb_1,\hilb_{-1}),$ it follows
from the definition of the operators $\mathfrak a_\pm(\lambda;
H_1,H_0),$ that they are also compact.

\subsection{Wave matrix $w_\pm(\lambda; H_1,H_0)$}
\label{SS: wave matrix} Let $\lambda \in \Lambda(H_0,F) \cap
\Lambda(H_1,F).$ Following (\ref{F: def of w(+-) from intro}) (see also \cite[Definition 5.2.1]{Az3v6}), we
define the wave matrix $w_\pm(\lambda; H_1,H_0)$ as a form
$$
  w_\pm(\lambda; H_1,H_0) \colon \hlambda(H_1) \times \hlambda(H_0) \to \mbC,
$$ by the formula
\begin{equation} \label{F: def of w +-}
  w_\pm(\lambda; H_1,H_0) \brs{\euE_\lambda(H_1)f, \euE_\lambda(H_0) g} =  \scal{f}{\mathfrak a_\pm(\lambda;H_1,H_0)g}_{1,-1},
\end{equation}
where $f,g \in \hilb_1(F).$ This definition is exactly the same as
the definition of the wave matrix for the case of trace-class
perturbations~$V,$ considered in \cite[\S 5.2]{Az3v6}; the only difference
being the treatment of the fiber Hilbert space $\hlambda$ as a sheaf
of Hilbert spaces (\ref{F: hlambda(H,F)}). The wave matrix $w_\pm(\lambda; H_1,H_0)$ is
correctly defined by equality (\ref{F: def of w +-}) in the sense that if $f'$ and $g'$ is another
pair of vectors from $\hilb_1(F)$ such that
$\euE_\lambda(H_1)f = \euE_\lambda(H_1)f'$ and $\euE_\lambda(H_0)
g = \euE_\lambda(H_0) g',$ then
$$
  w_\pm(\lambda; H_1,H_0) \brs{\euE_\lambda(H_1)f, \euE_\lambda(H_0) g} = w_\pm(\lambda; H_1,H_0) \brs{\euE_\lambda(H_1)f', \euE_\lambda(H_0) g'},
$$
as it follows from Proposition~\ref{P: delta(H)=euE*euE} and
definition (\ref{D: of mathfrak a pm}) of $\mathfrak a_\pm(\lambda;H_1,H_0).$

Proof of the following proposition follows verbatim the proof of
\cite[Proposition 5.2.2]{Az3v6} (the idea of which was taken in its turn from \cite[\S 5.2]{Ya}),
with reference to Proposition~\ref{P: delta(H)=euE*euE} instead of \cite[(5.5)]{Az3v6}.
\begin{prop} \label{P: w(pm) is well-defined}
Let~$H_0$ and~$H_1$ be two self-adjoint operators on a rigged Hilbert space $(\hilb,F)$ which are compatible with the rigging~$F,$ such that $V=H_1-H_0$
admits decomposition (\ref{F: V=F^*JF}). For any $\lambda \in \Lambda\brs{H_1,F} \cap \Lambda\brs{H_0,F}$ the form
$w_\pm(\lambda; H_1,H_0)$ is well-defined, and it is bounded with norm $\leq 1.$
\end{prop}
\begin{proof} That $w_\pm(\lambda; H_1,H_0)$ is well-defined has already been shown.

Further, by Schwarz inequality, for any $f,g \in \hilb_1(F) \subset \hilb$
\begin{equation} \label{F: norm of mathfrak a}
 \begin{split}
  \frac y \pi & \abs{\scal{f}{ R_{\lambda-iy}(H_1) R_{\lambda+iy}(H_0)g}_\hilb}
    \\   & \qquad \leq \frac y \pi \norm{R_{\lambda+iy}(H_1)f}_\hilb \norm{R_{\lambda+iy}(H_0)g}_\hilb
    \\ & \qquad = \frac 1\pi \abs{\scal{f}{\Im R_{\lambda+iy}(H_1)f}_\hilb}^{1/2} \cdot \abs{\scal{g}{\Im R_{\lambda+iy}(H_0)g}_\hilb}^{1/2}.
 \end{split}
\end{equation}
In this inequality we take the limit $y\to 0^+$ to get, using
Proposition~\ref{P: delta(H)=euE*euE} and~(\ref{F: def of
diamond}),
$$
  \abs{\scal{f}{\mathfrak a_\pm(\lambda;H_1,H_0)g}_{1,-1}} \leq \norm{\euE_\lambda (H_1)f}_{\hlambda(H_r)} \cdot \norm{\euE_\lambda (H_0)g}_{\hlambda(H_0)}.
$$
It follows that the wave matrix is bounded with bound less or
equal to $1.$
\end{proof}
\noindent Henceforth we shall identify the form $w_\pm(\lambda;
H_1,H_0)$ with the corresponding operator
\begin{equation} \label{F: w(+-) as an op-r}
  w_\pm(\lambda; H_1,H_0) \colon \hlambda(H_0) \to \hlambda(H_1).
\end{equation}
It follows directly
from the definition of the wave matrix~(\ref{F: def of w +-}) and
Proposition~\ref{P: delta(H)=euE*euE} that
\begin{equation} \label{F: w(H0,H0)=1}
  w_\pm(\lambda; H_0,H_0) = 1_{\hlambda(H_0)}.
\end{equation}

\subsection{Multiplicative property of the wave matrix}
In this subsection we prove the multiplicative property of the
wave matrix. 
As is known (cf. e.g. \cite[2.7.3]{Ya}), proof of the multiplicative property is one of the key points of the stationary approach to scattering theory.
\begin{lemma} \label{L: Lemma Y}
Let $H_0,H_1,H_2$ be three self-adjoint operators on a rigged Hilbert space $(\hilb,F),$ which are compatible with the rigging~$F$ and
such that the operators $V_1 = H_1-H_0$ and $V_2=H_2-H_1$ admit decomposition (\ref{F: V=F^*JF}).
If $\lambda \in \Lambda(H_0,F) \cap \Lambda(H_1,F) \cap \Lambda(H_2,F),$ then the
equality
\begin{equation} \label{F: mathfrak a+=(1-RV)d(H)(1-VR)}
  \mathfrak a_\pm(\lambda; H_2,H_0) = (1 - R_{\lambda\mp i0}(H_2)V_2) \frac 1\pi \Im R_{\lambda+i0}(H_1)(1 + V_1 R_{\lambda\pm i0}(H_0))
\end{equation}
holds, where both sides are compact operators from $\clB(\hilb_1,\hilb_{-1}).$
\end{lemma}
\begin{proof} Let $y>0.$ We have
\begin{equation} \label{F: (1-R)V2 delta(H)...}
  \begin{split}
     (E) & := (1 - R_{\lambda \mp iy}(H_2)V_2) \frac 1\pi \Im R_{\lambda+iy}(H_1)(1 + V_1R_{\lambda\pm iy}(H_0))
      \\ & = \frac y\pi (1 - R_{\lambda\mp iy}(H_2)V_2) R_{\lambda\mp iy}(H_1)R_{\lambda\pm iy}(H_1)(1 + V_1R_{\lambda\pm iy}(H_0)).
  \end{split}
\end{equation}
Applying the second resolvent identity to this equality twice yields
\begin{equation} \label{F: (E)=R(bar z)R(z)}
  \begin{split}
     (E) = \frac y\pi R_{\lambda\mp iy}(H_2) R_{\lambda\pm iy}(H_0).
  \end{split}
\end{equation}
Since for $\lambda \in \Lambda(H_j,F)$ the limits $R_{\lambda\pm
i0}(H_j)$ of the resolvents exist in the norm topology of
$\clB(\hilb_1,\hilb_{-1}),$ it follows that for all $\lambda \in \Lambda(H_0,F) \cap \Lambda(H_1,F) \cap \Lambda(H_2,F)$
the limit of the left hand side of~(\ref{F: (1-R)V2 delta(H)...}) exists in norm as $y \to
0^+.$ By~(\ref{F: Ya (2.7.10)}) and~(\ref{D: of mathfrak a pm}),
the norm limit of~(\ref{F: (E)=R(bar z)R(z)}) as $y \to 0^+$
exists in the norm of the space $\clB(\hilb_1(F),\hilb_{-1}(F))$ and is equal to $\mathfrak a_\pm(\lambda; H_2,H_0).$
\end{proof}
Let $\lambda \in \Lambda(H,F)$ and let~$\Delta$ be any fixed open
bounded set containing~$\lambda.$ We can consider a compatible
pair $(E_\Delta H, F_\Delta)$ with Hilbert-Schmidt
rigging~$F_\Delta.$ By \cite[(2.20), 2.16(viii) and Lemma
3.1.6]{Az3v6}, for such riggings there exists a sequence of vectors
$b_j^\Delta(\lambda+i0) \in \hilb_1(\Delta),$ $j$ is an index of non-zero type (see \cite[\S 2.10]{Az3v6}), such
that the sequence of vectors
$$
  e_j^\Delta (\lambda+i0) := \euE_\lambda^\Delta(H) b_j^\Delta(\lambda+i0) \in \hlambda(\Delta), \ j \in \clZ_\lambda,
$$
is an orthonormal basis of the Hilbert space $\hlambda(\Delta),$ where $\clZ_\lambda$ is the set of indices of non-zero type.
Since, by Lemma~\ref{L: E(Delta) is a
contraction}, $\hilb_1(\Delta) = E_\Delta \hilb_1(F),$ there exist
vectors $b_j(\lambda+i0) \in \hilb_1(F),$ such that
$$
  b_j^\Delta(\lambda+i0) = E_\Delta b_j(\lambda+i0).
$$
Thus, according to (\ref{F: (euEl(f),euEl(g))=(euEl(D)(f),euEl(D)(g))}), the set of vectors
$$
  \set{b_j(\lambda+i0) \in \hilb_1(F) \colon j\in \clZ_\lambda}
$$
is such that the sequence
\begin{equation} \label{F: (bj) is o.n. basis}
  \brs{\euE_\lambda(H) b_j(\lambda+i0)}_{j\in \clZ_\lambda} \ \text{is an orthonormal basis of \ $\hlambda(H,F)$}
\end{equation}
\noindent (see (\ref{F: hlambda(H,F)}) for definition of the Hilbert space $\hlambda(H,F)$).
It is the only property of vectors $b_j(\lambda+i0)$ which is used
in the proof of the following lemma.

\begin{lemma} \label{L: Lemma X} Let~$H$ be a self-adjoint operator compatible with rigging~$F$
and let $\lambda \in \Lambda(H,F).$ For any $f,g \in \hilb_1(F)$
the equality
\begin{multline*}
  \frac 1\pi \scal{f}{\Im R_{\lambda+i0}(H) g}_{1,-1} \\ = \frac 1{\pi^2}
    \sum_{j \in \clZ_\lambda} \scal{f}{\Im R_{\lambda+i0}(H) b_j(\lambda+i0)}_{1,-1} \scal{\Im R_{\lambda+i0}(H) b_j(\lambda+i0)}{g}_{-1,1}
\end{multline*}
holds, where $\brs{b_j(\lambda+i0)}_{j\in \clZ_\lambda}$ is a sequence of
vectors from $\hilb_1(F),$ constructed above, with the property (\ref{F: (bj) is o.n. basis}).
\end{lemma}
\begin{proof} Using Proposition~\ref{P: delta(H)=euE*euE} and~(\ref{F: def of diamond}), we infer that the right hand side of the equality to be proved is equal to
\begin{equation*}
  \begin{split}
     (E) :&= \sum_{j\in \clZ_\lambda} \scal{f}{\euE_\lambda^\diamondsuit \euE_\lambda b_j(\lambda+i0)}_{1,-1} \scal{\euE_\lambda^\diamondsuit \euE_\lambda b_j(\lambda+i0)}{g}_{-1,1}
      \\ &  = \sum_{j\in \clZ_\lambda} \scal{\euE_\lambda f}{\euE_\lambda b_j(\lambda+i0)}_{\hlambda} \scal{\euE_\lambda b_j(\lambda+i0)}{\euE_\lambda g}_{\hlambda}.
  \end{split}
\end{equation*}
It follows from this and (\ref{F: (bj) is o.n. basis}) that
$$
  (E) = \scal{\euE_\lambda f}{\euE_\lambda g}_{\hlambda} = \scal{f}{\euE_\lambda^\diamondsuit \euE_\lambda g}_{1,-1}.
$$ Now, another application of Proposition~\ref{P: delta(H)=euE*euE} completes the proof.
\end{proof}

\begin{thm} \label{T: mult property of w(pm)} Let $H_0,H_1,H_2$ be three self-adjoint operators compatible with rigging~$F$
on~$\hilb,$ which satisfy conditions of Lemma~\ref{L: Lemma Y}. If $\lambda \in \Lambda(H_0,F) \cap \Lambda(H_1,F) \cap \Lambda(H_2,F),$ then
$$
  w_\pm(\lambda; H_2,H_0) = w_\pm(\lambda; H_2,H_1)w_\pm(\lambda; H_1,H_0).
$$
\end{thm}
\begin{proof} Let $f,g \in \hilb_1(F).$ Since the linear manifold $\euE_\lambda(H)\hilb_1(F)$ is dense in $\hlambda(H),$ it is enough to show that
\begin{equation*}
  \begin{split}
    (E) & := \scal{\euE_\lambda(H_2) f}{w_\pm(\lambda; H_2,H_1)w_\pm(\lambda; H_1,H_0)\euE_\lambda(H_0)g}
      \\ & = \scal{\euE_\lambda(H_2) f}{w_\pm(\lambda; H_2,H_0)\euE_\lambda(H_0) g}.
  \end{split}
\end{equation*}
Let vectors $b_j(\lambda+i0) \in \hilb_1(F), \ j \in \clZ_\lambda,$ be as in Lemma~\ref{L: Lemma X} for
the operator $H_1.$
It follows from definition of $(E)$ just given and (\ref{F: (bj) is o.n. basis}) that
\begin{equation*}
  \begin{split}
     (E) = \sum_{j\in \clZ_\lambda} & \scal{\euE_\lambda(H_2) f}{w_\pm(\lambda; H_2,H_1)\euE_\lambda(H_1) b_j(\lambda+i0)}
     \\ & \qquad \qquad \qquad \qquad \cdot \scal{\euE_\lambda(H_1) b_j(\lambda+i0)}{w_\pm(\lambda; H_1,H_0)\euE_\lambda(H_0)g}.
  \end{split}
\end{equation*}
Combining this equality with definition (\ref{F: def of w +-}) of $w_\pm(\lambda; H_1,H_0)$ gives
$$
  (E) = \sum_{j\in \clZ_\lambda} \scal{f}{\mathfrak a_\pm(\lambda; H_2,H_1) b_j(\lambda+i0)}_{1,-1} \scal{b_j(\lambda+i0)}{\mathfrak a_\pm(\lambda; H_1,H_0)\euE_\lambda(H_0)g}_{1,-1}.
$$
Combining this equality with formulas~(\ref{D: of mathfrak a pm}) for $\mathfrak a_\pm$ yields
\begin{multline*}
  (E) = \frac 1 {\pi^2} \sum_{j\in \clZ_\lambda} \scal{f}{\SqBrs{1 - R_{\lambda\mp i0}(H_2)V} \Im R_{\lambda+i0}(H_1) b_j(\lambda+i0)}_{1,-1}
    \\ \scal{b_j(\lambda+i0)}{\Im R_{\lambda+i0}(H_1) \SqBrs{1 + V R_{\lambda\pm i0}(H_0)}g}_{1,-1}.
\end{multline*}
By Lemma~\ref{L: Lemma X}, it follows that
$$
  (E) = \frac 1 \pi \scal{f}{\SqBrs{1 - R_{\lambda\mp i0}(H_2)V} \Im R_{\lambda+i0}(H_1) \SqBrs{1 + V R_{\lambda\pm i0}(H_0)}g}_{1,-1}.
$$
Now, Lemma~\ref{L: Lemma Y} implies that
$$
  (E) = \scal{f}{\mathfrak a_\pm(\lambda; H_2,H_0) g}_{1,-1} = \scal{\euE_\lambda(H_2) f}{w_\pm(\lambda; H_2,H_0)\euE_\lambda(H_0) g}.
$$
Proof is complete.
\end{proof}
\begin{cor} \label{C: w(pm) is unitary} For any two operators $H_0,H_1$ which satisfy conditions of Proposition~\ref{P: w(pm) is well-defined}
and for any $\lambda \in \Lambda(H_0,F)\cap \Lambda(H_1,F),$
the wave matrix (\ref{F: w(+-) as an op-r}) is a unitary operator. Moreover, $w_\pm^*(\lambda; H_1,H_0)=w_\pm(\lambda; H_0,H_1).$
\end{cor}
\noindent Proof of this corollary follows verbatim the proof of
\cite[Corollary 5.3.8]{Az3v6}, but one has to use~(\ref{F:
w(H0,H0)=1}) and Proposition~\ref{P: w(pm) is well-defined}
instead of \cite[Propositions 5.2.3 and 5.2.2]{Az3v6}.

\begin{rems} \rm The argument of this proof works in the case of Hilbert-Schmidt rigging~$F,$ and thus it simplifies
proof the multiplicative property of the wave matrix given in \cite{Az3v6}.
\end{rems}

\subsection{Wave operator}
Let $H_0, H_1$ be operators which satisfy conditions of Proposition~\ref{P: w(pm) is well-defined}.
Clearly, the operator-valued function
$\Lambda(H_0,F) \cap \Lambda(H_1,F) \ni \lambda \mapsto
w_\pm(\lambda; H_1,H_0)$ is measurable. Therefore, as in
\cite{Az3v6}, we can define the wave operator as an operator
$$
  W_\pm(H_1,H_0) \colon \euS(H_0,F) \to \euS(H_1,F)
$$
by the formula
\begin{equation} \label{F: def of W(pm)}
  W_\pm(H_1,H_0) = \int_{\Lambda(H_0,F) \cap \Lambda(H_1,F)}^\oplus w_\pm(\lambda; H_1,H_0)\,d\lambda.
\end{equation}
Since, by Theorem~\ref{T: op-r euE}, the Hilbert space
$\euS(H,F)$ is naturally isomorphic to the absolutely continuous
subspace $\hilb^{(a)}(H)$ of the self-adjoint operator~$H,$ it
follows that the wave operator $W_\pm(H_1,H_0),$ thus defined, can
also be considered as an operator
\begin{equation} \label{F: W(pm) def}
  W_\pm(H_1,H_0) \colon \hilb^{(a)}(H_0) \to \hilb^{(a)}(H_1).
\end{equation}
As an immediate consequence of the definition~(\ref{F: def of W(pm)}), Theorem~\ref{T: mult property of w(pm)} and Corollary~\ref{C: w(pm) is unitary}, we obtain the following
theorem.
\begin{thm} \label{T: properties of W(pm)}
Let $H_0, H_1, H_2$ be three self-adjoint operators which satisfy conditions of Lemma~\ref{L: Lemma Y}.
The wave operator~(\ref{F: W(pm) def}) defined by~(\ref{F: def of W(pm)}) have the following properties:
\begin{enumerate}
  \item[(i)] The operator (\ref{F: W(pm) def}) is unitary.
  \item[(ii)] $W_{\pm}(H_2,H_0) =  W_{\pm}(H_2,H_1)W_{\pm}(H_1,H_0).$
  \item[(iii)] $W^*_{\pm}(H_1,H_0) = W_{\pm}(H_0,H_1).$
  \item[(iv)] The operator $W_{\pm}(H_0,H_0)$ is the identity operator on $\hilb^{(a)}(H_0).$
\end{enumerate}
\end{thm}
\noindent If we define the operator (\ref{F: W(pm) def}) to be zero on the
singular subspace $\hilb^{(s)}(H_0),$ then the part~(iv) of this
theorem becomes $W_{\pm}(H_0,H_0) = P^{(a)}(H_0).$ It
follows from (ii) with $H_2 = H_0,$ (iii) and (iv) that
$$
  W_{\pm}(H_1,H_0) = W_{\pm}(H_1,H_0)P^{(a)}(H_0) = P^{(a)}(H_1)W_{\pm}(H_1,H_0).
$$
\begin{thm} For any bounded measurable function~$h$ on~$\mbR$ and any two operators self-adjoint operators $H_0,H_1$
which satisfy conditions of Proposition~\ref{P: w(pm) is well-defined}, we have
$$
  h(H_1) W_{\pm}(H_1,H_0) = W_{\pm}(H_1,H_0) h(H_0).
$$
Also,
$$
  H_1 W_{\pm}(H_1,H_0) = W_{\pm}(H_1,H_0) H_0.
$$
\end{thm}
\begin{proof} This immediately follows from the definition~(\ref{F: def of W(pm)}) of the wave operator, Theorem~\ref{T: op-r euE} and Corollary~\ref{C: euE is diagonalizing}.
\end{proof}
\noindent Hence, we have the following corollary (generalized
Kato-Rosenblum theorem).
\begin{cor} If $H_0, H_1$ are two self-adjoint operators which satisfy conditions of Proposition~\ref{P: w(pm) is well-defined},
then restrictions of operators~$H_0$ and~$H_1$ to their absolutely continuous subspaces $\hilb^{(a)}(H_0)$
and $\hilb^{(a)}(H_1)$ respectively are unitarily equivalent.
\end{cor}

A natural question is whether definition~(\ref{F: def of W(pm)}) coincides with
classical time-dependent definition of the wave matrix. The answer is positive:
inspection of the proof of \cite[Theorem 6.1.4]{Az3v6} (which itself is an adjustment of appropriate proofs from \cite{Ya}
to our setting) shows that the following theorem holds.
\begin{thm} \label{T: class-l def of W(pm)} If $H_0$ and $H_1$ are two self-adjoint operators which satisfy conditions of Proposition~\ref{P: w(pm) is well-defined},
then the strong operator limit
$$
  \lim_{t \to \pm \infty} e^{itH_1}e^{-itH_0}P_0^{(a)}
$$
exists and coincides with $\euE^*(H_1)W_\pm(H_1,H_0)\euE(H_0).$
\end{thm}
\noindent
Proof of this theorem is the same as that of \cite[Theorem 6.1.4]{Az3v6} with only one change:
\cite[Lemma 6.1.1]{Az3v6} should be replaced by the following lemma:
if $g \in \hilb$ is such that $\norm{\euE_\lambda(H_0) g}_\hlambda \leq N$ for a.e. $\lambda \in \Lambda(H_0,F)$
and $\euE_\lambda(H_0) g = 0$ for a.e. $\lambda \in \Lambda(H_0,F) \setminus \Delta,$ where~$\Delta$ is a bounded open set,
then
$$
  \int_{-\infty}^\infty \norm{F e^{-itH_0}P_0^{(a)} g}^2 \,dt \leq 2\pi N^2 \norm{FE_\Delta^{H_0}}_2^2.
$$
This lemma itself follows from \cite[Lemma 6.1.1]{Az3v6} applied to the self-adjoint operator $H_0E_\Delta^{H_0}$ acting on the Hilbert
space $E_\Delta^{H_0}\hilb$ with Hilbert-Schmidt rigging $F_\Delta = FE_\Delta^{H_0}.$
The set of vectors $g$ which satisfy the premise of this lemma is dense in $\hilb,$ and this is what is used in the proof of
\cite[Theorem 6.1.4]{Az3v6}.

\section{Scattering matrix}
\label{S: scat matrix}
As in \cite[Section 7]{Az3v6}, given a number $\lambda \in
\Lambda(H_0,F) \cap \Lambda(H_1,F),$ we define the scattering
matrix $S(\lambda; H_1,H_0)$ as an operator on the Hilbert space
$\hlambda(H_0)$ by the formula
\begin{equation} \label{F: def of S(lambda)}
  S(\lambda; H_1,H_0) = w_+^*(\lambda; H_1,H_0)w_-(\lambda; H_1,H_0).
\end{equation}
\noindent Just as in \cite[Section 7]{Az3v6}, we list here some
properties of the scattering matrix, which immediately follow from
this definition and the properties of the wave matrix already
established.
\begin{thm} \label{T: properties of S(lambda)}
Let $H_0,H_1,H_2$ be three self-adjoint operators compatible with rigging~$F$
on~$\hilb,$ which satisfy conditions of Lemma~\ref{L: Lemma Y}.
Let $\lambda \in \Lambda(H_0,F) \cap \Lambda(H_1,F) \cap \Lambda(H_2,F).$ The scattering matrix~(\ref{F: def of S(lambda)})
possesses the following properties.
  \begin{enumerate}
     \item[(i)] The scattering matrix~(\ref{F: def of S(lambda)}) is a unitary operator.
     \item[(ii)] $S(\lambda; H_{2},H_0) = w_+^*(\lambda;H_1,H_0)S(\lambda; H_{2},H_1)w_-(\lambda;H_1,H_0).$
     \item[(iii)] $S(\lambda; H_{2},H_0) = w_+^*(\lambda;H_1,H_0)S(\lambda; H_{2},H_1)w_+(\lambda;H_1,H_0) S(\lambda;H_1,H_0).$
  \end{enumerate}
\end{thm}
\noindent Another important property of the scattering matrix is the stationary formula.
\begin{thm} \label{T: stationary formula} Let~$H_0$ and~$H_1$ be two self-adjoint operators which satisfy conditions of Proposition~\ref{P: w(pm) is well-defined}.
If $\lambda \in \Lambda(H_0,F) \cap \Lambda(H_1,F),$ then
$$
  S(\lambda; H_1, H_0) = 1_\hlambda - 2\pi i \euE_\lambda(H_0) V (1 + R_{\lambda+i0}(H_0)V)^{-1}\euE_\lambda^\diamondsuit(H_0),
$$
where $\euE_\lambda^\diamondsuit(H_0)$ is defined by (\ref{F: def of diamond}).
\end{thm}
\noindent Recall that in this theorem~$V$ is understood as a bounded operator
acting from $\hilb_{-1}(F)$ to $\hilb_1(F)$ and
$R_{\lambda+i0}(H_0)$ is understood as a compact operator acting from
$\hilb_{1}(F)$ to $\hilb_{-1}(F).$ 
Proof of this stationary formula follows verbatim the proof of \cite[Theorem 7.2.2]{Az3v6}
and therefore is omitted. Further, as a by-product of the proof of this theorem, as in \cite{Az3v6}
one obtains the following formulas for the wave matrices
$$
  w_\pm(\lambda;H_1,H_0) \euE_\lambda(H_0) = \euE_\lambda(H_1)[1+VR_{\lambda\pm i0}(H_0)].
$$
\begin{cor} If $H_0,H_1$ are two self-adjoint operators which satisfy conditions of Proposition~\ref{P: w(pm) is well-defined},
and $\lambda \in \Lambda(H_0,F) \cap \Lambda(H_1,F),$ then $S(\lambda; H_1,H_0) \in 1 + \clL_1(\hlambda(H_0)).$
\end{cor}
\begin{proof} This follows from Theorem~\ref{T: stationary formula} and the fact that, by Proposition~\ref{P: euE is HS},
the evaluation operator $\euE_\lambda(H_0) \colon \hilb_1(F) \to
\hlambda(H_0)$ is Hilbert-Schmidt, and therefore so is the
operator $\euE_\lambda^\diamondsuit(H_0) \colon \hlambda(H_0) \to
\hilb_{-1}(F).$
\end{proof}

\subsection{Scattering operator}
We define the scattering operator as an operator
$$
  \mathbf S(H_1,H_0) \colon \euS(H_0,F) \to \euS(H_0,F),
$$ given by equality
\begin{equation} \label{F: def of bfS}
  \mathbf S(H_1,H_0) = \int_{\Lambda(H_0,F) \cap \Lambda(H_1,F)}^\oplus S(\lambda; H_1,H_0)\,d\lambda.
\end{equation}
The usual definition of the scattering operator
$$
  \mathbf S(H_1,H_0) = W_+^*(H_1,H_0)W_-(H_1,H_0)
$$
follows from definitions of the wave operators (\ref{F: def of W(pm)}) and the scattering matrix (\ref{F: def of S(lambda)}).
Obviously, for scattering operator we have analogues of properties
of the scattering matrix given in Theorem~\ref{T: properties of
S(lambda)}, similar to those in \cite[Theorem 7.1.3]{Az3v6}. In
particular, the scattering operator $\mathbf S(H_1,H_0)$ is
unitary and commutes with $H_0.$

\section{Example}
\label{S: example}
Recall that the class of potentials~$K_\nu$ on $\mbR^\nu$
\cite[p.\,453]{Si82BAMS} is defined as follows: a real-valued
measurable function~$V$ belongs to~$K_\nu$ if and only if
\\ (a) if $\nu\geq 3$
$$
  \lim_{\alpha \downarrow 0} \SqBrs{\sup_x \int_{\abs{x-y}\leq \alpha} \abs{x-y}^{-(\nu-2)}\abs{V(y)}\,d^\nu y} = 0,
$$
\\ (b) if $\nu = 2$
$$
  \lim_{\alpha \downarrow 0} \SqBrs{\sup_x \int_{\abs{x-y}\leq \alpha} \ln\brs{\abs{x-y}^{-1}} \abs{V(y)}\,d^2 y} = 0,
$$
\\ (c) if $\nu = 1$
$$
  \sup_x \int_{\abs{x-y}\leq 1} \abs{V(y)}\,dy < \infty.
$$
In particular, $L^\infty(\mbR^\nu,\mbR) \subset K_\nu.$
A potential~$V$ belongs to $K_\nu^{loc}$ if $V \chi_{R} \in
K_\nu,$ where $\chi_R$ is the characteristic function of the ball
$\set{x \colon \abs{x} \leq R}.$
The following theorem provides a large class of compatible pairs
$(H,F)$ among Schr\"odinger operators.
\begin{thm}\cite[Theorem B.9.1]{Si82BAMS}
Let $\nu$ be any positive integer. Let $H = -\Delta + V(x)$ be a
Schr\"odinger operator on $L^2(\mbR^\nu)$ with potential~$V$
satisfying $V_- \in K_\nu,$ $V_+ \in K_\nu^{loc}.$ Let~$f$ be a
bounded Borel function on the spectrum of~$H$ obeying $\abs{f(x)} \leq C (1+\abs{x})^{-\alpha}
$ for some $\alpha > \nu/4.$ If $g \in L^2(\mbR^\nu),$ then
$g(x)f(H)$ is Hilbert-Schmidt.
\end{thm}
\noindent It follows from this theorem that if~$F$ is multiplication
by $g(x) \in L^2(\mbR^\nu) \cap L^\infty(\mbR^\nu)$ where
$g(x) \neq 0$ for a.e. $x \in \mbR,$
then 
for any dimension $\nu$ the operator $F E^H_\Delta$ is Hilbert-Schmidt.
Hence, we have, in particular, a new proof of the following theorem as a corollary of part (i) of Theorem~\ref{T: properties of W(pm)}.
\begin{thm} Let $\nu$ be a positive integer and let~$H_0$ be self-adjoint extension of the differential operator $-\Delta + V_0$ on $\mbR^\nu,$
where $V_0$ is a potential satisfying $V_- \in K_\nu,$ $V_+ \in K_\nu^{loc}.$
For any self-adjoint operator~$H_1$ such that $H_1-H_0 \in L^\infty(\mbR^\nu) \cap L^1(\mbR^\nu),$ the wave operators $W_\pm(H_1,H_0)$
exist and complete.
\end{thm}
Other results of Sections~\ref{S: wave matrix} and~\ref{S: scat matrix} also apply to this class of Schr\"odinger operators.


\rndef{\emph}[1]{{\it #1}}

\mathsurround 0pt
\ndef{\AndSoOn}{$\dots$}

\end{document}
%
%
%
%
%
%
%